\documentclass[11pt, a4paper, reqno]{article}
\usepackage{amsmath,amsthm,amsfonts,amssymb}
\usepackage[english]{babel}
\usepackage{a4wide}
\usepackage{microtype}
\usepackage{bm}
\usepackage{graphics}
\usepackage{epic}
\usepackage{hyperref}
\hypersetup{colorlinks=true,citecolor=blue,linkcolor=blue,urlcolor=blue}
\usepackage{url}
\usepackage{mathtools}
\usepackage[shortlabels]{enumitem}
\usepackage{tikz}

\newtheorem{theorem}{Theorem}
\newtheorem*{theorem*}{Theorem}
\newtheorem{lemma}[theorem]{Lemma}
\newtheorem{corollary}[theorem]{Corollary}
\newtheorem{proposition}[theorem]{Proposition}
\newtheorem*{proposition*}{Proposition}
\newtheorem{claim}{Claim}

\theoremstyle{definition}
\newtheorem{remark}[theorem]{Remark}
\newtheorem{example}[theorem]{Example}
\newtheorem{definition}[theorem]{Definition}

\def\structA{\mathbb{A}}
\def\structB{\mathbb{B}}
\def\structC{\mathbb{C}}

\def\C{\mathcal{C}}
\def\G{\mathcal{G}}
\def\H{\mathcal{H}}
\def\S{\mathcal{S}}
\def\B{\mathcal{B}}

\def\X{\mathcal{X}}
\def\I{\mathcal{I}}
\def\T{\mathcal{T}}

\def\M{\mathcal{M}}

\def\qplus{\mathbb{Q}_{\geq 0}}
\def\qinf{\overline{\mathbb{Q}}_{\geq 0}}

\def\btuples{\text{BTup}}
\def\tuples{\text{BTup}}

\def\equi{\equiv}

\def\lessfrac{\preceq^{\text{SA}}_k}

\def\id{\text{id}}

\newcommand{\problem}[1]{{\sc #1}}

\newcommand{\cost}[1]{\text{cost}(#1)}
\newcommand{\costb}[2]{\text{cost}_{#1}(#2)}
\newcommand{\opt}[1]{\text{opt}(#1)}
\newcommand{\optfrac}[2]{\text{opt}^\text{SA}_{#1}(#2)}
\newcommand{\tup}[1]{\text{tup}(#1)}
\newcommand{\tuple}[1]{\mathbf{#1}}
\newcommand{\rel}[1]{\text{rel}(#1)}
\newcommand{\pos}[1]{\text{Pos}(#1)}
\newcommand{\ar}[1]{\text{ar}(#1)}
\newcommand{\supp}[1]{\text{supp}(#1)}

\newcommand{\toset}[1]{\text{Set}(#1)}
\newcommand{\dom}[1]{\text{dom}(#1)}

\newcommand{\rels}[1]{\mathbf{#1}}

\newcommand{\tw}[1]{\text{$tw$}(#1)}
\newcommand{\twms}[1]{\text{$tw_{ms}$}(#1)}

\newcommand{\defeq}{\vcentcolon=}

\newcommand{\tupsetsize}[1]{\left\lVert#1\right\rVert}

\def\less{\preceq} 
\def\vecpred{\leq} 

\renewcommand{\mapsto}{\to}

\begin{document}

\title{The complexity of general-valued CSPs seen from the other
side\thanks{An extended abstract of this work appeared in the
\emph{Proceedings of the 59th Annual IEEE Symposium on Foundations of Computer
Science (FOCS'18)}~\cite{crz18:focs}. Stanislav \v{Z}ivn\'y was supported by a Royal Society University
Research Fellowship. This project has received funding from the European
Research Council (ERC) under the European Union's Horizon 2020 research and
innovation programme (grant agreement No 714532). The paper reflects only the
authors' views and not the views of the ERC or the European Commission. The
European Union is not liable for any use that may be made of the information
contained therein. Miguel Romero was supported by Fondecyt grant 11200956. 
Work done while Cl\'ement Carbonnel and Miguel Romero were at the University of Oxford.}}

\author{
Cl\'ement Carbonnel\\
CNRS\\
University of Montpellier\\
France\\
\texttt{clement.carbonnel@lirmm.fr}
\and
Miguel Romero\\
Faculty of Engineering and Science\\
Universidad Adolfo Ib\'a\~nez\\
Santiago, Chile\\
\texttt{miguel.romero.o@uai.cl}
\and
Stanislav \v{Z}ivn\'{y}\\
University of Oxford, UK\\
\texttt{standa.zivny@cs.ox.ac.uk}
}

\date{}
\maketitle

\begin{abstract}

The constraint satisfaction problem (CSP) is concerned with homomorphisms
between two structures. For CSPs with restricted left-hand side structures,
the results of Dalmau, Kolaitis, and Vardi [CP'02], Grohe [FOCS'03/JACM'07],
and Atserias, Bulatov, and Dalmau [ICALP'07] establish the precise borderline
of polynomial-time solvability (subject to complexity-theoretic assumptions)
and of solvability by bounded-consistency algorithms (unconditionally) as
bounded treewidth modulo homomorphic equivalence.

The general-valued constraint satisfaction problem (VCSP) is a generalisation of
the CSP concerned with homomorphisms between two \emph{valued} structures. For
VCSPs with restricted left-hand side valued structures, we establish the
precise borderline of polynomial-time solvability (subject to
complexity-theoretic assumptions) and of solvability by the $k$-th level of
the Sherali-Adams LP hierarchy (unconditionally). We also obtain results on
related problems concerned with finding a solution and recognising the
tractable cases; the latter has an application in database theory.

\end{abstract}

\section{Introduction}
\label{sec:intro}

\subsection{Constraint Satisfaction Problems}
The homomorphism problem for relational structures is a fundamental computer
science problem: Given two relational structures $\rels{A}$ and $\rels{B}$ over
the same signature, the goal is to determine the existence of a homomorphism
from $\rels{A}$ to $\rels{B}$ (see, e.g., the book by Hell and Ne\v{s}et\v{r}il
on this topic~\cite{Hell:graphs}). The homomorphism problem is known to be
equivalent to the evaluation problem and the containment problem for conjunctive
database queries~\cite{Chandra77:stoc,Kolaitis98:pods}, and also to the
constraint satisfaction problem (CSP)~\cite{Feder98:monotone}, which
originated in artificial intelligence~\cite{Montanari74:constraints} and
provides a common framework for expressing a wide range of both theoretical and
real-life combinatorial problems.

For a class $\C$ of relational structures, we denote by \problem{CSP}($\C$, $-$)
the restriction of the homomorphism problem in which the input structure
$\rels{A}$ belongs to $\C$ and the input structure $\rels{B}$ is arbitrary
(these types of restrictions are known as \emph{structural} restrictions).
Similarly, by \problem{CSP}($-$, $\C$) we denote the restriction of the
homomorphism problem in which the input structure $\rels{A}$ is arbitrary and
the input structure $\rels{B}$ belongs to $\C$.

Feder and Vardi initiated the study of \problem{CSP}($-$, $\{\rels{B}\}$), also
known as non-uniform CSPs, and famously conjectured that, for every fixed finite
structure $\rels{B}$, \problem{CSP}($-$, $\{\rels{B}\}$) is in PTIME or
\problem{CSP}($-$, $\{\rels{B}\}$) is NP-complete. For example, if $\rels{B}$ is
a clique on $k$ vertices then \problem{CSP}($-$, $\{\rels{B}\}$) is the
well-known $k$-colouring problem, which is known to be in PTIME for $k\leq 2$
and NP-complete for $k\geq 3$. Most of the progress on the Feder-Vardi
conjecture (e.g.,
\cite{Bulatov06:3-elementJACM,Barto09:siam,Idziak10:siam,Bulatov11:conservative,Barto14:jacm})
is based on the algebraic approach~\cite{Bulatov05:classifying}, culminating in
two (affirmative) solutions to the Feder-Vardi conjecture obtained
independently by Bulatov~\cite{Bulatov17:focs} and Zhuk~\cite{Zhuk20:jacm}.

\problem{CSP}($\C$, $-$) is only interesting if $\C$ is an infinite
class of structures as otherwise \problem{CSP}($\C$, $-$) is always in PTIME.
(This is, however, not the case for \problem{CSP}($-$, $\C$) as we have seen in
the example of $3$-colouring.) Freuder observed that \problem{CSP}($\C$, $-$) is
in PTIME if $\C$ consists of trees~\cite{Freuder82:backtrack-free} or, more
generally, if it has bounded treewidth~\cite{Freuder90:complexity}. Later,
Dalmau, Kolaitis, and Vardi showed that \problem{CSP}($\C$, $-$) is solved by
$k$-consistency, a fundamental local propagation
algorithm~\cite{Dechter03:processing}, if $\C$ is of bounded treewidth
\emph{modulo homomorphic equivalence}, i.e., if the treewidth of the cores of
the structures from $\C$ is at most $k$, for some fixed $k\geq
1$~\cite{Dalmau02:width}. Atserias, Bulatov, and Dalmau showed that this is
precisely the class of structures solved by
$k$-consistency~\cite{Atserias07:power}. Grohe proved~\cite{Grohe07:jacm}
that the tractability result of Dalmau et al.~\cite{Dalmau02:width} is optimal
for classes $\C$ of bounded arity: Under the assumption that FPT $\ne $ W[1],
\problem{CSP}($\C$, $-$) is tractable if and only if $\C$ has bounded treewidth
modulo homomorphic equivalence. 

\subsection{General-valued Constraint Satisfaction Problems} General-valued
Constraint Satisfaction Problems (VCSPs) are generalisations of CSPs which allow
for not only decision problems but also for optimisation problems (and the mix
of the two) to be considered in one framework~\cite{Cohen06:complexitysoft}. In
the case of VCSPs we deal with \emph{valued} structures. Regarding tractable
restrictions, the situation of the non-uniform case is by now well-understood.
Indeed, it holds that for any
fixed valued structure $\structB$,  either \problem{VCSP}($-$, $\{\structB\}$)
is in PTIME or \problem{VCSP}($-$, $\{\structB\}$) is
NP-complete~\cite{Kozik15:icalp,Kolmogorov17:sicomp}. 

For structural restrictions, it is a folklore result that \problem{VCSP}($\C$,
$-$) is tractable if $\C$ is of bounded treewidth; see,
e.g.~\cite{Bertele72:nonserial}. So is the belief that the $(k+1)$-st level of the
Sherali-Adams LP hierarchy~\cite{Sherali1990} solves \problem{VCSP}($\C$, $-$)
to optimality if the treewidth of all structures in $\C$ is at most $k$. (We are
not aware of any reference for this. For certain special problems, it is
discussed in~\cite{Bienstock04:do}. For the extension complexity of such
problems, see~\cite{Kolman15:ejc}.) However, unlike the CSP case, the precise
borderline of polynomial-time solvability and the power of fundamental
algorithms (such as the Sherali-Adams LP hierarchy) for VCSP($\C$, $-$) is still
unknown.  Understanding these complexity and algorithmic frontiers for
VCSP($\C$, $-$) is the main goal of this paper. 

\subsection{Contributions}

We study the problem \problem{VCSP}($\C$, $-$) for classes $\C$ of valued
structures over finite signatures and give three main results.

\paragraph{(1) Complexity classification}

As our first result, we give (in Theorem~\ref{theo:main}) a complete complexity
classification of \problem{VCSP}($\C$, $-$) and identify the precise borderline
of tractability, for classes $\C$ of bounded arity. A key ingredient in our
result is a novel notion of \emph{valued equivalence} and a characterisation of
this notion in terms of \emph{valued cores}. More precisely, we show that
\problem{VCSP}($\C$, $-$) is tractable if and only if $\C$ has bounded treewidth
modulo valued equivalence. This latter notion \emph{strictly} generalises
bounded treewidth and it is \emph{strictly} weaker than bounded treewidth modulo
homomorphic equivalence. Our proof builds on the characterisation by Dalmau et
al.~\cite{Dalmau02:width} and Grohe~\cite{Grohe07:jacm} for CSPs.
We show that the newly identified tractable classes are solvable by the
Sherali-Adams LP hierarchy. 

\paragraph{(2) Power of Sherali-Adams}

Our second result (Theorem~\ref{thm:sa}) gives a precise characterisation of the
power of Sherali-Adams for \problem{VCSP}($\C$, $-$). In particular, we show
that the $(k+1)$-st level of the Sherali-Adams LP hierarchy solves
\problem{VCSP}($\C$, $-$) to optimality if and only if the valued cores of the
structures from $\C$ have treewidth \emph{modulo scopes} at most $k$ and the
\emph{overlaps} of scopes are of size at most $k+1$. The proof builds on the
work of Atserias et al.~\cite{Atserias07:power} and Thapper and
\v{Z}ivn\'y~\cite{tz17:sicomp}, as well as on an adaptation of the classical
result connecting treewidth and brambles by Seymour and
Thomas~\cite{Seymour93:graph}.

\paragraph{(3) Search VCSP}

Our first two results are for the \problem{VCSP} in which we ask for the cost of
an optimal solution. It is also possible to define the \problem{VCSP} as a
\emph{search} problem, in which one is additionally required to return a
solution with the optimal cost. A complete characterisation of tractable search
cases in terms of structural properties of (a class of structures) $\C$ is open
even for CSPs and there is some evidence that the tractability frontier cannot
be captured in simple terms.\footnote{In particular, \cite[Lemma~1]{BDGM12:enum}
shows that a description of tractable cases of \problem{SCSP($\C$, $-$)}, which
is the search variant of \problem{CSP($\C$, $-$)} defined in
Section~\ref{sec:search}, would
imply a description of tractable cases of \problem{CSP($-$, $\{\rels{B}\}$)}.}
Building on our first
two results as well as on techniques from~\cite{tz16:jacm}, we give in
Section~\ref{sec:search} a characterisation of the tractable cases for search
\problem{VCSP($\C$, $-$)} in terms of tractable core computation
(Theorem~\ref{thm:search}). 

\paragraph{(4) Additional results}

In addition to our main results, we provide in Section~\ref{sec:meta} tight
complexity bounds for several problems related to our classification results,
e.g., deciding whether the treewidth is at most $k$ modulo valued equivalence,
deciding solvability by the $k$-th level of the Sherali-Adams LP hierarchy, and
deciding valued equivalence for valued structures.  These results have
interesting consequences to database theory. Specifically, to the evaluation and
optimisation of conjunctive queries over \emph{annotated} databases. In
particular, we show that the containment problem of conjunctive queries over the
tropical semiring is in NP, thus improving on the work of~\cite{Kostylev14:tds},
which put it in $\Pi^p_2$.

\subsection{Related work} 

In his PhD thesis~\cite{Farnqvist13:phd}, F\"arnqvist
studied the complexity of VCSP($\C$, $-$) and some other fragments of VCSPs (see
also~\cite{Farnqvist07:isaac,Farnqvist12:cpaior}). He considered a very specific
framework that only allows for particular types of classes $\C$'s to be
classified. For these classes, he showed that only bounded treewidth gives rise
to tractability (assuming bounded arity) and asked about more general classes.
In particular, decision CSPs do \emph{not} fit in his framework and Grohe's
classification~\cite{Grohe07:jacm} is not implied by F\"arnqvist's work. In
contrast, our characterisation (of \emph{all} classes $\C$'s of valued
structures) gives rise to new tractable cases going beyond those identified by
F\"arnqvist. Moreover, we can derive both Grohe's classification and
F\"arnqvist's classification directly from our results, as explained in Section~\ref{sec:compl}.

It is known that Grohe's characterisation applies only to classes $\C$ of
\emph{bounded arity}, i.e., when the arities of the signatures are always
bounded by a constant (for instance, CSPs over digraphs) and fails for classes
of unbounded arity.
In this direction, several hypergraph-based restrictions that lead to
tractability have been proposed (for a survey see, e.g.~\cite{GGLS16:pods}).
Nevertheless, the precise tractability frontier for CSP($\C$, $-$) is not known.
The situation is different for \emph{fixed-parameter tractability}: Marx gave a
complete classification of the fixed-parameter tractable restrictions CSP($\C$,
$-$), for classes $\C$ of structures of unbounded arity~\cite{Marx13:jacm}. 
Gottlob et al.~\cite{Gottlob09:icalp} and F\"arnqvist~\cite{Farnqvist12:cpaior} applied well-known hypergraph-based
tractable restrictions of CSPs to VCSPs. 

\section{Preliminaries}
\label{sec:prelims}

We assume familiarity with relational structures and homomorphisms. 
Briefly, a \emph{relational signature} is a finite set $\tau$ of relation
symbols $R$, each with a specified arity $\ar{R}$. A \emph{relational structure}
${\bf A}$ over a relational signature $\tau$ (or a relational $\tau$-structure,
for short) is a finite non-empty universe $A$ together with one relation $R^{\bf
A}\subseteq A^{\ar{R}}$ for each symbol $R \in \tau$. A \emph{homomorphism} from
a relational $\tau$-structure ${\bf A}$ (with universe $A$) to a relational $\tau$-structure ${\bf B}$
(with universe $B$) is a mapping $h:A \mapsto B$ such that for all $R\in\tau$
and all tuples $\tuple{x}\in R^{\bf A}$ we have $h(\tuple{x})\in R^{\bf B}$. We
refer the reader to~\cite{Hell:graphs} for more details.

We use $\qinf$ to denote the set of nonnegative rational numbers with positive infinity, i.e. $\qinf = \qplus \cup \{\infty\}$. 
As usual, we assume that $\infty+c=c+\infty=\infty$ for all $c\in \qinf$, $\infty \times 0 = 0 \times \infty = 0$, $\infty \times c=c\times \infty=\infty$ for all $c>0$, and $c/\infty=0$ for all $c\in \qplus$.

\paragraph{Valued structures}

A \emph{signature} is a finite set $\sigma$ of function symbols $f$, each with a specified arity $\ar{f}$. A \emph{valued structure} $\structA$ over a signature $\sigma$ (or valued $\sigma$-structure, for short) is a finite non-empty universe $A$ together with one function $f^{\structA}: A^{\ar{f}} \mapsto \qinf$ for each symbol $f \in \sigma$. 
We define $\tup{\structA}$ to be the set of all pairs $(f,\tuple{x})$ such that $f \in \sigma$ and $\tuple{x} \in A^{\ar{f}}$. The set of \emph{positive} tuples of $\structA$ is defined by $\tup{\structA}_{>0} \defeq \{(f,\tuple{x})\in \tup{\structA}\mid f^{\structA}(\tuple{x})>0\}$. If $\structA,\structB,\dots$ are valued structures, then $A,B,\dots$ denote their respective universes.

\paragraph{VCSPs}

We define \emph{Valued Constraint Satisfaction Problems} (VCSPs) as in~\cite{tz12:focs}. 
An instance of the \problem{VCSP} is given by two valued structures $\structA$ and $\structB$ over the same signature $\sigma$. 
For a mapping $h:A \mapsto B$, we define 
\[
  \cost{h} = \sum_{(f,\tuple{x})\in
  \tup{\structA}}f^{\structA}(\tuple{x})f^{\structB}(h(\tuple{x})).
\]
If needed, we will write $\costb{\structA\mapsto\structB}{h}$ for $\cost{h}$ to
emphasise the source and target structures.
Given a VCSP instance, the goal is to find the minimum cost over all possible
mappings $h:A\mapsto B$.  We denote this cost by $\opt{\structA,\structB}$. 

For a class $\C$ of valued structures (not necessarily over the same signature),
we denote by \problem{VCSP}($\C$, $-$) the class of VCSP
instances $(\structA,\structB)$ such that $\structA\in \C$. We say that
\problem{VCSP}($\C$, $-$) is in PTIME, the class of problems solvable in
\emph{polynomial time}, if there is a deterministic algorithm that solves any
instance $(\structA, \structB)$ of \problem{VCSP}($\C$, $-$) in time
$(|\structA|+|\structB|)^{O(1)}$.
We also consider the parameterised version of \problem{VCSP}($\C$, $-$), denoted
by $p$-\problem{VCSP}($\C$, $-$), where the parameter is $|\structA|$.
We say that $p$-\problem{VCSP}($\C$, $-$) is in FPT, the class of problems that
are \emph{fixed-parameter tractable}, if there is a deterministic algorithm that
solves any instance $(\structA,\structB)$ of $p$-\problem{VCSP}($\C$, $-$) in
time $f(|\structA|)\cdot|\structB|^{O(1)}$, where $f:\mathbb{N} \mapsto
\mathbb{N}$ is an arbitrary computable
function. The class W[1], introduced in~\cite{Downey95:sicomp}, can be seen as an analogue of NP in parameterised
complexity theory. Proving W[1]-hardness of $p$-\problem{VCSP}($\C$, $-$) (under
an fpt-reduction, formally defined in Section~\ref{sec:hard-main}) is a strong indication that $p$-\problem{VCSP}($\C$, $-$) is
not in FPT as it is believed that FPT $\neq$ W[1]. We refer the reader
to~\cite{FG06} for more details on parameterised complexity.

\paragraph{Treewidth of a valued structure} 
 
The notion of treewidth (originally introduced by Bertel\'e and Brioschi~\cite{Bertele72:nonserial} and later rediscovered
by Robertson and Seymour~\cite{Robertson84:minors3}) is a well-known measure of the tree-likeness of a graph \cite{diestel10:graph}.  
Let $G=(V(G),E(G))$ be a graph. A \emph{tree decomposition} of $G$ 
is a pair $(T,\beta)$ where $T=(V(T),E(T))$ is a tree and $\beta$ is a function that maps 
each node $t\in V(T)$ to a subset of $V(G)$ such that
\begin{enumerate}
\item $V(G)=\bigcup_{t\in V(T)} \beta(t)$, 
\item for every $u\in V(G)$, the set 
$\{t\in V(T)\mid u\in \beta(t)\}$ induces a connected subgraph of $T$, and 
\item for every edge $\{u,v\}\in E(G)$, there is a node $t\in V(T)$ with $\{u,v\}\subseteq \beta(t)$. 
\end{enumerate} 
The \emph{width} of the decomposition $(T,\beta)$ is $\max\{|\beta(t)|\mid t\in V(T)\}-1$. 
The \emph{treewidth} $\tw{G}$ of a graph $G$ is the minimum width over all its tree decompositions. 

Let ${\bf A}$ be a relational structure over relational signature $\tau$. Its
\emph{Gaifman graph} (also known as its \emph{primal graph}), 
denoted by $G(\rels{A})$, is the graph whose vertex set is the universe of ${\bf A}$ and whose edges are the
pairs $\{u,v\}$ for which there is a tuple $\tuple{x}$ and a relation symbol $R\in \tau$ 
such that $u,v$ appear in $\tuple{x}$ and $\tuple{x}\in R^{\rels{A}}$. 
We define the treewidth of $\rels{A}$ to be $\tw{\rels{A}}=\tw{G(\rels{A})}$. 

Let $\structA$ be a valued $\sigma$-structure.  
 If $\structA$ is the left-hand side of an instance of the \problem{VCSP}, the only tuples relevant to the problem are those in $\tup{\structA}_{>0}$. 
 Hence, in order to define structural restrictions and in particular, the notion of treewidth, we focus on the structure induced by $\tup{\structA}_{>0}$. 
 Formally, we associate with the signature $\sigma$ a relational signature $\rel{\sigma}$ that contains, for
 every $f\in \sigma$, a relation symbol $R_f$ of the same arity as $f$. 
We define $\pos{\structA}$ to be the relational structure over $\rel{\sigma}$ with the same universe $A$ of $\structA$ such that $\tuple{x}\in R_f^{\pos{\structA}}$ if and only if $f^{\structA}(\tuple{x})>0$.
We let the treewidth of $\structA$ be $\tw{\structA}=\tw{\pos{\structA}}$. 

\begin{remark}
Observe that, in the \problem{VCSP}, we allow infinite costs not only in
$\structB$ but also in the left-hand side structure $\structA$. This allows us
to consider the \problem{VCSP} as the minimum-cost mapping problem between two
mathematical objects of the same nature. Intuitively, mapping the tuples of
$\structA$ to infinity ensures that those are logically equivalent to
\textit{hard constraints}, as any minimum-cost solution of finite cost must map them to
tuples of cost exactly $0$ in $\structB$. Thus, decision CSPs, which are
$\{0,\infty\}$-valued VCSPs, are a special case of our definition and all our
results also apply to CSPs.
\end{remark}

\section{Equivalence for valued structures}
\label{sec:equiv}

We start by introducing the notion of valued equivalence that is crucial for our results. 

\begin{definition}
\label{def:improve}
Let $\structA, \structB$ be valued $\sigma$-structures. We say that $\structA$ \emph{improves} $\structB$, denoted by $\structA \less \structB$, if 
$\opt{\structA,\structC} \leq \opt{\structB,\structC}$ for all valued $\sigma$-structures $\structC$. 
\end{definition}

When two valued structures improve each other, we call them equivalent. In
Section~\ref{sec:intro}, we used the term ``valued equivalence''. In the rest of
the paper, we drop the word ``valued'' unless needed for clarity.

\begin{definition}
\label{def:equiv}
Let $\structA, \structB$ be valued $\sigma$-structures. We say that $\structA$ and $\structB$ are \emph{equivalent}, denoted by $\structA\equi\structB$, if $\structA\less\structB$ and $\structB\less\structA$. 
\end{definition}

Hence, two valued $\sigma$-structures $\structA$ and $\structB$ are equivalent
if they have the same optimal cost over all right-hand side valued structures. 
Observe that equivalence implies \emph{homomorphic} equivalence of
$\pos{\structA}$ and $\pos{\structB}$. 
Indeed, 
whenever $\pos{\structA}$ is not homomorphic to $\pos{\structB}$, 
we can define a valued $\sigma$-structure $\structC$ as follows: $\structC$ and $\structB$ have the same universe $B$, 
and for every $f\in \sigma$ and $\tuple{x}\in B^{\ar{f}}$, $f^{\structC}(\tuple{x})=0$ if $(f,\tuple{x})\in \tup{\structB}_{>0}$, and $f^{\structC}(\tuple{x})=\infty$ otherwise. 
Since the identity mapping has cost $0$, we have $\opt{\structB,\structC}=0$. On
the other hand, $\opt{\structA,\structC}=\infty$ 
as any finite-cost map from $\structA$ to $\structC$ is a homomorphism from
$\pos{\structA}$ to $\pos{\structB}$.
Consequently, $\structA\not\equiv\structB$. 
As the following example shows, the converse does not hold in general. 

\begin{figure}
\centering
\begin{tikzpicture}
  \pgftext{\includegraphics[scale=0.9]{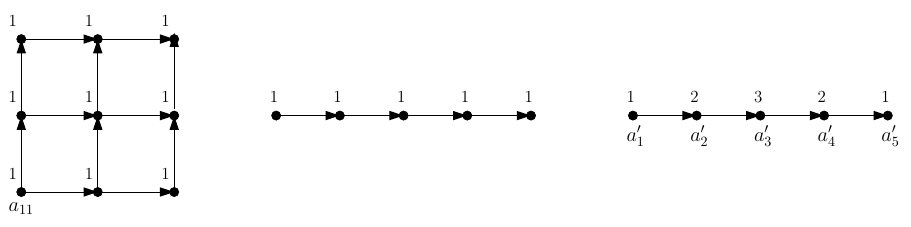}} at (0,0);
  \node at (-5.4,-2) {$\structA$};
  \node at (-0.8,-2) {$\structB$};
  \node at (4.8,-2) {$\structA'$};
\end{tikzpicture}
\caption{The valued structures from Examples~\ref{ex1} and~\ref{ex2}.}
\label{fig:ex1}
\end{figure}

\begin{example}
\label{ex1}
Consider the valued $\sigma$-structures $\structA$ and $\structB$ from Figure~\ref{fig:ex1}, 
with $\sigma=\{f,\mu\}$, where $f$ and $\mu$ are binary and unary function symbols, respectively.  
In Figure~\ref{fig:ex1}, $\mu$ is represented by the numbers labelling the nodes, and $f$ is represented as follows: 
pairs receiving cost $\infty$ are depicted as edges, while all remaining pairs are mapped to $0$. 
Observe that $\pos{\structA}$ and $\pos{\structB}$ are homomorphically equivalent. 
However, they are not (valued) equivalent. Indeed, 
consider the valued $\sigma$-structure $\structC$ with same universe $B$ as $\structB$ such that (i) for every $\tuple{x}\in B^{2}$,  
$f^{\structC}(\tuple{x})=0$ if $f^{\structB}(\tuple{x})=\infty$, otherwise $f^{\structC}(\tuple{x})=\infty$, and (ii) $\mu^{\structC}=\mu^{\structB}$. 
It follows that $\opt{\structA,\structC}=9$ and $\opt{\structB,\structC}=5$, and
thus $\structA\not\equiv \structB$. 
\end{example}

In the rest of the section, we give characterisations of equivalence in terms of certain types of homomorphisms and (valued) cores. 
We conclude with a useful characterisation of bounded treewidth modulo equivalence in terms of cores. 
In order to keep the flow uninterrupted, we defer (most of) the
proofs from this section to Appendix~\ref{app:equiv}.

\subsection{Inverse fractional homomorphisms}
\label{sec:inv-frac}

A homomorphism between two relational structures is a structure-preserving
mapping; e.g., for graphs, this means preserving adjacency -- edges are mapped
to edges only. An inverse homomorphism between two relational structures is a
mapping with the property that the preimage of any tuple in a relation in the
target structure can only contain tuples in the corresponding relation in the
source structure; e.g., for graphs, only edges (and not non-edges) from the
source graph can be mapped to the edges of the target graph.

A \emph{fractional} homomorphism between two valued structures has played
an important role in the work of Thapper and \v{Z}ivn\'y on the power of the
basic LP relaxation for VCSPs with a fixed valued constraint language~\cite{tz12:focs}. 
Intuitively, it is a probability distribution over mappings between the universes of the two structures with the property that the expected cost is not increased~\cite{tz12:focs}. 
In this paper, we will use a different but related notion of \emph{inverse} fractional homomorphism.

First, we need to define two elementary operations on valued structures (summation and scaling), as well as pointwise comparison. Let $\structA_1, \ldots, \structA_k$ be valued $\sigma$-structures. The valued $\sigma$-structure $\structC = \sum_{i} \structA_i$ has universe $C = \cup_{i} A_i$, and for each $f\in\sigma$ and $\tuple{x}\in C^{\ar{f}}$, $f^{\structC}(\tuple{x}) = \sum_{i : \tuple{x} \in A_i^{\ar{f}}} f^{\structA_i}(\tuple{x})$. Similarly, for $c\in\qplus$, the structure $c\cdot\structA_1$ is obtained from $\structA_1$ by scaling its functions by $c$, i.e. $f^{c\cdot\structA_1}(\tuple{x}) = c\cdot f^{\structA_1}(\tuple{x})$. Finally, if $\structA$ and $\structB$ are valued $\sigma$-structures with $A = B$, then we define the partial order $\structA \leq \structB$ pointwise, i.e. $\structA \leq \structB$ if and only if $f^{\structA}(\tuple{x}) \leq f^{\structB}(\tuple{x})$ for every symbol $f$ and tuple $\tuple{x}$. 
 
For sets $A$ and $B$, we denote by $B^A$ the set of all mappings $g:A\mapsto B$
from $A$ to $B$. For a probability distribution $\omega$ on $B^A$, the
\emph{support} of $\omega$ is the set $\supp{\omega} \defeq \{g\in B^A \mid
\omega(g)>0\}$. The following notion of an ``image valued structure'' will be useful.

\begin{definition}
\label{def:ga-new}
Let $\structA$ be a valued $\sigma$-structure and $g:A \mapsto B$ be a mapping.
We define $g(\structA)$ to be the valued $\sigma$-structure over universe
$g(A)$ such that, for all $f\in \sigma$ and $\tuple{y}\in g(A)^{\ar{f}}$,
  \[f^{g(\structA)}(\tuple{y})=f^{\structA}(g^{-1}(\tuple{y}))\defeq\sum_{\tuple{x}\in A^{\ar{f}}:
  g(\tuple{x})=\tuple{y}} f^{\structA}(\tuple{x}).\] If $\omega$ is a probability
distribution over $B^A$ we put $\omega(\structA)=\sum_{g\in
B^A}\omega(g)g(\structA)$. 
\end{definition}

Observe that, in the definition above, the universe of $\omega(\structA)$ is always $B$. Consequently, if $\structB$ is a valued $\sigma$-structure over $B$, then pointwise comparison is possible between $\omega(\structA)$ and $\structB$.

\begin{definition}\label{def:ifh}
Let $\structA, \structB$ be valued $\sigma$-structures. A probability
  distribution $\omega$ over $B^A$ is an \emph{inverse fractional homomorphism}
  (IFH) from $\structA$ to $\structB$ if $\omega(\structA)\vecpred\structB$.
\end{definition}

Example~\ref{ex2} on page~\pageref{ex2} describes an explicit example of an IFH. 

The following result relates improvement from Definition~\ref{def:improve} and inverse
fractional homomorphisms from Definition~\ref{def:ifh}.
The proof is based on Farkas' Lemma. 

\begin{proposition}
\label{prop:charfrac}
Let $\structA, \structB$ be valued $\sigma$-structures. Then, $\structA\less
  \structB$ if and only if there exists an IFH from $\structA$ to $\structB$.
\end{proposition}

\begin{remark}\label{rem:pos}
  Let us remark that an IFH $\omega$ from $\structA$ to $\structB$ is actually a distribution over the set of homomorphisms from $\pos{\structA}$ to $\pos{\structB}$, i.e., 
 every $g\in \supp{\omega}$ is a homomorphism from $\pos{\structA}$ to $\pos{\structB}$. 
 Indeed, for every $\tuple{x}\in R_f^{\pos{\structA}}$, where $f\in \sigma$ and $\tuple{x}\in A^{\ar{f}}$, 
 it must be the case that $f^{\structB}(g(\tuple{x}))\geq \sum_{h\in B^{A}} \omega(h)f^{\structA}(h^{-1}(g(\tuple{x}))) \geq  \omega(g)f^{\structA}(\tuple{x})>0$. 
Hence, $g(\tuple{x})\in R_f^{\pos{\structB}}$. In view of Proposition~\ref{prop:charfrac}, 
this offers another explanation of the fact that equivalence implies homomorphic equivalence (of the positive parts). 
\end{remark}

Let $\omega$ and $\omega'$ be inverse fractional homomorphisms from $\structA$ to $\structB$, and from $\structB$ to $\structC$, respectively. 
We define $\omega'\circ \omega: C^{A}\mapsto \qplus$ as 
\[\omega'\circ \omega(h)=\sum_{\substack{h_1:A\mapsto B, h_2:B\mapsto C\\ h_2\circ h_1=h}} \omega(h_1)\omega'(h_2).\]
Observe that $\omega'\circ \omega$ is an inverse fractional homomorphism from $\structA$ to $\structC$.

\subsection{Cores}

Appropriate notions of cores have played an important role in the complexity
classifications of left-hand side restricted CSPs~\cite{Grohe07:jacm}, right-hand
side restricted CSPs~\cite{Bulatov05:classifying,Bulatov17:focs,Zhuk20:jacm},
and right-hand side restricted VCSPs~\cite{tz16:jacm,Kolmogorov17:sicomp}. In
this paper, we will define cores around IFHs.

For two valued $\sigma$-structures $\structA$ and $\structB$, 
we say that an IFH $\omega$ from $\structA$ to $\structB$ is \emph{surjective} if every $g\in \supp{\omega}$ is surjective. 

\begin{definition}
A valued $\sigma$-structure $\structA$ is a \emph{core} if every IFH from $\structA$ to $\structA$ is surjective.
\end{definition}

Next we show that equivalent valued structure that are cores are in fact isomorphic. 

\begin{definition}
Let $\structA,\structB$ be valued $\sigma$-structures. 
An \emph{isomorphism} from $\structA$ to $\structB$ is a bijective mapping $h: A \mapsto B$ such that 
$f^{\structA}(\tuple{x}) = f^{\structB}(h(\tuple{x}))$ for all $(f,\tuple{x}) \in \tup{\structA}$. 
If such an $h$ exists, we say that $\structA$ and $\structB$ are \emph{isomorphic}. 
\end{definition}

\begin{proposition}
\label{prop:iso}
If $\structA$, $\structB$ are core valued $\sigma$-structures such that $\structA \equi \structB$, then $\structA$ and $\structB$ are isomorphic.
\end{proposition}

We now introduce the central notion of a core of a valued structure and show that every valued structure has a unique core (up to isomorphism). 

\begin{definition}
Let $\structA,\structB$ be valued $\sigma$-structures. We say that $\structB$ is a \emph{core} of $\structA$ if $\structB$ is a core and 
$\structA\equi\structB$.
\end{definition}

\begin{proposition}
\label{prop:exist-core}
Every valued structure $\structA$ has a core and all cores of $\structA$ are isomorphic. 
Moreover, for a given valued structure $\structA$, it is possible to effectively compute a core of $\structA$ and all cores of $\structA$ are over a universe of size at most $|A|$.
\end{proposition}

Proposition~\ref{prop:exist-core} allows us to speak about \emph{the} core of a valued structure. 
It follows then that equivalence can be characterised in terms of cores: $\structA$ and $\structB$ are equivalent if 
and only if their cores are isomorphic. 

We conclude with a technical characterisation of the property of being a core that will be important in the paper. 
Intuitively, it states that every non-surjective mapping from a core $\structA$ to itself is suboptimal with respect to a fixed weighting of the tuples of $\structA$. 

\begin{proposition}
\label{prop:core-char}
Let $\structA$ be a valued $\sigma$-structure. Then, $\structA$ is a core if and only if there exists a mapping $c: \tup{\structA} \mapsto \qplus$ such that for every non-surjective mapping $g: A \mapsto A$,
\[\sum_{(f,\tuple{x})\in \tup{\structA}}f^{\structA}(\tuple{x})c(f,\tuple{x}) <
  \sum_{(f,\tuple{x})\in \tup{\structA}}f^{\structA}(\tuple{x})c(f,g(\tuple{x})).\]
Moreover, such a mapping $c: \tup{\structA} \mapsto \qplus$ is computable, whenever $\structA$ is a core. 
\end{proposition}

\begin{example}
\label{ex2}
Let $\structA$ and $\structA'$ be the valued $\sigma$-structures depicted in Figure~\ref{fig:ex1}. 
Recall that $\sigma=\{f,\mu\}$, where $f$ is a $\{0,\infty\}$-valued binary function (represented by edges) and $\mu$ is unary (represented by node labels). 
Also, the elements of $\structA$ are denoted by $a_{ij}$, where $i$ and $j$ indicate the corresponding row and column of the grid, respectively  
(for readability, only $a_{11}$ is depicted in Figure~\ref{fig:ex1}).
We claim that $\structA'$ is the core of $\structA$. Indeed, since $\pos{\structA'}$ is a relational core, it follows that $\structA'$ is a core. 
To see that $\structA\less \structA'$, 
let $g:A\mapsto A'$ be the mapping that maps all elements in the $i$-th diagonal of $\structA$ (first diagonal is $\{a_{11}\}$, second diagonal is $\{a_{21},a_{12}\}$, and so on) 
to $a'_i$. Assigning $\omega(g)=1$ gives an IFH $\omega$ from $\structA$ to $\structA'$, 
and thus $\structA\less \structA'$ by Proposition~\ref{prop:charfrac}. 

Conversely, consider the mappings $g_1$,$g_2$,$g_3$,$g_4$,$g_5$,$g_6$ from $A'$ to $A$ 
that map $(a'_1,a'_2,a'_3,a'_4,a'_5)$ to $(a_{11},a_{21},a_{31},a_{32},a_{33})$, $(a_{11},a_{21},a_{22},a_{32},a_{33})$, $(a_{11},a_{12},a_{22},a_{23},a_{33})$, $(a_{11},a_{12},a_{13},a_{23},a_{33})$, $(a_{11},a_{12},a_{22},a_{32},a_{33})$ and $(a_{11},a_{21},a_{22},a_{23},a_{33})$, respectively. We define the distribution $\omega'(g_1)=1/3$, $\omega'(g_2)=1/12$, $\omega'(g_3)=1/12$, $\omega'(g_4)=1/3$, $\omega'(g_5)=1/12$ and $\omega'(g_6)=1/12$. Then, we have
\[
\sum_{k}\omega'(g_k)\mu^{\structA'}(g_k^{-1}(a_{ij})) =
\begin{cases}
(2 \times 1/3 + 4 \times 1/12) \times 1 \quad &\text{if } a_{ij} \in \{a_{11},a_{33}\}\\
(1/3 + 2 \times 1/12) \times 2 \quad &\text{if } a_{ij} \in \{a_{12},a_{21},a_{23},a_{32}\}\\
(1/3) \times 3 \quad &\text{if } a_{ij} \in \{a_{13},a_{31}\}\\
(4 \times 1/12) \times 3 \quad &\text{if } a_{ij} = a_{22}
\end{cases}
\]
and hence for all $i,j$, $\sum_{k}\omega'(g_k)\mu^{\structA'}(g_k^{-1}(a_{ij}))
  = 1 \leq \mu^{\structA}(a_{ij})$. It follows that $\omega'$ is an IFH from $\structA'$ to $\structA$. Therefore, $\structA'\less\structA$ and $\structA'$ is the core of $\structA$. In particular, $\structA$ is not a core. 
As we explain later in Example~\ref{ex:cores}, it is possible to modify $\structA$ (more precisely, $\mu^{\structA}$) so that it becomes a core. 

\end{example}

\subsection{Treewidth modulo equivalence} 

In this section we show an elementary property of cores that is crucial for our purposes: the treewidth of the core of a valued structure $\structA$ is the lowest possible among all structures equivalent to $\structA$.

\begin{proposition}
\label{prop:core-tw}
Let $\structA$ be a valued $\sigma$-structure and $\structA'$ be its core. Then, $\tw{\structA'}\leq \tw{\structA}$. 
\end{proposition}

\begin{proof}
Since treewidth is preserved under relational substructures, it suffices to show that 
$\pos{\structA'}$ is isomorphic to a substructure of $\pos{\structA}$, i.e., 
there is an injective homomorphism from $\pos{\structA'}$ to $\pos{\structA}$. 
Let $\omega$ and $\omega'$ be IFHs from $\structA'$ to $\structA$, and from $\structA$ to $\structA'$, 
respectively. 
Pick any mapping $g\in\supp{\omega}$. 
As observed at the end of Section~\ref{sec:inv-frac}, $g$ has to be a homomorphism from $\pos{\structA'}$ to $\pos{\structA}$. 
It suffices to show that $g$ is injective. But this follows immediately from the
fact that $\omega'\circ\omega$ is an IFH from $\structA'$ to itself and the
 assumption that $\structA'$ is a core.
\end{proof}

We conclude Section~\ref{sec:equiv} with the following useful characterisation of ``being equivalent to a bounded treewidth structure'' in terms of cores. 

\begin{proposition}
\label{prop:tw-equiv-core}
Let $\structA$ be a valued $\sigma$-structure and $k\geq 1$. Then, the following are equivalent:
\begin{enumerate}
\item There is a valued $\sigma$-structure $\structA'$ such that $\structA'\equi\structA$ and $\structA'$ has treewidth at most $k$.
\item The treewidth of the core of $\structA$ is at most $k$. 
\end{enumerate}
\end{proposition}

\begin{proof}
(2) $\Rightarrow$ (1) is immediate. 
For (1) $\Rightarrow$ (2), let $\structA'$ be of treewidth at most $k$ such that $\structA'\equi\structA$. 
Let $\structB$ and $\structB'$ be the cores of $\structA$ and $\structA'$, respectively. 
Since $\structB\equi\structB'$, $\structB$ and $\structB'$ are isomorphic, by Proposition~\ref{prop:iso}. 
By Proposition~\ref{prop:core-tw}, the treewidth of $\structB'$ is at most $k$, and so is the treewidth of $\structB$. 
\end{proof}

\section{Complexity of \problem{VCSP}($\C$, $-$)}
\label{sec:compl}

Let $\C$ be a class of valued structures. 
We say that $\C$ has \emph{bounded arity} if there is a constant $r\geq 1$ such that for every valued $\sigma$-structure $\structA\in \C$ and $f\in \sigma$, 
we have that $\ar{f}\leq r$. 
Similarly, we say that $\C$ has bounded treewidth \emph{modulo equivalence} if there is a constant $k\geq 1$ such that every 
$\structA\in \C$ is equivalent to a valued structure $\structA'$ with $\tw{\structA'}\leq k$.
The following is our first main result.

\begin{theorem}[\textbf{Complexity classification}]
\label{theo:main}
Assume FPT $\neq$ W[1]. Let $\C$ be a recursively enumerable class of valued structures of bounded arity. Then, the following are equivalent: 
\begin{enumerate}
\item \problem{VCSP}($\C$, $-$) is in PTIME.
\item $p$-\problem{VCSP}($\C$, $-$) is in FPT. 
\item $\C$ has bounded treewidth modulo equivalence. 
\end{enumerate}
\end{theorem}

\begin{remark}
Although Grohe's result~\cite{Grohe07:jacm} for CSPs looks almost identical to
Theorem~\ref{theo:main}, we emphasise that his result involves a different type of equivalence. In
Grohe's case, the equivalence in question is \emph{homomorphic equivalence}
whereas in our case the equivalence in question involves \emph{improvement} (cf.
Definition~\ref{def:equiv}). As we will explain later in this section, 
Grohe's classification follows as a special case of Theorem~\ref{theo:main}. 
\end{remark}

\begin{remark}
As in~\cite{Grohe07:jacm}, we can remove the condition in Theorem~\ref{theo:main} of $\C$ being recursively enumerable, 
by assuming a stronger hypothesis than FPT $\neq$ W[1] regarding non-uniform complexity classes.
\end{remark}

Bounded treewidth implies bounded treewidth modulo equivalence but we will show
in Example~\ref{ex:btw} that the converse is not true in general. Thus
Theorem~\ref{theo:main} gives new tractable cases compared to classes of VCSPs
of (previously known to be tractable) bounded treewidth.

Bounded treewidth modulo equivalence implies bounded treewidth modulo
homomorphic equivalence (of the positive parts) but we will show in
Example~\ref{ex:cores} below that the converse is not true in general either.
Therefore, Theorem~\ref{theo:main} tells us that the tractability frontier for
VCSP($\C$, $-$) lies \emph{strictly} between bounded treewidth and bounded
treewidth modulo homomorphic equivalence. 

\begin{example}
\label{ex:btw}
We construct a class of structures of unbounded treewidth whose cores are of
bounded treewidth (in fact treewidth $1$).

By Proposition~\ref{prop:tw-equiv-core}, a class $\C$ has bounded treewidth
modulo equivalence if and only if the class given by the \emph{cores} of the
valued structures in $\C$ has bounded treewidth. Thus we obtain a separation of
bounded treewidth and bounded treewidth modulo equivalence.

Consider the signature $\sigma=\{f,\mu\}$, where $f$ and $\mu$ 
are binary and unary function symbols, respectively. 
For $n\geq 1$, let $\structA_n$ be the valued $\sigma$-structure with universe $A_n=\{1,\dots,n\}\times \{1,\dots,n\}$ such that 
(i) $f^{\structA_n}((i,j),(i',j'))=\infty$ if $i\leq i'$, $j\leq j'$, and $(i'-i)+(j'-j)=1$; otherwise $f^{\structA_n}((i,j),(i',j'))=0$, and (ii) $\mu^{\structA_n}((i,j))=1$, for all $(i,j)\in A_n$. 
Also, for $n\geq 1$, let $\structA'_n$ be the valued $\sigma$-structure with universe $A'_n=\{1,\dots,2n-1\}$ such that 
(i) $f^{\structA_n'}(i,j)=\infty$ if $j=i+1$; otherwise $f^{\structA_n'}(i,j)=0$, 
and (ii) $\mu^{\structA_n'}(i)=i$, for $1\leq i\leq n$, and $\mu^{\structA_n'}(i)=2n-i$, for $n+1\leq i\leq 2n-1$. 
The structures $\structA$ and $\structA'$ from Figure~\ref{fig:ex1} correspond
to $\structA_3$ and $\structA'_3$, respectively; informally $\structA_n$ is a
crisp directed grid of size $n \times n$ with a unary function $\mu$ with weight
$1$ applied to each element. Generalising the reasoning behind
Example~\ref{ex2}, we argue in Appendix~\ref{app:example} that for each $n \geq
1$ the valued structure $\structA'_n$ is the core of $\structA_n$. Since
$\tw{\structA_n'}=1$, the class $\C\defeq\{\structA_n\mid n\geq 1\}$ has bounded
treewidth modulo equivalence. However, $\C$ has unbounded treewidth as the
Gaifman graphs in $\{G(\pos{\structA_n})\mid n\geq 1\}$ correspond to the class
of (undirected) \emph{grids}, which is a well-known family of graphs with
unbounded treewidth (see, e.g.~\cite{diestel10:graph}). We also describe in
Appendix~\ref{app:example} how to alter the definition of $\C$ to obtain a class
of \emph{finite-valued} structures (taking on finite values in $\qplus$)
that has bounded treewidth modulo equivalence but Gaifman graphs of unbounded treewidth.
\end{example}

\begin{figure}
\centering
\begin{tikzpicture}
  \pgftext{\includegraphics[scale=0.4]{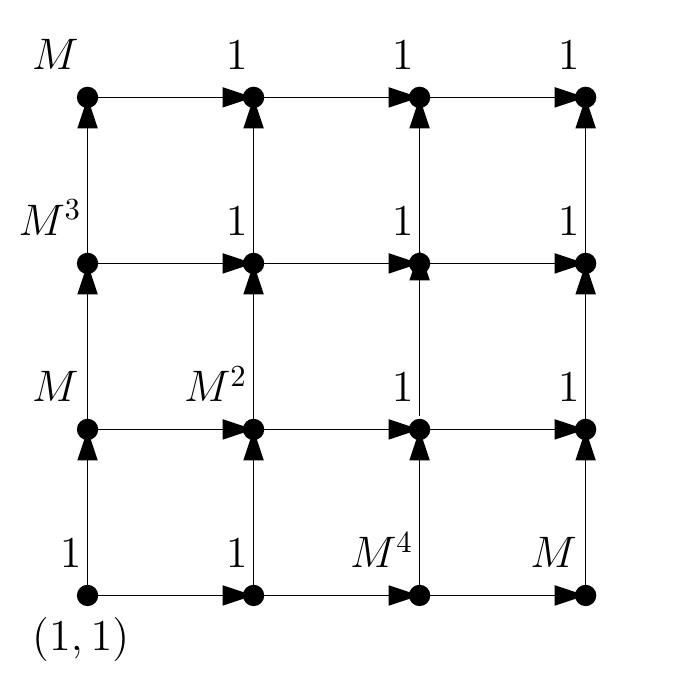}} at (0,0):
  \node at (-2.2,-0.6) {\scriptsize $(2,1)$};
  \node at (-0.6,-2.0) {\scriptsize $(1,2)$};
\end{tikzpicture}
\caption{The valued $\sigma$-structure $\structC_4$ from Example~\ref{ex:cores} ($M>16$). }
\label{fig:ex-cores}
\end{figure}

\begin{example}
\label{ex:cores}
We construct a class of structures of unbounded treewidth modulo equivalence but
bounded treewidth modulo homomorphic equivalence, thus separating bounded
treewidth modulo equivalence from bounded treewidth modulo homomorphic
equivalence.

For $n\geq 3$, let $\structA_n$ be the valued $\sigma$-structure from Example~\ref{ex:btw}. 
Let $\structC_n$ be the valued $\sigma$-structure with the same universe as $\structA_n$, i.e., $C_n=\{1,\dots,n\}\times \{1,\dots,n\}$, 
such that $f^{\structC_n}=f^{\structA_n}$ and $\mu^{\structC_n}$ is defined as follows. 
Let $D_1,\dots,D_n$ be the $n$ first diagonals of $\structC_n$ starting from the bottom left corner $(1,1)$ (see Figure~\ref{fig:ex-cores} for an illustration of $\structC_4$). 
For $1\leq i\leq n$, let $E_i$ be the top-left to bottom-right enumeration of $D_i$. 
In particular, $E_1=\left((1,1)\right)$, $E_2=\left((2,1), (1,2)\right)$, $E_3=\left((3,1),(2,2),(1,3)\right)$ and $E_n=\left((n,1),(n-1,2),\dots,(1,n)\right)$.  
Fix an integer $M=M(n)$ such that $M>n^2$. 
The values assigned by $\mu^{\structC_n}$ to $E_1$, $E_2$ and $E_3$ are $\left(1\right)$, $\left(M, 1\right)$ and 
$\left(M^3, M^2, M^4\right)$, respectively, and for $E_i$, with $4\leq i\leq n$, is $\left(M,1,\dots,1,M\right)$. 
All remaining elements in $C_n\setminus \bigcup_{1\leq i\leq n} D_i$ receive cost $1$. 
Figure~\ref{fig:ex-cores} depicts the case of $\structC_4$.

Let $\C\defeq\{\structC_n\mid n\geq 3\}$. 
Note first that $\pos{\structC_n}$ is homomorphically equivalent to the relational structure $\rels{P}_{2n-1}$ over relational 
signature $\rel{\sigma}=\{R_f,R_\mu\}$ 
(recall the definition of $\rel{\sigma}$ from Section~\ref{sec:prelims}), whose universe is $P_{2n-1}=\{1,\dots,2n-1\}$, 
$R_{\mu}^{\rels{P}_{2n-1}}=P_{2n-1}$ and $R_{f}^{\rels{P}_{2n-1}}=\{(i,i+1)\mid 1\leq i\leq 2n-2\}$. 
Since $\tw{\rels{P}_{2n-1}}=1$, for all $n\geq 3$, 
it follows that $\{\pos{\structC_n}\mid n\geq 3\}$ has bounded treewidth modulo homomorphic equivalence. 
We claim that $\C$ has unbounded treewidth modulo (valued) equivalence. 
It suffices to show that $\structC_n$ is a core, for all $n\geq 3$. 
In order to prove this, we apply Proposition~\ref{prop:core-char}. 
Fix $n\geq 3$ and define $c:\tup{\structC_n}\mapsto \qplus$ 
such that (i) $c(f,\tuple{x})=0$ if $f^{\structC_n}(\tuple{x})=\infty$; otherwise $c(f,\tuple{x})=1$, 
and (ii) $c(\mu,x)=1/\mu^{\structC_n}(x)$, for all $x\in C_n$. 
We have $\sum_{(p,\tuple{y})\in \tup{\structC_n}}p^{\structC_n}(\tuple{y})c(p,\tuple{y})=|C_n|=n^2$. 
Next we show that if $g:C_n\mapsto C_n$ satisfies that $v(g)=\sum_{(p,\tuple{y})\in \tup{\structC_n}}p^{\structC_n}(\tuple{y})c(p,g(\tuple{y}))\leq n^2$, 
then $g$ is the identity mapping. Using Proposition ~\ref{prop:core-char}, this implies that $\structC_n$ is a core. 

Let $g:C_n\mapsto C_n$ such that $v(g)\leq n^2$. 
The mapping $g$ must satisfy the following two conditions: 
(a) $g$ is a homomorphism from $\pos{\structC_n}$ to $\pos{\structC_n}$ (otherwise $v(g)=\infty$), and 
(b) for every $x\in C_n$, $\mu^{\structC_n}(x)\leq\mu^{\structC_n}(g(x))$, 
otherwise $v(g)\geq \mu^{\structC_n}(x) c(\mu,g(x))=\mu^{\structC_n}(x)/\mu^{\structC_n}(g(x))\geq M>n^2$. 
We can argue inductively, and show that $g$ is the identity over $D_i$, for all $1\leq i\leq n$. 
Condition (a) implies that $g$ is the identity over the remaining elements in $C_n\setminus \bigcup_{1\leq i\leq n} D_i$, as required. 
For $D_1$, we have that $g((1,1))=(1,1)$ by condition (a). 
For $D_2$, note that (a) implies that $\{g((2,1)),g((1,2))\}\subseteq \{(2,1),(1,2)\}$. 
By condition (b), $g((2,1))=(2,1)$. To see that $g((1,2))=(1,2)$, 
suppose by contradiction that $g((1,2))=(2,1)$, then condition (a) implies that $g((1,3))\in\{(3,1),(2,2)\}$, 
which violates (b). 
For the case $3\leq i\leq n$, recall that $E_i=x_1,x_2,\dots, x_{|D_i|}$ is the above-defined enumeration of $D_i$. 
Since $g$ is the identity over $D_{i-1}$ and by condition (a), we have that $g$ is the identity over $\{x_2,\dots,x_{|D_i|-1}\}$. 
As $\mu^{\structC_n}(x_1)>\mu^{\structC_n}(x_2)$ and $\mu^{\structC_n}(x_{|D_i|})>\mu^{\structC_n}(x_{|D_i|-1})$, 
conditions (a) and (b) imply that $g(x_1)=x_1$ and $g(x_{|D_i|})=x_{|D_i|}$, as required.

To conclude this example we note that the class of valued $\sigma$-structures $\{\structB_n\mid  n \geq 3 \}$ where each $\structB_n$ is derived from $\structC_n$ by setting $B_n = C_n$, $\mu^{\structB_n} = \mu^{\structC_n}$ and $f^{\structB_n}(\tuple{x}) = \min(1,f^{\structC_n}(\tuple{x}))$ for all $\tuple{x} \in (B_n)^2$ is an example of a \emph{finite-valued} class of structures that has bounded treewidth modulo homomorphic equivalence but unbounded treewidth modulo valued equivalence, since each $\structB_n$ is a core (this follows from exactly the same argument used for $\structC_n$).
\end{example}

\paragraph{Corollaries of the complexity classification}

We can obtain the classification for CSPs of Dalmau et al.~\cite{Dalmau02:width} and Grohe~\cite{Grohe07:jacm} 
as a special case of Theorem~\ref{theo:main}. 
Indeed, we can associate with a relational $\tau$-structure $\rels{A}$ a valued $\sigma_\tau$-structure $\structA_{0,\infty}$ such that 
(i) $\sigma_\tau=\{f_{R}\mid R\in \tau, \ar{f_R}=\ar{R}\}$, (ii) $\rels{A}$ and $\structA_{0,\infty}$ have the same universe $A$, and 
(iii) if $\tuple{x}\in R^{\rels{A}}$, then $f_R^{\structA_{0,\infty}}(\tuple{x})=\infty$, otherwise $f_R^{\structA_{0,\infty}}(\tuple{x})=0$, for every $R\in \tau$ and $\tuple{x}\in A^{\ar{R}}$.  
For a class $\C$ of relational structures, we define the class of valued structures $\C_{0,\infty}\defeq \{\structA_{0,\infty}\mid \rels{A}\in \C\}$. 
It is not hard to check that, when $\C$ is of bounded arity, \problem{CSP}($\C$,$-$) reduces in polynomial time to \problem{VCSP}($\C_{0,\infty}$,$-$) and vice versa. 
Hence, a classification of \problem{CSP}($\C$,$-$), for $\C$'s of bounded arity, is equivalent to a classification of \problem{VCSP}($\C_{0,\infty}$,$-$). 
Finally, note that $\C$ has bounded treewidth modulo homomorphic equivalence if and only if $\C_{0,\infty}$ has bounded treewidth modulo (valued) equivalence. 
This implies the known CSP classification from \cite{Dalmau02:width} and \cite{Grohe07:jacm}.

In his PhD thesis~\cite{Farnqvist13:phd}, F\"arnqvist also considered the complexity
of \problem{VCSP}($\C$, $-$). However, he considered a different definition of the problem, that we denote by  \problem{VCSP}$_F$($\C$, $-$). 
Formally, for a relational $\tau$-structure $\rels{A}$, let $\structA_{F}$ be the valued $\sigma_{\rels{A}}$-structure such that 
 (i) $\sigma_{\rels{A}}=\{f_{R,\tuple{x}}\mid R\in \tau, \tuple{x}\in R^{\rels{A}}, \ar{f_{R,\tuple{x}}}=\ar{R}\}$, 
 (ii) $\rels{A}$ and $\structA_F$ have the same universe $A$, and 
 (iii) for every $f_{R,\tuple{x}}\in \sigma_{\rels{A}}$ and $\tuple{x}\in A^{\ar{f_{R,\tuple{x}}}}$, 
 we have that $f_{R,\tuple{x}}^{\structA_F}(\tuple{x})=1$ and $f_{R,\tuple{x}}^{\structA_F}(\tuple{y})=0$, 
 for all $\tuple{y}\neq \tuple{x}$. 
 For a class of relational structures $\C$, 
 \problem{VCSP}$_F$($\C$, $-$) is precisely the problem \problem{VCSP}($\C_F$, $-$), 
 where $\C_F\defeq\{\structA_F\mid \rels{A}\in \C\}$. 
 It was shown in~\cite{Farnqvist13:phd} that for a class $\C$ of relational structures of bounded arity,  \problem{VCSP}$_F$($\C$, $-$) is tractable if and only if 
 $\C$ has bounded treewidth. 
 This result follows directly from Theorem~\ref{theo:main} as every valued structure in a class of the form $\C_F$ is a (valued) core, 
 and hence,  bounded treewidth modulo equivalence  boils down to bounded treewidth. 

 Intuitively, \problem{VCSP}$_F$($\C$, $-$) restricts \problem{VCSP} only based
 on the (multiset of) tuples appearing in the structures from $\C$. 
 In contrast, our definition of \problem{VCSP($\C$, $-$)} considers directly the structures in $\C$.  
 This allows us for a more fine-grained analysis of structural restrictions, and in particular, provides us with new tractable classes beyond bounded treewidth. 
 Indeed, as Example~\ref{ex:btw} illustrates, we can find simple tractable classes of valued structures with unbounded treewidth.
 
Finally, let us note that since Theorem~\ref{theo:main} applies to \emph{all}
valued structures, it in particular covers the \emph{finite-valued} \problem{VCSP}, where 
 all functions are restricted to take finite values in $\qplus$, and hence the tractability part of Theorem~\ref{theo:main} 
 directly applies to the finite-valued case. The hardness part also applies to the finite-valued case. 
 Indeed, the right-hand side structure $\structB$ constructed in the reduction of Proposition~\ref{prop:main-hardness} 
 can be made to be finite-valued as explained at the end of the proof. Therefore, Theorem~\ref{theo:main} also gives a classification for finite-valued VCSPs.
 Moreover, Examples~\ref{ex:btw} and~\ref{ex:cores} demonstrate that already for finite-valued structures the tractability frontier is strictly between bounded treewidth and bounded treewidth modulo homomorphic equivalence.

\medskip
The rest of this section is devoted to proving the hardness part of Theorem~\ref{theo:main}, 
i.e., the implication (2) $\Rightarrow$ (3). The tractability part of Theorem~\ref{theo:main} (implication (3) $\Rightarrow$ (1)) is established in Section~\ref{sec:sa}. 
In particular, it will follow from Theorem \ref{thm:sa} that, if there is a constant $k\geq 1$ such that every valued structure in the class $\C$ is equivalent to a valued structure of treewidth at most $k$, 
then \problem{VCSP}($\C$, $-$) can be solved in polynomial time using the
$(k+1)$-st level of the Sherali-Adams LP hierarchy. 
The remaining implication (1) $\Rightarrow$ (2) is immediate. 

\subsection{Hardness}
\label{sec:hard-main}

We start with the notion of fpt-reductions~\cite{FG06} tailored to our setting.
Formally, a decision problem $\mathcal{P}$ with parameter $\kappa$ over $\Sigma$ is a
subset of $\Sigma^*$, the set of all strings over the alphabet $\Sigma$,
describing the ``yes'' instances of $\mathcal{P}$, and
$\kappa:\Sigma^*\mapsto\mathbb{N}$. It is known that each optimisation problem has an
equivalent decision problem. We denote by $p$-\problem{VCSP$_d$}($\C$, $-$) the
decision version of $p$-\problem{VCSP}($\C$, $-$). Formally, $p$-\problem{VCSP$_d$}($\C$,
$-$) with parameter $\kappa'$ over $\Sigma'$ is a subset of $(\Sigma')^*$ such that
$x=((\structA,\structB), c)\in p$-\problem{VCSP$_d$}($\C$, $-$) if and only if 
$(\structA,\structB)$ is a \problem{VCSP}($\C$, $-$) instance such that $\opt{\structA,\structB}\leq c$, and
$\kappa':(\Sigma')^*\mapsto\mathbb{N}$ is defined by $\kappa'(x)=|\structA|$.

An \emph{fpt}-reduction from $(\mathcal{P},\kappa)$ to $p$-\problem{VCSP$_d$}($\C$, $-$)
is a mapping $red:\Sigma^* \mapsto (\Sigma')^*$ such that (i) for all
$x\in\Sigma^*$ we have $x\in\mathcal{P}$ if and only if $red(x)\in
p$-\problem{VCSP$_d$}($\C$, $-$); (ii) there is a computable function
$f:\mathbb{N}\mapsto\mathbb{N}$ and an algorithm that, given $x\in\Sigma^*$,
computes $red(x)$ in time $f(\kappa(x))\cdot|x|^{O(1)}$; and (iii) there is a
computable function $g:\mathbb{N}\mapsto\mathbb{N}$ such that for all instances
$x\in\Sigma^*$, we have $\kappa'(red(x))\leq g(\kappa(x))$.

Let us mention that our hardness result does not follow directly from Grohe's result for CSPs \cite{Grohe07:jacm}. 
The natural approach is to define, for a class of valued structures $\C$, the class of relational structures $\pos{\C}=\{\pos{\structA}\mid \structA\in \C\}$. 
Then one can observe that $p$-\problem{CSP}($\pos{\C}$, $-$) fpt-reduces to
$p$-\problem{VCSP$_d$}($\C$, $-$), and hence W[1]-hardness of the former problem implies hardness for the latter. 
However, if $\C$ has unbounded treewidth modulo equivalence, the class $\pos{\C}$ does not necessarily have unbounded treewidth modulo homomorphic equivalence 
(see Example~\ref{ex:cores}), 
and hence $\pos{\C}$ is not necessarily hard according to Grohe's classification. 

We instead adapt Grohe's proof to the case of VCSPs. 
We need some notation. 
For $k\geq 1$, a \emph{$k$-clique} of a graph is a clique on $k$ vertices.  
A graph $H$ is a \emph{minor} of a graph $G$ if $H$ is isomorphic 
to a graph that can be obtained from a subgraph of $G$ by contracting edges (for more details see, e.g.~\cite{diestel10:graph}). 
For $k,\ell\geq 1$, the $(k\times \ell)$-grid is the graph with vertex set $\{1,\dots,k\}\times\{1,\dots,\ell\}$ and 
an edge between $(i,j)$ and $(i',j')$ if $|i-i'|+|j-j'|=1$. 
The parameterised problem $p$-\problem{CLIQUE} asks, given instance $(G,k)$, 
whether there is a $k$-clique in $G$, and has parameter $\kappa$ such that $\kappa(G,k)=k$. 
It is a well-known result that $p$-\problem{CLIQUE} is complete for $W[1]$ under fpt-reductions \cite{DF95:fpt}. 
The implication $(2)\Rightarrow (3)$ in Theorem~\ref{theo:main} follows from the following proposition. 

\begin{proposition}
\label{prop:main-hardness}
Let $\C$ be a recursively enumerable class of valued structures of bounded arity. Suppose $\C$ is of unbounded treewidth modulo equivalence. 
If $p$-\problem{VCSP}($\C$, $-$) is fixed-parameter tractable then FPT $=$ W[1].  
\end{proposition}

\begin{proof}
We present an fpt-reduction from $p$-\problem{CLIQUE} to $p$-\problem{VCSP$_d$}($\C$, $-$). 
More precisely, given an instance $(G,k)$ of $p$-\problem{CLIQUE}, we shall construct valued structures $\structA'\in \C$ and 
$\structB$, together with a threshold $M^*\geq 0$, such that $G$ contains a
$k$-clique if and only if $\opt{\structA',\structB}\leq M^*$. 
As in~\cite{Grohe07:jacm}, we rely on the Excluded Grid
Theorem~\cite{Robertson86:excluding}, which states that there is a function $w:\mathbb{N}\mapsto\mathbb{N}$ such that every graph $H$ of treewidth at least $w(k)$ contains the $(k\times k)$-grid as a minor. 
Given an instance $(G,k)$ of $p$-\problem{CLIQUE}, with $k\geq 2$, 
we start by enumerating the class $\C$ until we obtain a valued structure $\structA'\in \C$ with core $\structA$ such that $\tw{\structA}\geq w(K)$, where $K={k\choose 2}$. 
By Proposition~\ref{prop:tw-equiv-core}, such $\structA'\in \C$ always exists. 
Let $\rels{A}\defeq \pos{\structA}$. 
By the Excluded Grid Theorem, the Gaifman graph $G(\rels{A})$ of $\rels{A}$ (see
the definition of the Gaifman graph in Section~\ref{sec:prelims}) 
contains the $(K\times K)$-grid, and hence, the $(k\times K)$-grid as a minor. 
Since $\C$ is recursively enumerable and cores are computable (by Proposition~\ref{prop:exist-core}), 
the valued structure $\structA'$ and its core $\structA$ can be effectively computed in time $\alpha(k)$, where $\alpha$ is a computable function. 

In order to define $\structB$, 
we exploit the main construction in~\cite{Grohe07:jacm}, which defines a relational structure $\rels{B}$ from $G$, $k$ and $\rels{A}$. 
The key property of $\rels{B}$ is that, assuming $\rels{A}$ is a relational core, 
then $G$ contains a $k$-clique if and only if there is a homomorphism from $\rels{A}$ to $\rels{B}$. 
Since in our case $\rels{A}$ is not necessarily a relational core, 
we restate in the following lemma the properties of $\rels{B}$ simply in terms of surjective homomorphisms.  
Together with our characterisation of (valued) cores in Proposition~\ref{prop:core-char}, 
this will allow us to define our required valued structure $\structB$ and threshold $M^*$.

\begin{lemma}
\label{lemma:grohe}
Given $k\geq 2$, $K={k\choose 2}$, graph $G$ and relational $\tau$-structure $\rels{A}$ such that the $(k\times K)$-grid is a minor of its Gaifman graph $G(\rels{A})$, 
there is a relational $\tau$-structure $\rels{B}$ such that 
\begin{enumerate}
\item There exists a fixed homomorphism $\pi$ from $\rels{B}$ to $\rels{A}$ such that the following are equivalent:
\begin{enumerate}[label=(\alph*)]
\item $G$ contains a $k$-clique. 
\item There is a homomorphism $h$ from $\rels{A}$ to $\rels{B}$ such that $\pi\circ h$ is the identity mapping from $A$ to itself, where $A$ is the universe of $\rels{A}$.
\end{enumerate}
\item $\rels{B}$ can be computed in time $\beta(|\rels{A}|,k)\cdot |G|^{O(r(\rels{A}))}$, where $\beta$ is a computable function and $r(\rels{A})\geq 1$ is the arity of the relational signature $\tau$. 
\end{enumerate}

\end{lemma}

\begin{proof}
We can separate $\rels{A}$ into two relational $\tau$-structures $\rels{A}'$ and $\rels{A}\setminus \rels{A'}$ such that the disjoint union of their Gaifman graphs $G(\rels{A}')$ and $G(\rels{A}\setminus \rels{A'})$ 
is precisely $G(\rels{A})$, and $G(\rels{A}')$ is a connected graph containing the $(k\times K)$-grid as a minor. 
The main hardness reduction in~\cite{Grohe07:jacm} produces for the connected structure $\rels{A}'$ a relational $\tau$-structure $\rels{B}'$ and a 
homomorphism $\pi'$ from $\rels{B}'$ to $\rels{A}'$ such that conditions (1) and (2) hold with respect to $\rels{A}'$, $\rels{B}'$ and $\pi'$. 
(This is shown in Section 4 from \cite{Grohe07:jacm} as a preparation for the proof of Theorem 4.1.)
We define our required structure $\rels{B}$ to be the disjoint union of $\rels{B}'$  and $\rels{A}\setminus \rels{A}'$.  
We define $\pi$ to be the homomorphism from $\rels{B}$ to $\rels{A}$ such that 
$\pi(b)=\pi'(b)$, if $b$ belongs to $\rels{B}'$ and $\pi(b)=b$, otherwise. It suffices to check condition (1).

Suppose that $G$ contains a $k$-clique. Then there is a homomorphism $h'$ from $\rels{A}'$ to $\rels{B}'$ such that $\pi'\circ h'=\id$.  
We can define the mapping $h$ from $\rels{A}$ to $\rels{B}$ such that $h(a)=h'(a)$, if $a$ belongs to $\rels{A}'$, and $h(a)=a$, otherwise. 
It follows that $h$ is a homomorphism and $\pi\circ h=\id$. 
Assume now that there is a homomorphism $h$ from $\rels{A}$ to $\rels{B}$ such that $\pi\circ h=\id$. 
Then there exists a homomorphism $h'$ from $\rels{A'}$ to $\rels{B}'$ such that $\pi'\circ h'=\id$ ($h'$ is simply the restriction of $h$ to $\rels{A}'$). 
Since $\rels{A}'$ and $\rels{B}'$ satisfy condition (1), we conclude that $G$ contains a $k$-clique.
\renewcommand{\qedsymbol}{$\blacksquare$}\end{proof}

Let $\rels{B}$ be the relational structure from Lemma~\ref{lemma:grohe} applied to $G,k$ and $\rels{A}=\pos{\structA}$. 
Recall that $\rels{A}$, and hence $\rels{B}$, are defined over the relational signature $\rel{\sigma}$, where $\sigma$ is the signature of $\structA$. 
Since $\structA$ is a core, by Proposition~\ref{prop:core-char} we can compute a function $c^*:\tup{\structA}\mapsto \qplus$ such that for every non-surjective mapping $g: A \mapsto A$, it is the case that 
\[\sum_{(f,\tuple{x})\in \tup{\structA}}f^{\structA}(\tuple{x})c^*(f,\tuple{x}) < \sum_{(f,\tuple{x})\in
\tup{\structA}}f^{\structA}(\tuple{x})c^*(f,g(\tuple{x})).\]

From $\rels{B}$ and $c^*$ we construct a valued structure $\structB$ over the same signature $\sigma$ of $\structA$ as follows. 
Let $\pi$ be the homomorphism from $\rels{B}$ to $\rels{A}$ given by Lemma~\ref{lemma:grohe}. 
Let $M^* \defeq  \sum_{(f,\tuple{x})\in \tup{\structA}}f^{\structA}(\tuple{x})c^*(f,\tuple{x})<\infty$. 
The universe of $\structB$ is the universe $B$ of $\rels{B}$. For each $f\in \sigma$ and $\tuple{x}\in B^{\ar{f}}$, we define:
\[
f^{\structB}(\tuple{x})=
\begin{cases}
c^*(f,\pi(\tuple{x})) &\quad \text{if $\tuple{x}\in R_f^{\rels{B}}$}\\
\infty &\quad \text{otherwise}.
\end{cases}
\]
We show that $G$ contains a $k$-clique if and only if $\opt{\structA',\structB}\leq M^*$. 
Before doing so, note that the total running time of the reduction is
$\alpha(k)+\beta(\alpha(k),k)|G|^{O(r(\rels{A}))}$, where $\beta$ is from
Lemma~\ref{lemma:grohe} and $r(\rels{A})$ is the arity of $\rel{\sigma}$, and hence the arity of $\sigma$. 
Since the class $\C$ has bounded arity, there is a constant $r\geq 1$ such that $r(\rels{A})\leq r$. 
Thus the running time of the reduction is $\beta'(k)|G|^{O(1)}$ for a computable function $\beta'$. 
Also, since $\structA'$ is computed in time $\alpha(k)$, we have that $|\structA'|\leq \alpha(k)$. 
It follows that our reduction is actually an fpt-reduction, and hence, $p$-\problem{VCSP$_d$}($\C$, $-$) is W[1]-hard. 
If $p$-\problem{VCSP}($\C$, $-$) is in FPT, then $p$-\problem{VCSP$_d$}($\C$, $-$) is in FPT, and consequently  FPT $=$ W[1], as required.

Assume that $G$ contains a $k$-clique. By Lemma~\ref{lemma:grohe}, 
there is a homomorphism $h$ from $\rels{A}$ to $\rels{B}$ such that $\pi\circ h$ is the identity. 
We have that 
\begin{align*}
\opt{\structA',\structB}= \opt{\structA,\structB} & \leq \costb{\structA\mapsto\structB}{h} = \sum_{(f,\tuple{x})\in\tup{\structA}_{>0}} f^{\structA}(\tuple{x}) f^{\structB}(h(\tuple{x}))=\sum_{\tuple{x}\in R_f^{\rels{A}}} f^{\structA}(\tuple{x}) f^{\structB}(h(\tuple{x}))\\
 & = \sum_{\tuple{x}\in R_f^{\rels{A}}} f^{\structA}(\tuple{x}) c^*(f,\pi(h(\tuple{x}))) = \sum_{(f,\tuple{x})\in\tup{\structA}_{>0}} f^{\structA}(\tuple{x}) c^*(f,\tuple{x})=M^*.
\end{align*}

Suppose now that $\opt{\structA',\structB}\leq M^*$. 
In particular, $\opt{\structA,\structB}\leq M^*$. 
Let $h^*:A\mapsto B$ be a mapping with cost $\opt{\structA,\structB}$. 
Note that $h^*$ is a homomorphism from  $\rels{A}$ to $\rels{B}$ (otherwise $\costb{\structA\mapsto\structB}{h^*}=\infty$). 
We show that $\pi\circ h^*$ is surjective. 
Assume to the contrary.  
By the definition of $c^*$, it follows that $M^*<\sum_{(f,\tuple{x})\in \tup{\structA}}f^{\structA}(\tuple{x})c^*(f, \pi(h^*(\tuple{x})))$. 
On the other hand, using the fact that $h^*$ is a homomorphism, we have that 

\begin{align*}
\sum_{(f,\tuple{x})\in \tup{\structA}}f^{\structA}(\tuple{x})c^*(f, \pi(h^*(\tuple{x})))&=
\sum_{\tuple{x}\in R_f^\rels{A}}f^{\structA}(\tuple{x})c^*(f, \pi(h^*(\tuple{x})))=
\sum_{\tuple{x}\in R_f^\rels{A}}f^{\structA}(\tuple{x})f^{\structB}(h^*(\tuple{x}))\\
&=\sum_{(f,\tuple{x})\in \tup{\structA}_{>0}}f^{\structA}(\tuple{x})f^{\structB}(h^*(\tuple{x}))=\cost{h^*}.
\end{align*}

Hence, $M^*<\cost{h^*}$; a contradiction. 
It follows that $\pi\circ h^*$ is an isomorphism of $\rels{A}$. 
We can define $g=h^*\circ (\pi\circ h^*)^{-1}$ and obtain a homomorphism $g$ from $\rels{A}$ to $\rels{B}$ 
such that $\pi\circ g=\id$. 
By Lemma~\ref{lemma:grohe} we obtain that $G$ contains a $k$-clique.

\medskip

Finally, note that we can make $\structB$ in the reduction to be finite-valued.  We only need to 
replace $\infty$ by a sufficiently large $N$ in its definition. More precisely, 
we need $\costb{\structA\mapsto\structB}{h}>M^*$, whenever $h:A\mapsto B$ is not a homomorphism from $\rels{A}$ to $\rels{B}$, 
i.e., we need $f^{\structA}(\tuple{x})\cdot N > M^*$ for any  $(f,\tuple{x})\in \tup{\structA}_{>0}$. 
(We can take  $N=1+(M^*/\min\{f^{\structA}(\tuple{x}): (f,\tuple{x})\in \tup{\structA}_{>0}\})$.)
\end{proof}

\section{Power of Sherali-Adams for VCSP($\C$, $-$)}
\label{sec:sa}

In this section we will present and prove our second main result,
Theorem~\ref{thm:sa}. First, in Section~\ref{subsec:sa}, we will define the
Sherali-Adams LP relaxation for VCSPs and state Theorem~\ref{thm:sa}.
Second, in Section~\ref{subsec:tw}, we will define the key concept of
\emph{treewidth modulo scopes}, which essentially captures the applicability of the
Sherali-Adams for VCSP($\C$, $-$). Section~\ref{subsec:suff} proves the
sufficiency part of Theorem~\ref{thm:sa}, whereas Sections~\ref{subsec:nec-tw}
and~\ref{subsec:overlap} prove the necessity part of Theorem~\ref{thm:sa}.

\subsection{Sherali-Adams LP Relaxations}
\label{subsec:sa}

Before we define the Sherali-Adams LP relaxation formally, we give an informal
explanation that might be helpful to the reader who has not seen this before.
Any VCSP instance $(\structA,\structB)$ over signature $\sigma$ has a natural
integer linear programming formulation (ILP) with $0/1$-variables.
The variables of the ILP are pairs $(f^\structA,\tuple{x})$, where $f^\structA$ is a ``constraint''
and $\tuple{x}\in A^{\ar{f}}$ is a ``local assignment to the variables on which
the constraint $f^\structA$ depends'', denoted by $s$ in Figure~\ref{fig:sa}.
The ILP has several constraints. Firstly,
local assignments with infinite costs are forbidden by setting the corresponding
ILP variables to $0$. Secondly, the integrality constraints (i.e., having $0/1$
variables) ensure that for every constraint exactly one assignment is selected.
Finally, ``marginalisation constraints'' ensure that the local assignments
are consistent in the sense that if $x\in A$ is mapped by a solution to $b\in B$
then it is the same $b$ independently of the constraint that depends on $x$. The
objective function of the ILP is the cost of an assignment given by
the ILP variables.

By relaxing the integrality constraints, one obtains the so-called basic
LP relaxation. One can think of the LP variables, for every constraint, as a
probability distribution over the set of local assignments. Looking at
Figure~\ref{fig:sa}, conditions \textcolor{red}{\textbf{(SA2)}} and 
\textcolor{red}{\textbf{(SA4)}} ensure that the variables are indeed probability
distributions. Conditions \textcolor{red}{\textbf{(SA3)}} correspond to the
feasibility part (the underlying CSP) and rule out locally-infeasible assignments.
Conditions \textcolor{red}{\textbf{(SA1)}} ensure consistency between pairs of
constraints, called ``marginalisation constraints'' above. Finally, the
objective function of the LP is the expected cost of an assignment given by the
LP variables with respect to the local probability distributions.

The idea of the $k$-th level of the Sherali-Adams relaxation is
to strengthen the LP by introducing constraints on all sets of variables of size
up to $k$ and imposing (local) consistency among these constraints. A further
(notational) complication in the definition is the possible repetition of a
variable in a constraint scope, which is dealt with using the $\toset{}$
operator, which now define.

Given a tuple $\tuple{x}$, we write $\toset{\tuple{x}}$ to denote the set of
elements appearing in $\tuple{x}$. 
Let $(\structA,\structB)$ be an instance of the VCSP over a signature $\sigma$ and $k\geq 1$. We define a new signature $\sigma_k =\sigma \cup \{ \rho_k \}$, where $\rho_k$ is a new function symbol of arity $k$. Then, we create from $(\structA,\structB)$ an instance $(\structA_k,\structB_k)$ over $\sigma_k$ such that $A_k=A$, $B_k=B$, $\rho_k^{\structA_k}(\tuple{x}) = 1$ for any $\tuple{x} \in A_k^k$, $\rho_k^{\structB_k}(\tuple{x}) = 0$ for any $\tuple{x} \in B_k^k$, and for every $f\in\sigma$ we have $f^{\structA_k}=f^{\structA}$ and $f^{\structB_k}=f^{\structB}$.
Because
the new function $\rho_k$ is identically zero in $\structB_k$, we have that for
any mapping $h: A \mapsto B$, $\cost{h}$ is the same in both instances
$(\structA,\structB)$ and $(\structA_k,\structB_k)$. 
The \emph{Sherali-Adams relaxation of level $k$}~\cite{Sherali1990} of $(\structA,\structB)$,
denoted by SA$_k(\structA,\structB)$, is the linear program given in
Figure~\ref{fig:sa}, which has
one variable $\lambda(f,\tuple{x},s)$ for each $(f,\tuple{x}) \in
\tup{\structA_k}_{>0}$ and $s: \toset{\tuple{x}} \mapsto B_k$.

\begin{figure}
\begin{align*}
&\min \sum_{\substack{(f,\tuple{x}) \in \tup{\structA_k}_{>0}, \; s: \toset{\tuple{x}} \mapsto B_k \\ f^{\structA_k}(\tuple{x}) \times f^{\structB_k}(s(\tuple{x})) < \infty}} \lambda(f,\tuple{x},s) f^{\structA_k}(\tuple{x}) f^{\structB_k}(s(\tuple{x})) &\\[3\jot]
&\lambda(f,\tuple{x},s) = \sum_{r:\toset{\tuple{y}} \mapsto B_k,
r|_{\toset{\tuple{x}}}=s} \lambda(p,\tuple{y},r) ~~ \textcolor{red}{\textbf{(SA1)}} & \hspace{-10mm} \forall (f, \tuple{x}), (p, \tuple{y}) \in \tup{\structA_k}_{>0}:\\[-5\jot]
&~& \hspace{-50mm} 
\text{$\toset{\tuple{x}} \subseteq \toset{\tuple{y}}$ and $|\toset{\tuple{x}}|\leq k$}; \\ 
&~& \forall s: \toset{\tuple{x}} \mapsto B_k\\[3\jot]
&\sum_{s: \toset{\tuple{x}} \mapsto B_k} \lambda(f,\tuple{x},s) = 1 ~~ \textcolor{red}{\textbf{(SA2)}} & \forall (f,\tuple{x}) \in \tup{\structA_k}_{>0}\\
&\lambda(f,\tuple{x},s) = 0 ~~ \textcolor{red}{\textbf{(SA3)}} &\hspace{-50mm}
\forall (f,\tuple{x}) \in \tup{\structA_k}_{>0},s: \toset{\tuple{x}} \mapsto B_k: f^{\structA_k}(\tuple{x}) \times f^{\structB_k}(s(\tuple{x})) = \infty\\
&\lambda(f,\tuple{x},s) \geq 0 ~~ \textcolor{red}{\textbf{(SA4)}} &
\hspace{-20mm} \forall (f,\tuple{x}) \in \tup{\structA_k}_{>0}, s: \toset{\tuple{x}} \mapsto B_k
\end{align*}
\caption{The Sherali-Adams relaxation SA$_k(\structA,\structB)$ of level $k$ of $(\structA,\structB)$.} 
\label{fig:sa}
\end{figure}
The variables are indexed not only by $\tuple{x}$ and $s$
but also by $f$. This would not be necessary if $k\geq \ar{f}$ 
but we are also interested in the case of $k<\ar{f}$. 

\begin{definition}
\label{def:improvek}
Let $\structA, \structB$ be valued $\sigma$-structures and $k\geq 1$. 
\begin{itemize}
\item We denote by $\optfrac{k}{\structA,\structB}$ the minimum cost of a solution to SA$_k(\structA,\structB)$.
\item We write $\structA \lessfrac \structB$ if  $\optfrac{k}{\structA,\structC} \leq \optfrac{k}{\structB,\structC}$ for all valued $\sigma$-structures $\structC$. 
\end{itemize}
\end{definition}

The proof of the following can be found in Appendix~\ref{app:lessfrac}.

\begin{proposition}
\label{prop:lessfrac}
Let $\structA, \structB$ be valued $\sigma$-structures and $k\geq 1$. If there exists an IFH from $\structA$ to $\structB$, then $\structA \lessfrac \structB$.
\end{proposition}

We also have the following.

\begin{proposition}
\label{prop:corefrac}
Let $\structA$ be a valued $\sigma$-structure, $\structA'$ be the core of $\structA$ and $k\geq 1$. Then, 
$\optfrac{k}{\structA,\structC} = \optfrac{k}{\structA',\structC}$, for all valued $\sigma$-structures $\structC$. 
\end{proposition}

\begin{proof}
Since $\structA \equiv \structA'$, by Proposition~\ref{prop:charfrac}, there
exist IFHs from $\structA$ to $\structA'$ and from
$\structA'$ to $\structA$. Therefore, by Proposition~\ref{prop:lessfrac},
$\optfrac{k}{\structA,\structC} = \optfrac{k}{\structA',\structC}$ for all
valued $\sigma$-structures $\structC$.
\end{proof}

Given a valued $\sigma$-structure $\structA$, the \emph{overlap} of $\structA$
is the largest integer $m$ such that there exist $(f,\tuple{x}), (p,\tuple{y})
\in \tup{\structA}_{>0}$ with $(f,\tuple{x}) \neq (p,\tuple{y})$ and $|\toset{\tuple{x}} \cap \toset{\tuple{y}}| = m$.

The following is our second main result; $\twms{\structA}$ is defined in
Section~\ref{subsec:tw} and Theorem~\ref{thm:sa} is implied by Theorems~\ref{thm:suff-core}, \ref{thm:nece-core-core}, and~\ref{thm:nece-scopes-core}
proved in Sections~\ref{subsec:suff}, \ref{subsec:nec-tw},
and~\ref{subsec:overlap}, respectively.

\begin{theorem}[\textbf{Power of Sherali-Adams}]
\label{thm:sa}
Let $\structA$ be a valued $\sigma$-structure and let $k\geq 1$. Let $\structA'$ be the core of $\structA$. The
Sherali-Adams relaxation of level $k$ is always tight for $\structA$, i.e.,
for every valued $\sigma$-structure $\structB$, we have that
$\optfrac{k}{\structA,\structB}=\opt{\structA,\structB}$, \emph{if and only if} (i)
$\twms{\structA'}\leq k-1$ and (ii) the overlap of $\structA'$ is at most $k$.
\end{theorem}

We remark that although Theorem~\ref{thm:sa} deals with valued
structures, it also shows the same result for finite-valued structures, where all
functions are restricted to finite  values in $\qplus$.
In particular, the sufficiency part of Theorem~\ref{thm:sa}, i.e., if the core
satisfies conditions (i) and (ii) then the $k$-th level of Sherali-Adams is tight, 
applies directly to the finite-valued case. 
It will follow from the proofs of Theorems~\ref{thm:nece-core} and~\ref{thm:nece-scopes} that 
whenever the core violates condition (i) or (ii) then the $k$-th level of Sherali-Adams is not tight even for finite-valued structures. 
Hence, Theorem~\ref{thm:sa} also characterises the tightness of Sherali-Adams for finite-valued VCSPs.

Let us note that the characterisation given by Theorem~\ref{thm:sa} for levels
$k\geq r$, where $r$ is the arity of the signature of $\structA$, boils down to
the notion of treewidth. That is, if $k\geq r$, the $k$-th level of Sherali-Adams
is tight if and only if the treewidth of the core of $\structA$ is at most
$k-1$. Interestingly, Theorem~\ref{thm:sa} tells us \emph{precisely} under which
conditions the $k$-th level works even for $k<r$. 

Finally, we remark that the tractability part of Theorem~\ref{theo:main}, which
we obtain as a corollary of Theorem~\ref{thm:sa}, does not immediately follow
from a naive algorithm that would compute, for $\structA \in \C$, $\opt{\structA,\structB}$ using dynamic programming along a tree decomposition of the
core $\structA'$ of $\structA$. Such an algorithm would first need to compute
$\structA'$, and it is not clear that it can be done in polynomial time, even
with the promise that $\structA'$ has bounded treewidth. (The situation is
different for relational structures, where this promise problem is known to be
solvable in polynomial time~\cite[Lemma~25]{Chen15:icdt}.) Theorem~\ref{thm:sa} gives a way to circumvent this issue since the linear program SA$_k(\structA,\structB)$ does not depend on $\structA'$ in any way.

\subsection{Treewidth modulo scopes}
\label{subsec:tw}

Let $\rels{A}$ be a relational structure with universe $A$ over a relational
signature $\tau$.  Recall from Section~\ref{sec:prelims} that $G(\rels{A})$
denotes the Gaifman graph of $\rels{A}$. A \emph{scope} of $G(\rels{A})$ is a
set $X$ for which there is relation symbol $R\in \tau$ and a tuple $\tuple{x}\in
R^{\rels{A}}$ such that $X=\toset{\tuple{x}}$. In other words, the scopes of
$G(\rels{A})$ are the sets that appear precisely in the tuples of
$\rels{A}$.\footnote{In a (V)CSP instance, the term \emph{scope} usually refers to the
list of variables a (valued) constraint depends on.}
Observe that every scope $X$ of $G(\rels{A})$ induces a clique in $G(\rels{A})$.  
\begin{definition}
Let $\rels{A}$ be a relational structure and $G(\rels{A})$ its Gaifman graph. 
Let $(T,\beta)$ be a tree decomposition of $G(\rels{A})$, where $T=(V(T),E(T))$. 
The \emph{width modulo scopes} of $(T,\beta)$ is defined by 
\[\max\{|\beta(t)|-1\mid \text{$t\in V(T)$ and $\beta(t)$ is not a scope of $G(\rels{A})$}\}.\]
If $\beta(t)$ is a scope for all nodes $t\in V(T)$ then we set the width modulo scopes of $(T,\beta)$ to be $0$. 
The \emph{treewidth modulo scopes} of $G(\rels{A})$, denoted by $\twms{G(\rels{A})}$, is the 
minimum width modulo scopes over all its tree decompositions. 
The treewidth modulo scopes of $\rels{A}$ is $\twms{\rels{A}}=\twms{G(\rels{A})}$. 
For a valued structure $\structA$, we define the treewidth modulo
scopes of $\structA$ as $\twms{\structA}=\twms{\pos{\structA}}$.
\end{definition}

Unlike treewidth, the notion of treewidth modulo scopes is not \emph{monotone}, i.e., it can increase after taking substructures.  
To see this, take for instance the relational structure $\rels{A}$ that corresponds to the undirected $k\times k$ grid. 
We have $\twms{\rels{A}}=k$. However, adding a new relation with only one tuple containing all elements of $\rels{A}$ 
lowers the treewidth modulo scopes to $0$. 
Let us also remark that the relational structures with treewidth modulo scopes $0$ are precisely the relational structures 
whose underlying hypergraphs are $\alpha$-acyclic (see e.g.~\cite{GGLS16:pods}). 

\medskip

Let $G=(V,E)$ be a graph. A bramble $\B$ of $G$ is a collection of subsets of $V$ such that (i) each $B\in \B$ is a connected set, and 
(ii) every pair of sets $B,B'\in \B$ touch, i.e., they have a vertex in common or $G$ contains an edge between them. 
A subset of $V$ is a cover of $\B$ if it intersects every set in $\B$. There is a well-known connection between treewidth and brambles. 

\begin{theorem}[\cite{Seymour93:graph}]
Let $G$ be a graph and $k\geq 1$. Then the treewidth of $G$ is at most $k$ if and only if any bramble in $G$ can be covered by a set of size at most $k+1$.
\end{theorem}

We show an analogous characterisation for treewidth modulo scopes, which will be
important later in the proof of Theorem~\ref{thm:nece-core}.

\begin{theorem}
\label{theo:char-twms}
Let $\rels{A}$ be a relational structure and $k\geq 0$. Then $\twms{\rels{A}}\leq k$ if and only if any bramble in $G(\rels{A})$ can be covered by a set of size at most $k+1$ or by a scope in $G(\rels{A})$. 
\end{theorem}

\begin{proof}
The proof is an adaptation of the proof of~\cite[Theorem~12.3.9]{diestel10:graph}.
Suppose first that $\twms{\rels{A}}\leq k$ and let $(T,\beta)$ be a tree decomposition of $G(\rels{A})$ of 
width modulo scopes at most $k$, where $T=(V(T),E(T))$. 
Let $\B$ be any bramble of $G(\rels{A})$. We show that there is $t\in V(T)$ such that $\beta(t)$ covers $\B$. 
If there is an edge $\{t_1,t_2\}$ of $T$ such that $\beta(t_1)\cap \beta(t_2)$ covers $\B$, then we are done. 
Otherwise, we can define an orientation for each edge $\{t_1,t_2\}$ of $T$ as follows. 
Let $X \defeq \beta(t_1)\cap \beta(t_2)$ and $\B_X=\{B\in\B\mid X\cap B=\emptyset\}$. By assumption, $X$ does not cover $\cal B$ and then $\B_X$ is not empty. 
If we remove the edge $\{t_1,t_2\}$ from $T$, we obtain exactly two connected components $T_1=(V(T_1),E(T_1))$ and $T_2=(V(T_2),E(T_2))$ containing $t_1$ and $t_2$, respectively. 
Let $U_1 \defeq \bigcup_{t\in V(T_1)}\beta(t)$ and $U_2 \defeq \bigcup_{t\in V(T_2)}\beta(t)$. 
It is a well-known property of tree decompositions (e.g., \cite[Lemma~12.3.1]{diestel10:graph}) that $X$ separates $U_1$ and $U_2$ in $G(\rels{A})$. 
Take $B\in \B_X$. Since $B$ is connected, it follows that $B\subseteq U_i\setminus X$, for some $i\in \{1,2\}$. 
Since all sets in $\B_X$ mutually touch, we can choose the same $i\in \{1,2\}$, for all $B\in \B_X$. We then 
orient $\{t_1,t_2\}$ towards $t_i$. 

Let $t$ be a sink node in this orientation of $T$, i.e., $t$ has only incoming arcs (note that there must be at least one sink). 
We claim that $\beta(t)$ covers $\B$. 
Towards a contradiction, suppose there exists $B\in \B$ with $B\cap \beta(t)=\emptyset$. Take an arbitrary $b\in B$ and a node $t_b\in V(T)$ such that $b\in \beta(t_b)$. 
Let $\{t',t\}$ be the last edge in the unique simple path from $t_b$ to $t$. 
Observe that $(\beta(t')\cap\beta(t))\cap B= \emptyset$, and hence $B\in \B_{\beta(t')\cap\beta(t)}$. 
Since the edge $\{t',t\}$ is oriented towards $t$, we have that $B\subseteq U_t\setminus \beta(t')\cap\beta(t)$, 
where $U_t=\bigcup_{t\in V(T_t)}\beta(t)$, and $T_t=(V(T_t),E(T_t))$ is the connected component of $T-\{t',t\}$ containing $t$. 
In particular, there is $t^*\in V(T_t)$ such that $b\in \beta(t^*)$. By the connectedness property 
of tree decompositions, we have that $b\in \beta(t)$; a contradiction. 

We now prove the other direction. Assume that any bramble in $G(\rels{A})$ can be covered by a set of size at most $k+1$ or by a scope. 
We say that a set of vertices  satisfying this property is \emph{good}; otherwise it is \emph{bad}.
Let $G'$ be a subgraph of $G(\rels{A})$ and $\B$ be a bramble of $G(\rels{A})$. 
A \emph{$\B$-admissible} tree decomposition for $G'$, 
is a tree decomposition $(T,\beta)$ of $G'$ with $T=(V(T),E(T))$ such that every $\beta(t)$, with $t\in V(T)$, that is a bad set of $G(\rels{A})$ does not cover $\B$. 
We will show that for every bramble $\B$ of $G(\rels{A})$, there is a
  $\B$-admissible tree decomposition for $G(\rels{A})$. 
Since every set covers the empty bramble $\B=\emptyset$, the result follows as any $\B$-admissible tree decomposition for $\B=\emptyset$ 
has only good bags, and then its width modulo scopes is at most $k$. 

We now fix a bramble $\B$ of $G(\rels{A})$ and assume inductively that $G(\rels{A})$ has a $\B'$-admissible tree decomposition for every bramble $\B'$
with more sets than $\B$. (The induction starts, since every bramble of $G(\rels{A})$ has at most $2^{|V(G(\rels{A}))|}$ sets.) By hypothesis, we can take a good set $X$ covering $\B$. 
We show the following:

\medskip
\noindent
For every connected component $C=(V(C),E(C))$ of $G(\rels{A})-X$, 
there is a $\B$-admissible tree decomposition of $G(\rels{A})[X\cup V(C)]$ (i.e., the subgraph of $G(\rels{A})$ induced by $X\cup V(C)$) 
which has $X$ as one of its bags. 
\medskip

These decompositions can be glued to define a $\B$-admissible tree decomposition of $G(\rels{A})$ as required.

Let $C=(V(C),E(C))$ be a fixed connected component of $G(\rels{A})-X$. 
Since $X\cap V(C)=\emptyset$ and $X$ covers $\B$, it follows that $V(C)\not\in \B$. 
Then if we define $\B' \defeq \B\cup\{V(C)\}$, we have that $|\B'|>|\B|$. 
Suppose first that $\B'$ is not a bramble. (Note that in the base case of the induction we always have this case.)
This means that $V(C)$ fails to touch a set $B^*\in \B$. 
Let $N(C)$ be the union of $V(C)$ and all vertices in $G(\rels{A})$ adjacent to some vertex in $V(C)$. 
We have $N(C)\cap B^*=\emptyset$, and hence $N(C)$ does not cover $\B$. 
We also have $N(C)\subseteq X\cup V(C)$. 
Then our desired $\B$-admissible tree decomposition for $G(\rels{A})[X\cup V(C)]$ consists of  two adjacent bags $X$ and $N(C)$.

We now assume that $\B'$ is a bramble. 
By induction, there is a $\B'$-admissible tree decomposition $(T,\beta)$ for $G(\rels{A})$ with $T=(V(T),E(T))$. 
We start by considering the case when $|X|\leq k+1$. Without loss of generality, 
we can assume that $X$ is a cover of $\B$ of minimum size. In this case we can argue as in the proof of~\cite[Theorem~12.3.9]{diestel10:graph}. 
Let $\ell \defeq |X|\leq k+1$. If $(T,\beta)$ is $\B$-admissible then we are done, so we assume that there exists 
$s\in V(T)$ such that $\beta(s)$ is bad and covers $\B$. 
Every separator of two covers of a bramble is also a cover of that bramble~\cite[Lemma~12.3.8]{diestel10:graph}.
It follows that the minimum size of a separator of $X$ and $\beta(s)$ is $\ell$. 
By Menger's theorem~\cite[Theorem~3.3.1]{diestel10:graph}, there exist $\ell$
disjoint paths $P_1,\dots,P_\ell$ linking $X$ and $\beta(s)$ such that
each $P_i$ intersects $X$ and $\beta(s)$ exactly in its initial node and final node, respectively. 
We denote by $x_i$ the initial node of $P_i$. We have $X=\{x_1,\dots,x_\ell\}$. 
Observe that since $(T,\beta)$ is $\B'$-admissible, the bag $\beta(s)$ fails to cover $\B'$, that is, $\beta(s)\cap V(C)=\emptyset$. 
It follows that the path $P_i$ intersects $X\cup V(C)$ exactly in its initial node $x_i$.

 Let $(T,\beta')$ be the restriction of $(T,\beta)$ to the set of nodes $X\cup V(C)$, i.e., 
 $\beta'(t)=\beta(t)\cap (X\cup V(C))$, for all $t\in V(T)$. 
 The pair $(T,\beta')$ is a tree decomposition of $G(\rels{A})[X\cup V(C)]$. 
 The desired tree decomposition $(T,\beta'')$ for $G(\rels{A})[X\cup V(C)]$ is the result of adding 
 some of the nodes $x_i$'s to some particular bags of $(T,\beta')$. 
Let us fix for each $x_i$ a node $t_i\in V(T)$ such that $x_i\in \beta(t_i)$.  
 Formally, $(T,\beta'')$ is given by 
 \[\beta''(t)=\beta'(t)\cup\{x_i\mid \text{$t$ is in the unique simple path from $t_i$ to $s$ in $T$}\}.\]

Observe that $(T,\beta'')$ still satisfies the connectedness property of tree decompositions. 
Also, observe that $|\beta(t)|\geq |\beta''(t)|$, for all $t\in V(T)$. 
Indeed, for each $x_i\in \beta''(t)\setminus \beta(t)$, with $i \in \{1,\dots,\ell\}$, the node $t$ is in the unique simple path 
from $t_i$ to $s$ in $T$. 
It follows that $\beta(t)$ contains a node $w_i$ from the path $P_i$ that is different from $x_i$.
Since $P_i$ meets $X\cup V(C)$ only in $x_i$, we have that  $w_i\not\in\beta''(t)$. 
The claim follows since all the $w_i$'s are distinct. 

Since $\beta''(s)=X$, it remains to show that 
$(T,\beta'')$ is $\B$-admissible. 
Pick $t\in V(T)$ such that $\beta''(t)$ is bad, that is, $|\beta''(t)|>k+1$ and $\beta''(t)$ is not a scope. 
We need to show that $\beta''(t)$ does not cover $\B$. 
We claim first that $\beta(t)$ is also a bad set. 
By contradiction, suppose that $\beta(t)$ is good. 
As $|\beta(t)|\geq |\beta''(t)|>k+1$, we have that $\beta(t)$ is a scope. 
Since $|\beta(t)|\geq |\beta''(t)|>k+1$ and $|X|=\ell\leq k+1$, it is the case that $\beta(t)\cap V(C)\neq \emptyset$.
It follows that $\beta(t)\subseteq X\cup V(C)$ as $\beta(t)$ is a clique in $G(\rels{A})$. 
Consequently, $\beta(t)\subseteq \beta''(t)$, and since $|\beta(t)|\geq |\beta''(t)|$, 
we conclude that $\beta(t)=\beta''(t)$. But $\beta''(t)$ is not a scope, which is a contradiction. 

Since $(T,\beta)$ is $\B'$-admissible and $\beta(t)$ is bad, we have that 
there is $B^*\in \B'$ such that $\beta(t)\cap B^*=\emptyset$. 
As discussed in the previous paragraph, we know that $\beta(t)\cap V(C)\neq
\emptyset$. 
It follows that $B^*\in \B$. 
We claim that $\beta''(t)\cap B^*=\emptyset$, 
which implies that  $\beta''(t)$ does not cover $\B$, as required. 
Towards a contradiction, suppose that $x_i\in B^*$, for $x_i\in \beta''(t)\setminus \beta(t)$, 
with $i\in \{1,\dots,\ell\}$. 
By definition, $t$ is in the unique simple path from $t_i$ to $s$ in $T$. 
Also, by definition, $\beta(s)$ covers $\B$, and then there is $b\in B^*\cap \beta(s)$. 
Since $B^*$ is connected, there is a path $P$ from $x_i$ to $b$ in $G(\rels{A})$ whose nodes belong to $B^*$. 
We have that $\beta(t)$ must contain a node from $P$, and then a node from $B^*$, which is a contradiction.

It remains to consider the case  when $|X|>k+1$, and hence $X$ is a scope. Without loss of generality, we assume that $X$ is a maximal scope covering $\B$. 
 Again, let $(T,\beta')$ be the restriction of $(T,\beta)$ to the set of nodes $X\cup V(C)$.
  For $t\in V(T)$, we say that $\beta'(t)$ is \emph{maximal} if there is no 
 $t'\in V(T)$ such that $\beta'(t)\subsetneq \beta'(t')$. 
 We define $(\tilde{T},\tilde{\beta})$, where $\tilde{T}=(V(\tilde{T}),E(\tilde{T}))$, to be a decomposition of $G(\rels{A})[X\cup V(C)]$
 whose bags are precisely the maximal bags of $(T,\beta')$, or more formally, 
 
 \[\{\tilde{\beta}(t)\mid t\in V(\tilde{T})\}=\{\beta'(t)\mid \text{$t\in V(T)$ and $\beta'(t)$ is maximal}\}.\]

 In order to obtain $(\tilde{T},\tilde{\beta})$ we can iteratively remove non-maximal bags from $(T,\beta')$ as follows: 
 if $\beta'(t)$ is not maximal and it is strictly contained in $\beta'(t')$, for $t,t'\in V(T)$, and $\{t,t''\}$ is the first edge in the unique simple path from $t$ 
 to $t'$ in $T$, then remove $\beta'(t)$ by contracting the edge $\{t,t''\}$ into a new node $s$ defining $\beta'(s)=\beta'(t'')$.

 We claim that $(\tilde{T},\tilde{\beta})$ satisfies the required conditions. 
Since $X$ is a scope and then a clique in $G(\rels{A})$, 
 there is a node $t_X\in V(T)$ such that $X\subseteq \beta'(t_X)$. 
We show first that there is no $t\in V(T)$ such that $X\subsetneq \beta'(t)$. 
Towards a contradiction, suppose such a $t$ exists. 
Note that $X\subsetneq \beta'(t)\subseteq \beta(t)$. 
This implies that $\beta(t)$ covers $\B'$. Moreover, since $|X|>k+1$ and $X$ is a maximal scope, 
we deduce that $\beta(t)$ is a bad set. This contradicts the $\B'$-admissibility of $(T,\beta)$ and the claim follows. 
Using this, we obtain that $\beta'(t_X)=X$ and moreover, $\beta'(t_X)$ is a maximal bag for  $(T,\beta')$. 
By construction of $(\tilde{T},\tilde{\beta})$, there exists $\tilde{t}\in V(\tilde{T})$ 
such that $\tilde{\beta}(\tilde{t})=X$. 
It only remains to show that $(\tilde{T},\tilde{\beta})$ is $\B$-admissible.

Suppose $\tilde{\beta}(t)$ is a bad set for a node $t\in V(\tilde{T})$. 
We prove first that $\tilde{\beta}(t)\cap V(C)\neq \emptyset$. 
By contradiction, assume that $\tilde{\beta}(t)\subseteq X$. 
It follows that $\tilde{\beta}(t)=X$, otherwise $\tilde{\beta}(t)$ would not be a maximal bag of $(T,\beta')$. 
But $X$ is a scope and then $\tilde{\beta}(t)$ cannot be bad. This is a contradiction. 
Let $t^*\in V(T)$ be a node such that $\beta'(t^*)=\tilde{\beta}(t)$. 
We claim that $\beta(t^*)$ is a bad set. Since $\beta'(t^*)\subseteq \beta(t^*)$, 
we know that $|\beta(t^*)|\geq |\beta'(t^*)|>k+1$. 
For the sake of contradiction, suppose that $\beta(t^*)$ is a scope. 
In particular, $\beta(t^*)$ is a clique in $G(\rels{A})$. 
Since $\beta'(t^*)\cap V(C)\neq \emptyset$, and hence $\beta(t^*)\cap V(C)\neq \emptyset$, 
we have that $\beta(t^*)\subseteq X\cup V(C)$. Therefore, $\beta'(t^*)=\beta(t^*)$. 
This is a contradiction since $\beta'(t^*)$ is bad and then not a scope. 
Thus $\beta(t^*)$ is a bad set. As $(T,\beta)$ is $\B'$-admissible, 
there is $B^*\in \B'$ with $\beta(t^*)\cap B^*=\emptyset$. 
Since $\beta(t^*)\cap V(C)\neq\emptyset$, it follows that $B^*\in \B$. 
Finally, since $\beta'(t^*)\subseteq \beta(t^*)$, we have that $\beta'(t^*)\cap B^*=\emptyset$. 
Hence, $\beta'(t^*)=\tilde{\beta}(t)$ does not cover $\B$. 
We conclude that $(\tilde{T},\tilde{\beta})$ is $\B$-admissible. 
\end{proof}

\subsection{Sufficiency}
\label{subsec:suff}

We show the following. 

\begin{theorem}
\label{thm:suff}
Let $\structA$ be a valued $\sigma$-structure and let $k\geq 1$. Suppose that (i) $\twms{\structA}\leq k-1$ and (ii) the overlap of $\structA$ is at most $k$. 
  Then the Sherali-Adams relaxation of level $k$ is always tight for $\structA$, i.e., 
for every valued $\sigma$-structure $\structB$, we have that $\optfrac{k}{\structA,\structB}=\opt{\structA,\structB}$. 
\end{theorem}

\begin{proof}
Let $\structB$ be an arbitrary valued $\sigma$-structure with universe $B$. 
Let $\rels{A}\defeq\pos{\structA}$ and let $(T,\beta)$ be a tree decomposition of the Gaifman graph $G(\rels{A})$ of width modulo scopes $\leq k-1$, 
where $T=(V(T),E(T))$. 
As usual, we denote by $A$ the universe of $\structA$, $\rels{A}$ and $G(\rels{A})$. 
Recall that the solutions for SA$_k(\structA,\structB)$ are indexed by the set
\begin{align*}
{\I}\defeq &\{(f,\tuple{x},s): (f,\tuple{x})\in \tup{\structA_k}_{>0}, s:\toset{\tuple{x}}\mapsto B_k\}\\
=&\{(f,\tuple{x},s): (f,\tuple{x})\in \tup{\structA}_{>0}, s:\toset{\tuple{x}}\mapsto B\}\cup\{(\rho_k,\tuple{x},s): \tuple{x}\in A^k, s:\toset{\tuple{x}}\mapsto B\}.
\end{align*}
Let $\btuples\subseteq \tup{\structA_k}_{>0}$ be the set 
\[\btuples\defeq\{(f,\tuple{x})\in \tup{\structA_k}_{>0}\mid \text{$\toset{\tuple{x}}\subseteq \beta(t)$ for some $t\in V(T)$}\}.\]
We have $\tup{\structA}_{>0}\subseteq \btuples$.  
We define $\T\defeq\{(f,\tuple{x},s)\in \I: (f,\tuple{x})\in \btuples\}$. 
Let $P(\structA,\structB)$ be the system of linear inequalities given by the constraints of SA$_k(\structA,\structB)$ restricted to the variables 
indexed by $\T$. 
More precisely, $P(\structA,\structB)$ is the following system over variables $\{\lambda(f,\tuple{x},s): (f,\tuple{x},s)\in \T\}$: 

\begin{align}
&\lambda(f,\tuple{x},s) = \sum_{r:\toset{\tuple{y}} \mapsto B,\, 
r|_{\toset{\tuple{x}}}=s} \lambda(p,\tuple{y},r) & \hspace{-10mm} \qquad\qquad\qquad\quad \forall (f, \tuple{x}), (p, \tuple{y}) \in \tuples: \nonumber\\[-5\jot]
&~& \hspace{-50mm} 
 \text{$\toset{\tuple{x}} \subseteq \toset{\tuple{y}}$ and $|\toset{\tuple{x}}|\leq k$}; \nonumber \\ 
&~&  \forall s: \toset{\tuple{x}} \mapsto B \label{eq:p1} \\[3\jot]
&\sum_{s: \toset{\tuple{x}} \mapsto B} \lambda(f,\tuple{x},s) = 1  & \forall (f,\tuple{x}) \in \tuples \label{eq:p2}\\
&\lambda(f,\tuple{x},s) = 0  &\hspace{-50mm}
\forall (f,\tuple{x},s)\in \T,\, f^{\structA_k}(\tuple{x}) \times f^{\structB_k}(s(\tuple{x})) = \infty \label{eq:p3}\\
&\lambda(f,\tuple{x},s) \geq 0  &
\hspace{-20mm} \forall (f,\tuple{x},s)\in \T \label{eq:p4}
\end{align}

From the definition of the Sherali-Adams hierarchy, we have
$\optfrac{k}{\structA,\structB}\leq\opt{\structA,\structB}$. We need to prove that $\opt{\structA,\structB}\leq
\optfrac{k}{\structA,\structB}$. Let $\lambda=\{\lambda(f,\tuple{x},s):
(f,\tuple{x},s)\in \I\}$ be an optimal solution to SA$_k(\structA,\structB)$.
Let $c=\{c(f,\tuple{x},s): (f,\tuple{x},s)\in\I\}$ be the vector defining the
objective function of SA$_k(\structA,\structB)$.  Consider the projection
$\lambda|_{\T}=\{\lambda(f,\tuple{x},s): (f,\tuple{x},s)\in \T\}$ of $\lambda$
to $\T$. Similarly, consider the projection $c|_{\T}$ of $c$ to $\T$. 
The restriction of $c$ to $\I\setminus \T$ is the vector $0$, and hence,
$c|_{\T}\cdot \lambda|_{\T}=c\cdot \lambda$. Also, $\lambda|_{\T}$ is a
solution to $P(\structA,\structB)$. By Lemma~\ref{lemma:p-integral} proved
below, the polytope $P(\structA,\structB)$ is integral and thus $\lambda|_{\T}$
is a convex combination of integral solutions $I^{g_1},\dots,I^{g_n}$ of
$P(\structA,\structB)$, for assignments $g_i:A\mapsto B$.  It follows that there
exists $i\in \{1,\dots,n\}$ such that $c|_{\T}\cdot I^{g_i} \leq c|_{\T}\cdot
\lambda|_{\T}=c\cdot \lambda$. In particular, the cost of the assignment $g_i$
is $\leq c \cdot \lambda= \optfrac{k}{\structA,\structB}$. We conclude that
$\opt{\structA,\structB} \leq \optfrac{k}{\structA,\structB}$. 
\end{proof}

This is the last missing piece in the proof of Theorem~\ref{thm:suff}.

\begin{lemma}
\label{lemma:p-integral}
The polytope described by $P(\structA,\structB)$ is integral. 
\end{lemma}

\begin{proof}

We start with some notation. Let $\omega$ be a probability distribution over $W^X$. 
For $X'\subseteq X$, we define the \emph{marginal} of $\omega$, denoted by $\omega_{X'}$, to be the 
probability distribution over $W^{X'}$ defined by $\omega_{X'}(h)=\sum_{g:X\mapsto W,\, g|_{X'}=h} \omega(g)$. 
We have the following:

\begin{claim}
\label{claim:marginals}
Let $\omega$ and $\delta$ be probability distributions over $W^X$ and $W^Y$, respectively. Suppose $\omega_{X\cap Y}=\delta_{X\cap Y}$. 
Then there is a probability distribution $\epsilon$ over $W^{X\cup Y}$ with $\epsilon_X = \omega$ and $\epsilon_Y=\delta$. 
\end{claim}

\begin{proof}
We set $\epsilon(h)=\frac{\omega(h|_{X})\delta(h|_{Y})}{\omega(h|_{X\cap Y})}$ for every $h:X\cup Y\mapsto W$. ($\epsilon(h)=0$ if $\omega(h|_{X\cap Y})=0$.)
\renewcommand{\qedsymbol}{$\blacksquare$}\end{proof}

Let us fix a solution $\lambda$ of $P(\structA,\structB)$. 
We need to show that $\lambda$ is a convex combination of integral solutions. 
Since the width modulo scopes of the decomposition $(T,\beta)$ is $\leq k-1$, 
we know that either $|\beta(t)|\leq k$ or $\beta(t)$ is a scope of $G(\rels{A})$, for every $t\in V(T)$. 
It follows that for each $t\in V(T)$, we can find  $(f_t,\tuple{x}_t)\in \tuples$ such that $\toset{\tuple{x}_t}=\beta(t)$. 

For $t\in V(T)$, let $\omega^t$ be the probability distribution over $B^{\beta(t)}$ given by $\omega^t(h)=\lambda(f_t,\tuple{x}_t, h)$. 
Note that $(\star)$ for every edge $\{t,t'\}\in E(T)$, we have that $\omega^{t}_{\beta(t)\cap \beta(t')} = \omega^{t'}_{\beta(t)\cap \beta(t')}$. 
Indeed, since the overlap of $\structA$ is $\leq k$, we can pick a tuple $(\rho_k, \tuple{z})\in \tuples$ such that $\toset{\tuple{z}}=\beta(t)\cap\beta(t')$. 
By restriction (\ref{eq:p1}) in $P(\structA,\structB)$, we obtain that $\omega^{t}_{\beta(t)\cap \beta(t')} = \lambda(\rho_k, \tuple{z}, \cdot) =  \omega^{t'}_{\beta(t)\cap \beta(t')}$. 
We can then define a distribution $\epsilon$ over $B^A$ such that $\epsilon_{\beta(t)}=\omega^t$, for every $t\in V(T)$. To do this, we start at  a particular node $r\in V(T)$. 
We apply iteratively Claim~\ref{claim:marginals} to extend the current distribution ($\omega^r$ at the beginning) one bag at a time, until all bags are covered. 
Note that the hypothesis of Claim~\ref{claim:marginals} (same marginals) is always satisfied due to $(\star)$.

For $g:A\to B$, we define $I^g =\{I^g(f,\tuple{x},s): (f,\tuple{x},s)\in \T\}$ to be the $0/1$-vector such that 
$I^g(f,\tuple{x},s) = 1$ if and only if $g|_{\toset{\tuple{x}}} = s$. 
We show that $\lambda = \sum_{g\in\supp{\epsilon}} \epsilon(g) I^g$. 
Take $(f,\tuple{x},s)\in \T$. By definition of $\T$, there is $t\in V(T)$ such that $\toset{\tuple{x}}\subseteq \beta(t)=\toset{\tuple{x}_t}$. 
If $(f,\tuple{x})=(f_t,\tuple{x}_t)$, then 
\begin{align*}
\sum_{g\in\supp{\epsilon}} \epsilon(g) I^g(f_t,\tuple{x}_t, s) = \sum_{g\in\supp{\epsilon}, \, g|_{\beta(t)}=s} \epsilon(g) = \omega^t(s) = \lambda(f_t,\tuple{x}_t, s)
\end{align*}
and we are done. In case $(f,\tuple{x})\neq(f_t,\tuple{x}_t)$, since the overlap is $\leq k$, we have that $|\toset{\tuple{x}}|\leq k$, and by condition (\ref{eq:p1}) in $P(\structA,\structB)$, we obtain
\begin{align*}
\lambda(f,\tuple{x}, s) = \sum_{r:\beta(t)\mapsto B,\, r|_{\toset{\tuple{x}}}=s} \omega^t(r) = \epsilon_{\toset{\tuple{x}}}(s) =  \sum_{g\in\supp{\epsilon}} \epsilon(g) I^g(f,\tuple{x}, s). 
\end{align*}

It remains to show that $I^g$ is a solution to $P(\structA,\structB)$, for all $g\in\supp{\epsilon}$. 
Conditions (\ref{eq:p1}), (\ref{eq:p2}) and (\ref{eq:p4}) hold by construction. 
For condition (\ref{eq:p3}), take $(f,\tuple{x},s)\in \T$ such that $f^{\structA_k}(\tuple{x}) \times f^{\structB_k}(s(\tuple{x})) = \infty$. 
Since $0=\lambda(f,\tuple{x},s) = \sum_{g\in\supp{\epsilon}} \epsilon(g) I^g(f,\tuple{x},s)$, it follows that $I^g(f,\tuple{x},s)=0$, 
for all $g\in \supp{\epsilon}$, 
and hence condition (\ref{eq:p3}) is satisfied. 
\end{proof}

\medskip

\begin{theorem}
\label{thm:suff-core}
Let $\structA$ be a valued $\sigma$-structure and $\structA'$ be its core. 
Suppose that (i) $\twms{\structA'}\leq k-1$ and (ii) the overlap of $\structA'$ is at most $k$.
Then the Sherali-Adams relaxation of level $k$ is always tight for $\structA$, i.e., 
for every valued $\sigma$-structure $\structB$, we have that $\optfrac{k}{\structA,\structB}=\opt{\structA,\structB}$. 
\end{theorem}

\begin{proof}
Let $\structB$ be a valued $\sigma$-structure. We can apply
Theorem~\ref{thm:suff} to $\structA'$, and obtain that
$\optfrac{k}{\structA',\structB}=\opt{\structA',\structB}$. Since
$\structA'\equiv \structA$, we know that
$\opt{\structA,\structB}=\opt{\structA',\structB}$ (by the definition of
equivalence and cores) and
$\optfrac{k}{\structA,\structB}=\optfrac{k}{\structA',\structB}$ (by
Proposition~\ref{prop:corefrac}). Hence,
$\optfrac{k}{\structA,\structB}=\opt{\structA,\structB}$.  
\end{proof}

\subsection{Necessity of treewidth modulo scopes}
\label{subsec:nec-tw}

In this and the following section we will denote
by $\oplus$ the addition modulo $2$. 
We show the following. 

\begin{theorem}
\label{thm:nece-core}
Let $\structA$ be a valued $\sigma$-structure and let $k\geq 1$. 
Suppose that $\structA$ is a core and $\twms{\structA}\geq k$. 
Then there exists a valued $\sigma$-structure $\structB$ such that $\optfrac{k}{\structA,\structB}<\opt{\structA,\structB}$. 
\end{theorem}

\begin{proof}
Let $\rels{A} \defeq \pos{\structA}$. 
As usual, we denote by $A$ the universe of $\structA$, $\rels{A}$ and the vertex
  set of $G(\rels{A})$. 
Since $\twms{\structA}=\twms{G(\rels{A})}\geq k$, 
Theorem~\ref{theo:char-twms} implies that 
there exists a bramble $\B$ of $G(\rels{A})$ that cannot be covered by any scope nor subset of size at most $k$ in $G(\rels{A})$. 
Every $B\in \B$ must belong to the same connected component of $G(\rels{A})$, 
which we denote by $G_0=(A_0,E_0)$. 
Fix any $a_0 \in A_0$.  
For each $a\in A$, let $d_a$ denote the degree of $a$ in $G(\rels{A})$, 
and let $e_1^a,\dots,e_{d_a}^a$ be a fixed enumeration of all the edges incident to $a$ in $G(\rels{A})$.

Recall that $\rels{A}$ is defined over the relational signature $\rel{\sigma}=\{R_f\mid f\in \sigma, \ar{R_f}=\ar{f}\}$. 
We define a relational structure $\rels{B}$ over $\rel{\sigma}$ as in~\cite{Atserias07:power}. 
The universe of $\rels{B}$, denoted by $B$, contains precisely all the tuples $(a,(b_1,\dots,b_{d_a}))$ such that
\begin{enumerate}
\item $a\in A$ and $b_1,\dots,b_{d_a}\in \{0,1\}$; 
\item $b_1\oplus \cdots \oplus b_{d_a} = 0$, if $a\neq a_0$;
\item $b_1\oplus \cdots \oplus b_{d_a} = 1$, if $a=a_0$. 
\end{enumerate}

Let $R_f\in \rel{\sigma}$ be of arity $n$. A tuple $((a^1,(b_1^1,\dots,b^1_{d_{a^1}})), \dots, (a^n,(b_1^n,\dots,b^n_{d_{a^n}})))$ 
belongs to $R_f^{\rels{B}}$ if and only if
\begin{enumerate}
\item $(a^1,\dots,a^n)$ belongs to $R_f^{\rels{A}}$, 
\item if $\{a^\ell,a^m\}=e_{i}^{a^\ell}=e_j^{a^m}$, for some $\ell,m\in \{1,\dots,n\}$, $i\in \{1,\dots,d_{a^\ell}\}$ and $j\in \{1,\dots,d_{a^m}\}$, 
then $b_i^\ell=b_j^m$. 
\end{enumerate}

Let $\pi:B\mapsto A$ be the first projection, i.e., $\pi((a,(b_1,\dots,b_{d_a})))=a$, for all $(a,(b_1,\dots,b_{d_a}))\in B$. 
By definition of $\rels{B}$, $\pi$ is a homomorphism from $\rels{B}$ to $\rels{A}$. 

The following lemma follows directly from the proof of~\cite[Lemma~1]{Atserias07:power}.

\begin{lemma}
\label{lemma:nohom}
There is no homomorphism $h$ from $\rels{A}$ to $\rels{B}$ such that $\pi\circ h(A)=A$. 
\end{lemma}

By Proposition~\ref{prop:core-char}, since $\structA$ is a core, 
there is a function $c^*:\tup{\structA}\mapsto \qplus$ such that every non-surjective mapping $g: A \mapsto A$ satisfies 
\[\sum_{(f,\tuple{x})\in \tup{\structA}}f^{\structA}(\tuple{x})c^*(f,\tuple{x}) < \sum_{(f,\tuple{x})\in
  \tup{\structA}}f^{\structA}(\tuple{x})c^*(f,g(\tuple{x})).\]
 We denote $M^* \defeq \sum_{(f,\tuple{x})\in \tup{\structA}}f^{\structA}(\tuple{x})c^*(f,\tuple{x})<\infty$. 

Now we are ready to define $\structB$. The universe of $\structB$ is $B$, i.e.,
the same as $\rels{B}$. 
For each $f\in \sigma$ and tuple $\tuple{x}\in B^{\ar{f}}$ we define 
\[
f^{\structB}(\tuple{x})=
\begin{cases}
c^*(f,\pi(\tuple{x})) &\quad \text{if $\tuple{x}\in R_f^{\rels{B}}$}\\
\infty &\quad \text{otherwise}
\end{cases}
\]

\setcounter{claim}{0}
\begin{claim}
\label{claim:optAB}
$\opt{\structA,\structB}> M^*$.
\end{claim}

\begin{proof}
Let $h$ be a homomorphism from $\rels{A}$ to $\rels{B}$ (if it is not then $\costb{\structA\mapsto\structB}{h}=\infty$). 
Then 
\[\costb{\structA\mapsto\structB}{h}=\sum_{(f,\tuple{x})\in \tup{\structA}_{>0}} f^{\structA}(\tuple{x})f^{\structB}(h(\tuple{x})) =\sum_{(f,\tuple{x})\in \tup{\structA}_{>0}} f^{\structA}(\tuple{x})c^*(f,\pi(h(\tuple{x}))).\]
By Lemma~\ref{lemma:nohom}, $\pi\circ h$ is a non-surjective mapping from $A$ to $A$. 
By definition of $c^*$, we have that $\costb{\structA\mapsto\structB}{h}>M^*$. 
\renewcommand{\qedsymbol}{$\blacksquare$}\end{proof}

In the rest of the proof, we will show that $\optfrac{k}{\structA,\structB}\leq
M^*$. Together with Claim~\ref{claim:optAB}, this establishes Theorem~\ref{thm:nece-core}.

Let $s$ be a partial homomorphism from $\rels{A}$ to $\rels{B}$. We denote by $\dom{s}$ the domain of $s$. 
We say that $s$ is an \emph{identity partial homomorphism} from $\rels{A}$ to $\rels{B}$ if $\pi(s(a))=a$, for all $a\in \dom{s}$. 
We denote by IPHom$(\rels{A},\rels{B})$ the set of all identity partial homomorphisms from $\rels{A}$ to $\rels{B}$. 
Let $\lambda$ be a feasible solution of SA$_k(\structA,\structB)$ satisfying the following property $(\dagger)$: 
if $\lambda(f,\tuple{x},s)>0$ for $(f,\tuple{x})\in \tup{\structA_k}_{>0}, s:\toset{\tuple{x}}\mapsto B$, 
then $s$ is an identity partial homomorphism from $\rels{A}$ to $\rels{B}$.  
We claim that the cost of any such $\lambda$ is precisely $M^*$. 
Indeed, we have that 
\begin{align*}
&\sum_{\substack{(f,\tuple{x})\in \tup{\structA_k}_{>0}, s:\toset{\tuple{x}}\mapsto B \\ f^{\structA_k}(\tuple{x})\times f^{\structB_k}(s(\tuple{x}))<\infty}} \lambda(f,\tuple{x},s)f^{\structA_k}(\tuple{x})f^{\structB_k}(s(\tuple{x})) \\
&=\sum_{\substack{(f,\tuple{x})\in \tup{\structA}_{>0}, s:\toset{\tuple{x}}\mapsto B\\ f^{\structA}(\tuple{x})\times f^{\structB}(s(\tuple{x}))<\infty}} \lambda(f,\tuple{x},s)f^{\structA}(\tuple{x})f^{\structB}(s(\tuple{x})) \\
&= \sum_{(f,\tuple{x})\in \tup{\structA}_{>0}} f^{\structA}(\tuple{x}) \sum_{s:\toset{\tuple{x}}\mapsto B, f^{\structA}(\tuple{x})\times f^{\structB}(s(\tuple{x}))<\infty} \lambda(f,\tuple{x},s)f^{\structB}(s(\tuple{x}))\\
&=\sum_{(f,\tuple{x})\in \tup{\structA}_{>0}} f^{\structA}(\tuple{x}) \sum_{\substack{s:\toset{\tuple{x}}\mapsto B, f^{\structA}(\tuple{x})\times f^{\structB}(s(\tuple{x}))<\infty \\ s\in \text{IPHom}(\rels{A},\rels{B})}} \lambda(f,\tuple{x},s)f^{\structB}(s(\tuple{x}))\\
&=\sum_{(f,\tuple{x})\in \tup{\structA}_{>0}} f^{\structA}(\tuple{x}) \sum_{\substack{s:\toset{\tuple{x}}\mapsto B, f^{\structA}(\tuple{x})\times f^{\structB}(s(\tuple{x}))<\infty \\ s\in \text{IPHom}(\rels{A},\rels{B})}} \lambda(f,\tuple{x},s)c^*(f,\tuple{x})\\
&=\sum_{(f,\tuple{x})\in \tup{\structA}_{>0}} f^{\structA}(\tuple{x}) c^*(f,\tuple{x}) = M^*.
\end{align*}

Thus it suffices to show the existence of a feasible solution to SA$_k(\structA,\structB)$ satisfying $(\dagger)$. 
In~\cite[Lemma~2]{Atserias07:power} it is shown that, if the treewidth of $\rels{A}$ is at least $k$, then the $(k-1)$-consistency test succeeds over $\rels{A},\rels{B}$. 
To do this, the authors of~\cite{Atserias07:power} exhibit a winning strategy $\H_k$ for the Duplicator in the existential $k$-pebble game over $\rels{A},\rels{B}$. 
We built on their construction to define a set $\H$ of identity partial homomorphisms from $\rels{A}$ to $\rels{B}$, which will allow us to define our required $\lambda$. 

Recall that we fixed a bramble $\B$ of $G(\rels{A})$ that cannot be covered by any scope nor subset of size at most $k$ in $G(\rels{A})$, 
whose existence is guaranteed by Theorem~\ref{theo:char-twms}. 
Recall also that every set in $\B$ belongs to the connected component $G_0=(A_0,E_0)$ of $G(\rels{A})$ and that $a_0\in A_0$. 
Let $P$ be a path in $G(\rels{A})$ starting at $a_0$. For every edge $e$ of $G(\rels{A})$, we define
\[
x_e^P=
\begin{cases}
1 &\quad \text{if $e$ appears an odd number of times in $P$}\\
0 &\quad \text{otherwise}.
\end{cases}
\]
For $a\in A$, we define $h^{P}(a)=(a,(x_{e_1^a}^P,\dots,x_{e_{d_a}^a}^P))$. 
We have $h^P(a)=(a,(0,\dots,0))$, for all $a\in A\setminus A_0$. 
The following claim is shown in~\cite{Atserias07:power} (second claim in the proof of Lemma~2): 

\begin{claim}
\label{claim:iph-path}
Let $X\subseteq A$, and let $P$ be a path in $G(\rels{A})$ from $a_0$ to $b$, where $b\not\in X$. 
Then, the restriction $h^P|_X$ of $h^P$ to $X$ is a partial homomorphism from $\rels{A}$ to $\rels{B}$. 
\end{claim}

Notice that in the previous claim, $h^P|_X$ is actually an identity partial homomorphism. 
Let $\X \defeq \{X\subseteq A: \text{$X$ does not cover $\B$}\}$. 
We define a set $\H$ of identity partial homomorphisms from $\rels{A}$ to $\rels{B}$ as follows. 
For every $X\in \X$ and every path $P$ in $G(\rels{A})$ from $a_0$ to $b$, where $b$ belongs to a set $B\in \B$ 
disjoint from $X$, we include $h^P|_{X}$ to $\H$. 
By Claim~\ref{claim:iph-path}, $\H$ contains only identity partial homomorphisms from $\rels{A}$ to $\rels{B}$. 
For $X\in \X$, let $\H(X) \defeq \{h\in \H\mid \dom{h}=X\}$. 
We have $|\H(X)|>0$, for all $X\in \X$, as $X$ does not cover $\B$ and then there is a set $B\in \B$ disjoint from $X$, 
and by connectivity, there is at least one path $P$ from $a_0$ to some node $b\in B$. 

It follows from the proof of Lemma~2 in~\cite{Atserias07:power} that $\H$ has the following closure properties:

\begin{claim}
\label{claim:closureH}
Let $X,X'\in \X$ such that $X\subseteq X'$, then 
\begin{enumerate}
\item if $h\in \H(X)$, then there exists $h'\in \H(X')$ such that 
$h'|_{X}=h$. 
\item if $h'\in \H(X')$, then $h'|_{X}\in \H(X)$. 
\end{enumerate}
\end{claim}

For $X,X'\in \X$ such that $X\subseteq X'$ and $h\in B^{X}$, 
let $\H(X')^{X,h} \defeq \{h'\in \H(X')\mid h'|_{X}=h\}$. 
Claim~\ref{claim:closureH}, item (1) states that $|\H(X')^{X,h}|>0$, for all $h\in \H(X)$. 

The arguments below are an adaptation of the argument from~\cite{tz17:sicomp} for proving the existence of gap instances for Sherali-Adams relaxations of VCSP($E_{\G,3}$), 
and a refinement of an argument from~\cite[Lemma~2]{Atserias07:power} (specifically Claim~\ref{claim:uniform} below). 
In particular, we shall define a well-behaved subset $\S\subseteq\X$. 
A key property of $\S$ is that the family of distributions $\{U(\H(S))\mid S\in \S\}$, where $U(\H(S))$ is the uniform distribution over $\H(S)$, is \emph{consistent} in the following sense: 
for every $S\subseteq S'$, the marginal distribution of $U(\H(S'))$ over $S$ coincides with $U(\H(S))$. 
As it turns out, we will be able to extend $\{U(\H(S))\mid S\in \S\}$ to a consistent family $\{\mu(\H(X))\mid X\in \X\}$ over the whole set $\X$. 
We shall define our required $\lambda$ from this latter family. 

\medskip

We say that $X\subseteq A$ \emph{separates} $G_0$ if there exist nodes $u,u'\in A_0\setminus X$ such that $X$ \emph{separates} $u$ and $u'$, i.e., 
every path connecting $u$ and $u'$ in $G_0$ intersects $X$. 
Let $\S \defeq \{S\in \X\mid \text{$S$ does not separate $G_0$}\}$. 
Below we state two important properties of $\S$:

\begin{claim}
\label{claim:interS}
$\S$ is closed under intersection.
\end{claim}

\begin{proof}
Let $S,S'\in \S$. Clearly, $S\cap S'$  does not cover $\B$. To see that $S\cap S'$ does not separate $G_0$, take $a,a'\in A_0\setminus (S\cap S')$. 
We need to find a path in $G_0$ between $a$ and $a'$ avoiding $S\cap S'$. 
If $a,a'\in A_0\setminus S$ or $a,a'\in A_0\setminus S'$, we are done as $S$ and $S'$ do not separate $G_0$. 
Thus, without loss of generality we can assume that $a\in (A_0\cap S)\setminus S'$ and $a'\in (A_0\cap S')\setminus S$. 
Let $B,B'\in \B$ such that $B\cap S=\emptyset$ and $B'\cap S'=\emptyset$, and pick $b\in B$ and $b'\in B'$. 
Since $S'$ does not separate $a$ and $b'$, there exists a path $P$ from $a$ to $b'$ avoiding $S'$, and hence, $S\cap S'$. 
Since $B,B'$ are connected and they touch,  there is also a path $P'$ from $b'$ to $b$ avoiding $S\cap S'$. 
Finally, since $S$ does not separate $a'$ and $b$, there is a path $P''$ from $b$ to $a'$ avoiding $S$, and then, $S\cap S'$. 
The claim follows by taking the concatenation of $P$, $P'$ and $P''$.
\renewcommand{\qedsymbol}{$\blacksquare$}\end{proof}

\begin{claim}
\label{claim:uniform}
  Let $S,S'\in \S$ such that $S\subseteq S'$ and let $h\in \H(S)$. Then
  $|\H(S')^{S,h}|=|\H(S')|/|\H(S)|$. 
\end{claim}

\begin{proof}
The proof is by induction on $\ell \defeq |S'\setminus S|$. If $\ell=0$, we are done. 
Suppose that $\ell=1$ and $S'\setminus S=\{x\}$. 
Assume first that $x\not\in A_0$. 
Let $g\in \H(S)$ and $P$ be any path starting at $a_0$ and ending at some node in $B\in \B$, with $B\cap S=\emptyset$, such that $h^P|_{S}=g$. 
Since $B\cap S'=\emptyset$ as $B\subseteq A_0$ and thus $h^P|_{S'}\in \H(S')^{S,g}$.  
Notice also that for every $g'\in \H(S')^{S,g}$, it must be the case that $g'(x)=(x,(0,\dots,0))$, 
as a path $P'$ starting at $a_0$ cannot visit an edge $e$ incident to $x$, and hence $x^{P'}_e=0$. 
It follows that $\H(S')^{S,g}=\{h^P|_{S'}\}$. 
In particular, $|\H(S')^{S,h}|=1$ and $|\H(S')|=\sum_{g\in B^S} |\H(S')^{S,g}|$. 
By Claim~\ref{claim:closureH}, item (2), $\sum_{g\in B^S} |\H(S')^{S,g}|=\sum_{g\in \H(S)} |\H(S')^{S,g}|$. 
This implies that $|\H(S')|=|\H(S)|$, and then the claim follows.

We now assume that $x\in A_0$. We show that there is an edge $\{x,u\}$ in $G(\rels{A})$ such that $u\not\in S'$. 
Assume to the contrary, and take $B'\in \B$ such that $S'\cap B'=\emptyset$ and pick $b'\in B'$. 
We have $x,b'\in A_0\setminus S$.  
Since every neighbour of $x$ lies in $S$, we have that $S$ separates $x$ and $b'$, which contradicts the fact that $S\in \S$. 
Let $e_1,\dots, e_m$, with $m\geq 1$, be an enumeration of all the edges from $x$ to a node outside $S'$. 
Recall that $e^x_1,\dots,e^x_{d_x}$ is our fixed enumeration of all the edges incident to $x$. 
Without loss of generality, suppose that $e^x_1,\dots,e^x_{d_x}=e_1,\dots,e_m, e'_1,\dots,e'_n$, where $e'_1,\dots,e'_n$ is an enumeration ($n\geq 0$) of all the edges 
from $x$ to a node inside $S$. 
We claim that $|\H(S')^{S,g}|=2^{m-1}$, for all $g\in \H(S)$. 

We start by proving $|\H(S')^{S,g}|\leq 2^{m-1}$. 
We define an injective mapping $\eta$ from $\H(S')^{S,g}$ to $\{0,1\}^{m-1}$. 
Let $g'\in \H(S')^{S,g}$ and suppose that $g'(x)=(x, (y_1,\dots,y_m, y'_1,\dots,y'_n))$. 
We define $\eta(g')=(y_1,\dots,y_{m-1})$. Suppose that $\eta(g'_1)=\eta(g'_2)$, for $g'_1,g'_2\in \H(S')^{S,g}$. 
Thus, $g'_1(x), g'_2(x)$ are of the form $g'_1(x)=(x,(\eta(g'_1), y_{1m}, y'_{11},\dots,y'_{1n}))$ and 
$g'_2(x)=(x,(\eta(g'_2), y_{2m}, y'_{21},\dots,y'_{2n}))$. 
We have $(y'_{11},\dots,y'_{1n})=(y'_{21},\dots,y'_{2n})$ as these values are already determined by $g$. 
It follows that $y_{1m}=y_{2m}$ as the parity of the number of ones in $(\eta(g'_1), y_{1m}, y'_{11},\dots,y'_{1n}))$ and $(\eta(g'_2), y_{2m}, y'_{21},\dots,y'_{2n}))$ must be the same. 
In particular, $g'_1=g'_2$, and thus $\eta$ is injective.

Now we show that $|\H(S')^{S,g}|\geq 2^{m-1}$. 
As we mentioned above, if $g'\in \H(S')^{S,g}$ and $g'(x)=(x, (y_1,\dots,y_m, y'_1,\dots,y'_n))$, then $(y'_1,\dots,y'_n)$ is already determined by $g$. 
Let $\bar z=(z_1,\dots,z_{m})\in \{0,1\}^{m}$ be an arbitrary $0/1$-vector such that  

\begin{itemize}
\item $\bigoplus_{1\leq i\leq m}z_i \oplus \bigoplus_{1\leq j\leq n} y'_j = 1$, if $x=a_0$;  
\item $\bigoplus_{1\leq i\leq m}z_i \oplus \bigoplus_{1\leq j\leq n} y'_j = 0$, if $x\neq a_0$. 
\end{itemize}

The number of such $\bar z$'s is precisely $2^{m-1}$, and hence, it suffices to show that, 
for every such $\bar z$, there exists $g'_{\bar z}\in \H(S')^{S,g}$ such that $g'_{\bar z}(x)=(x, (\bar z,y'_1,\dots,y'_n))$. 
Let $P$ be any path in $G(\rels{A})$ from $a_0$ to a node $b\in B$, where $B\in \B$ and $B\cap S=\emptyset$, such that 
$h^P|_{S}=g$. Let $B'\in \B$ such that $B'\cap S'=\emptyset$ and pick any $b'\in B'$. 
Since $B,B'$ are connected and they touch, there exists a path $P'$ from $b$ to $b'$ completely contained in $B\cup B'$, 
and hence, avoiding $S$. 
Let $W$ be the concatenation of $P$ and $P'$. 
By construction, $g' \defeq h^{W}|_{S'}\in \H(S')$. 
Since $P'$ avoids $S$, $x^{W}_e=x^{P}_e$, for every edge $e$ incident to a node in $S$. 
Therefore, $h^{W}(a)=h^{P}(a)$, for all $a\in S$. 
It follows that $g'\in \H(S')^{S,g}$.

Let $g'(x)=(x, (y_1,\dots,y_m, y'_1,\dots,y'_n))$. 
Suppose without loss of generality that $\bar z=(z_1,\dots,z_r, z_{r+1},\dots,z_{m})$, where $r\in \{0,\dots,m\}$ and 
$z_i\neq y_i$ if and only if $i\in \{1,\dots,r\}$. If $r=0$, we are done as we can set $g'_{\bar z}=g'$. 
Suppose then that $r\geq 1$. Observe that $\bigoplus_{1\leq i\leq r}y_i = \bigoplus_{1\leq i\leq r} z_i$. 
This implies that $r$ is even.  
Recall that $e_1,\dots,e_m$ is our fixed enumeration of all the edges from $x$ to a node outside $S'$. 
Let $u_1,\dots,u_r$ be the nodes outside $S'$ such that $e_i=\{x,u_i\}$, for all $i\in \{1,\dots,r\}$. 
We have $u_i,b'\in A_0\setminus S'$. 
Since $S'\in \S$, $S'$ cannot separate $u_i$ and $b'$, there is a path $P_i$ from $b'$ to $u_i$ avoiding $S'$, for all $i\in \{1,\dots,r\}$. 
Let $W_i$ be the extension of $P_i$ with the edge $\{u_i,x\}$. 
We denote by $W_i^{-1}$ the reverse path of $W_i$, in particular, $W_i^{-1}$ is a path from $x$ to $b'$. 
Let $W'$ be the extension of the path $W$ with the concatenation $Z$ of the paths $W_1$, $W_{2}^{-1}$,$\dots$, $W_{r-1}$, $W_{r}^{-1}$. 
Since $W'$ ends at $b'$ we have $h^{W'}|_{S'}\in \H(S')$. 
Since the subpath $Z$ of $W'$ avoids $S$, $h^{W'}(a)=h^{W}(a)$, for all $a\in S$, and then, $h^{W'}|_{S'}\in \H(S')^{S,g}$. 
Notice also that the subpath $Z$ visits only the $e_i$'s with $i\in \{1,\dots,r\}$, and it visits 
each such $e_i$ exactly once. This implies that $h^{W'}(x)=(x,(y_1\oplus 1,\dots, y_r\oplus 1,y_{r+1},\dots,y_m,y'_1,\dots,y'_n))$.
In other words, $h^{W'}(x)=(x,(z_1,\dots,z_m,y'_1,\dots,y'_n))$. The claim follows by taking $g'_{\bar z}=h^{W'}|_{S'}$. 
Therefore, $|\H(S')^{S,g}|=2^{m-1}$, for all $g\in \H(S)$. 
In particular, $|\H(S')^{S,h}|=2^{m-1}$. 
On the other hand, $|\H(S')|=\sum_{g\in B^S} |\H(S')^{S,g}|$. 
By Claim~\ref{claim:closureH}, item (2), $\sum_{g\in B^S} |\H(S')^{S,g}|=\sum_{g\in \H(S)} |\H(S')^{S,g}|$. 
This implies that $|\H(S')|=2^{m-1}|\H(S)|$, and then our claim holds for $\ell=1$.

Assume now that $\ell\geq 2$. 
We claim that there is a node $x^*\in S'\setminus S$ such that $S'\setminus \{x^*\}\in \S$. 
By contradiction, suppose that this is not the case. 
If $x\not\in A_0$ for $x\in S'\setminus S$, then $S'\setminus \{x\} \in \S$ as $S'\setminus \{x\} $ cannot separate $G_0$. 
Similarly, if there is an edge $\{x,u\}$ in $G(\rels{A})$ with $x\in S'\setminus S$ and $u\not\in S'$, then $S'\setminus \{x\}\in \S$ 
for the same reason as before. 
It follows that $x\in A_0$ and every neighbour of $x$ lies inside $S'$, for all $x\in S'\setminus S$. 
Let $B'\in\B$ be disjoint from $S'$. Then any path from a node in $S'\setminus S$ to a node in $B'$ must intersect $S$. 
This contradicts the fact that $S\in \S$. Let $S^* \defeq S'\setminus\{x^*\}\in \S$. 

We have $|\H(S')^{S,h}|=\sum_{g\in \H(S^*)^{S,h}} |\H(S')^{S^*,g}|$. 
By inductive hypothesis, we know that $|\H(S')^{S^*,g}|=|\H(S')|/|\H(S^*)|$, for every $g\in \H(S^*)$, and 
$|\H(S^*)^{S,h}|=|\H(S^*)|/|\H(S)|$. It follows that $|\H(S')^{S,h}|=(|\H(S')|/|\H(S^*)|)|\H(S^*)^{S,h}|=|\H(S')|/|\H(S)|$. 
\renewcommand{\qedsymbol}{$\blacksquare$}\end{proof}

Now we are ready to define our vector $\lambda$ satisfying ($\dagger$). 
Fix $(f,\tuple{x})\in \tup{\structA_k}_{>0}$. $\toset{\tuple{x}}$ is either a scope or a 
subset of size at most $k$ of $G(\rels{A})$. 
In particular, $\toset{\tuple{x}}$ cannot cover $\B$ and thus $\toset{\tuple{x}}\in \X$. 
Let $\overline{\toset{\tuple{x}}} \defeq \bigcap_{S\in \S\mid \toset{\tuple{x}}\subseteq S} S$. 
Let us note that $\overline{\toset{\tuple{x}}}$ is well-defined, i.e., there is $S^*\in \S$ such that $\toset{\tuple{x}}\subseteq S^*$. 
Indeed, we can take $S^*=A\setminus B\in \S$, where $B$ is any set in $\B$ disjoint from $\toset{\tuple{x}}$. 
By Claim~\ref{claim:interS}, $\overline{\toset{\tuple{x}}}\in \S$. 
Observe then that $\overline{\toset{\tuple{x}}}$ is the inclusion-wise minimal set of $\S$ containing $\toset{\tuple{x}}$. 
For every mapping $s:\toset{\tuple{x}}\mapsto B$, we define
\[\lambda(f,\tuple{x},s)=\Pr_{h\sim U(\H(\overline{\toset{\tuple{x}}}))}\left[h|_{\toset{\tuple{x}}}=s\right]\] 
where $U(\H(\overline{\toset{\tuple{x}}}))$ denotes the uniform distribution over $\H(\overline{\toset{\tuple{x}}})$. 

By Claim~\ref{claim:closureH}, we have that ($*$) $\lambda(f,\tuple{x},s)>0$ if and only if $s\in \H(\toset{\tuple{x}})$. 
Hence, $\lambda$ satisfies ($\dagger$). 
It remains to prove that $\lambda$ is feasible for SA$_k(\structA,\structB)$. 
Conditions \textcolor{red}{\textbf{(SA4)}} and \textcolor{red}{\textbf{(SA2)}} follow from definition. 
For condition \textcolor{red}{\textbf{(SA3)}}, 
recall first that the function $c^*:\tup{\structA}\mapsto \qplus$, whose existence is guaranteed by the fact that $\structA$ is a core, 
satisfies that  
\[\sum_{(f,\tuple{x})\in \tup{\structA}}f^{\structA}(\tuple{x})c^*(f,\tuple{x}) < \sum_{(f,\tuple{x})\in \tup{\structA}}f^{\structA}(\tuple{x})c^*(f,g(\tuple{x}))\]
for every non-surjective mapping $g: A \mapsto A$. 
In particular, $\sum_{(f,\tuple{x})\in \tup{\structA}}f^{\structA}(\tuple{x})c^*(f,\tuple{x}) <\infty$. 
It follows that $f^\structA(\tuple{x})=\infty$ implies $c^*(f,\tuple{x})=0$, 
for every $(f,\tuple{x})\in \tup{\structA}$. 
By contradiction, suppose that condition \textcolor{red}{\textbf{(SA3)}} does not hold. 
Then there is $(f,\tuple{x})\in \tup{\structA}$ and $s:\toset{\tuple{x}}\mapsto B$ 
such that $f^{\structA}(\tuple{x})f^{\structB}(s(\tuple{x}))=\infty$, but $\lambda(f,\tuple{x},s)>0$. 
By ($*$), we know that $s\in \H(\toset{\tuple{x}})$. 
In particular, $s$ is a partial homomorphism from $\rels{A}$ to $\rels{B}$. 
On the other hand, note that $f^{\structB}(s(\tuple{x}))<\infty$ by construction of $\structB$, 
and then $f^{\structA}(\tuple{x})=\infty$. By a previous remark, we have that $c^*(f,\tuple{x})=0$. 
Additionally, note that $\tuple{x}\in R_f^{\rels{A}}$, and since $s$ is a partial homomorphism, 
$s(\tuple{x})\in R_f^{\rels{B}}$. By definition of $\structB$, it follows that $f^{\structB}(s(\tuple{x}))=c^*(f,\pi(s(\tuple{x})))=c^*(f,\tuple{x})=0$. 
Hence, $f^{\structA}(\tuple{x})f^{\structB}(s(\tuple{x}))=0$; a contradiction.

For condition \textcolor{red}{\textbf{(SA1)}}, let $(f,\tuple{x}), (p,\tuple{y})\in \tup{\structA_k}_{>0}$
and define $X \defeq \toset{\tuple{x}}$ and $Y \defeq \toset{\tuple{y}}$. 
Suppose that $|X|\leq k$ and $X\subseteq Y$.
Let $s:X\mapsto B$. 
We need to show that 
\[\lambda(f,\tuple{x},s)=\sum_{\substack{r:Y\mapsto B\\ r|_{X}=s}}\lambda(p,\tuple{y},r)=\sum_{r\in \H(Y)^{X,s}}\lambda(p,\tuple{y},r),\]
where the last equality holds due to ($*$). Using again ($*$) and Claim~\ref{claim:closureH}, item (2), the required equality holds directly when $s\not\in \H(X)$. 
Then we assume that $s\in \H(X)$. 
For any $Z\in \X$ such that $X\subseteq Z\subseteq \overline{Y}$, we have that
\begin{equation}
\label{eq:X}
\begin{split}
\Pr_{h\sim U(\H(\overline{Y}))}\left[h|_{X}=s\right] &= \sum_{r:Z\mapsto B}\Pr_{h\sim U(\H(\overline{Y}))}\left[h|_{X}=s, h|_{Z}=r\right] \\
&= \sum_{\substack{r:Z\mapsto B\\ r|_{X}=s}}\Pr_{h\sim
  U(\H(\overline{Y}))}\left[h|_{Z}=r \right].
\end{split}
\end{equation}

By applying Equation~\eqref{eq:X} with $Z=\overline{X}$, we obtain

\begin{equation}
\label{eq:barX}
\Pr_{h\sim U(\H(\overline{Y}))}\left[h|_{X}=s\right] =\sum_{\substack{r:\overline{X}\mapsto B\\ r|_{X}=s}} \Pr_{h\sim U(\H(\overline{Y}))}\left[h|_{\overline{X}}=r\right]
= \sum_{r\in \H(\overline{X})^{X,s}} \Pr_{h\sim U(\H(\overline{Y}))}\left[h|_{\overline{X}}=r\right],
\end{equation}
where the last equality holds due to Claim~\ref{claim:closureH}, item (2). 
By definition, $\Pr_{h\sim U(\H(\overline{Y}))}\left[h|_{\overline{X}}=r\right]=|\H(\overline{Y})^{\overline{X},r}|/|\H(\overline{Y})|$. 
Since $r\in \H(\overline{X})$, we can apply Claim~\ref{claim:uniform} and obtain that $\Pr_{h\sim U(\H(\overline{Y}))}\left[h|_{\overline{X}}=r\right]=1/|\H(\overline{X})|$. 
Using this in Equation~\eqref{eq:barX}, we obtain 

\begin{align*}
\Pr_{h\sim U(\H(\overline{Y}))}\left[h|_{X}=s\right]  = \sum_{r\in \H(\overline{X})^{X,s}}1/|\H(\overline{X})| = \Pr_{h\sim U(\H(\overline{X}))}\left[h|_{X}=s\right] = \lambda(f,\tuple{x},s).
\end{align*}

Hence, 

\begin{align*}
\lambda(f,\tuple{x},s) &= \Pr_{h\sim U(\H(\overline{Y}))}\left[h|_{X}=s\right]  \\ 
&= \sum_{\substack{r:Y\mapsto B\\ r|_{X}=s}}\Pr_{h\sim
U(\H(\overline{Y}))}\left[h|_{Y}=r \right]  \quad \text{(By Equation~\eqref{eq:X} with $Z=Y$)} \\
&= \sum_{\substack{r:Y\mapsto B\\ r|_{X}=s}} \lambda(p,\tuple{y},r).
\end{align*}

Therefore, $\lambda$ is a feasible solution of SA$_k(\structA,\structB)$.  

\medskip

Finally, we remark that, as in Proposition~\ref{prop:main-hardness}, 
we can define $\structB$ to be finite-valued by replacing $\infty$ with a sufficiently large $N$ 
(taking $N=1+(M^*/\min\{f^{\structA}(\tuple{x}): (f,\tuple{x})\in \tup{\structA}_{>0}\})$ will do). 
This finishes the proof of Theorem~\ref{thm:nece-core}.
\end{proof}

\begin{theorem}
\label{thm:nece-core-core}
Let $\structA$ be a valued $\sigma$-structure and let $k\geq 1$. 
Let $\structA'$ be the core of $\structA$. If $\twms{\structA'}\geq k$, 
then the Sherali-Adams relaxation of level $k$ is not always tight for $\structA$. 
\end{theorem}

\begin{proof}
We can apply Theorem~\ref{thm:nece-core} to $\structA'$, and obtain $\structB$ such that 
$\optfrac{k}{\structA',\structB}<\opt{\structA',\structB}$. 
Since $\structA'\equiv \structA$, we know that
$\opt{\structA,\structB}=\opt{\structA',\structB}$ (by the definition of
equivalence and cores)
and $\optfrac{k}{\structA,\structB}=\optfrac{k}{\structA',\structB}$ (by
Proposition~\ref{prop:corefrac}). 
Hence, $\optfrac{k}{\structA,\structB}<\opt{\structA,\structB}$, and the result follows. 
\end{proof}

\subsection{Necessity of bounded overlap}
\label{subsec:overlap}

In this section we provide the last missing piece in the proof of Theorem~\ref{thm:sa}.

\begin{theorem}
\label{thm:nece-scopes}
Let $\structA$ be a valued $\sigma$-structure and let $k \geq 1$. Suppose that
$\structA$ is a core and that the overlap of $\structA$ is at least $k+1$. Then there exists a valued
$\sigma$-structure $\structB$ such that $\optfrac{k}{\structA,\structB}<\opt{\structA,\structB}$. 
\end{theorem}

\begin{proof}
Given a tuple $\tuple{t}$ and an index $i \in \{1,\ldots,|\tuple{t}|\}$, we
denote by $\tuple{t}[i]$ the $i$-th entry of $\tuple{t}$. By extension, given a set $I$ of indices we use $\tuple{t}[I]$ to denote the tuple obtained from $\tuple{t}$ after discarding all entries $\tuple{t}[j]$ with $j \notin I$. Given a tuple $\tuple{t}$, we use $\tupsetsize{\tuple{t}}$ to denote the number of distinct elements appearing in $\tuple{t}$, i.e. $\tupsetsize{\tuple{t}} \defeq |\toset{\tuple{t}}|$.

Because the overlap of $\structA$ is at least $k+1$, there exist $(p,\tuple{x}),(q,\tuple{y}) \in \tup{\structA}$ such that:
\begin{itemize}
\item[(i)] $p^{\structA}(\tuple{x}), q^{\structA}(\tuple{y}) > 0$, 
\item[(ii)] $|\toset{\tuple{x}} \cap \toset{\tuple{y}}| > k$,  and 
\item[(iii)] $(p,\tuple{x}) \neq (q,\tuple{y})$. 
\end{itemize}
Let $I_x,I_y$ be sets of indices of size $n = |\toset{\tuple{x}} \cap \toset{\tuple{y}}|$ such that $\toset{\tuple{x}[I_x]} = \toset{\tuple{y}[I_y]} = \toset{\tuple{x}} \cap \toset{\tuple{y}}$. (There may be more than one possible choice for $I_x$ and $I_y$.)  We first define a relational structure $\rels{B}$ over the signature $\rel{\sigma} = \{ R_f \mid f \in \sigma, \text{ar}(R_f) = \text{ar}(f) \}$ and universe $B = A \times \{0,1\}$. Let $\pi : B \to A$ and $\eta : B \to \{0,1\}$ be the projections onto the first and second coordinates, respectively. For every symbol $R_f \in  \rel{\sigma}$, we add to $R_f^\rels{B}$ every tuple $\tuple{t} \in B^{\ar{R_f}}$ such that $\pi(\tuple{t}) \in R_f^{\pos{\structA}}$. Then, we remove from $R_p^\rels{B}$ all the tuples $\tuple{t}$ such that $\pi(\tuple{t}) = \tuple{x}$ and $\bigoplus_{i \in I_x} \eta(\tuple{t}[i]) = 1$ and we remove from $R_q^\rels{B}$ all the tuples $\tuple{t}$ such that $\pi(\tuple{t}) = \tuple{y}$ and $\bigoplus_{i \in I_y} \eta(\tuple{t}[i]) = 0$. This completes the definition of $\rels{B}$.

By construction, $\pi$ is a homomorphism from $\rels{B}$ to $\pos{\structA}$. As the next claim shows, whenever $h$ is a homomorphism from $\pos{\structA}$ to $\rels{B}$, 
the mapping $\pi\circ h$ cannot be surjective. 

\setcounter{claim}{0}
\begin{claim}
\label{claim:nonsurj}
If $h: A\to B$ is a homomorphism from $\pos{\structA}$ to $\rels{B}$, then $(\pi\circ h)(A)\neq A$. 
\end{claim}

\begin{proof}
Towards a contradiction, suppose that $(\pi \circ h)(A) = A$. In that case, the map $\pi \circ h$ is bijective and has an inverse $l$. The map $l$ is an isomorphism from $\pos{\structA}$ to itself, so in particular we have that $l(\tuple{x}) \in  R_p^{\pos{\structA}}$ and $l(\tuple{y}) \in  R_q^{\pos{\structA}}$. It follows that $h(l(\tuple{x})) \in R_p^\rels{B}$, $h(l(\tuple{y})) \in R_q^\rels{B}$, $\pi(h(l(\tuple{x}))) = \tuple{x}$ and $\pi(h(l(\tuple{y}))) = \tuple{y}$. Now, if we define
\[
b_S \defeq \bigoplus_{v \in \toset{\tuple{x}} \cap \toset{\tuple{y}}}\eta(h(l(v)))
\]
then by construction of $\rels{B}$ we have $b_S = \bigoplus_{i \in I_x} \eta(h(l(\tuple{x}))[i]) = 0$ and $b_S = \bigoplus_{i \in I_y} \eta(h(l(\tuple{y}))[i]) = 1$, a contradiction.
\renewcommand{\qedsymbol}{$\blacksquare$}\end{proof}

Now, since $\structA$ is a core, by Proposition~\ref{prop:core-char} there exists a function $c^*: \tup{\structA} \to \qplus$ such that for every non-surjective mapping $g : A \to A$,
\[M^* \defeq \sum_{(f,\tuple{z})\in \tup{\structA}}f^{\structA}(\tuple{z})c^*(f,\tuple{z}) < \sum_{(f,\tuple{z})\in \tup{\structA}}f^{\structA}(\tuple{z})c^*(f,g(\tuple{z})).\]
In particular, $M^*$ is finite so $f^{\structA}(\tuple{z}) = \infty$ implies $c^*(f,\tuple{z}) = 0$ for all $(f,\tuple{z}) \in \tup{\structA}$. 
Then, we define a valued $\sigma$-structure $\structB$ over the universe $B$ such that for all $f\in \sigma$ and $\tuple{t}\in B^{\ar{f}}$, 
\[
f^\structB(\tuple{t}) \defeq
\begin{cases}
c^*(f, \pi(\tuple{t})) &\quad \text{if } \tuple{t} \in R_f^\rels{B}\\
\infty &\quad \text{otherwise}
\end{cases}
\]

\begin{claim}
\label{claim:nec-scopes-int}
$\opt{\structA,\structB} > M^*$.
\end{claim}

\begin{proof}
Let $h$ be a homomorphism from $\pos{\structA}$ to $\rels{B}$ (if it is not then $\costb{\structA\mapsto\structB}{h}=\infty$). 
By Claim~\ref{claim:nonsurj}, we have $(\pi \circ h)(A) \neq A$, i.e., $\pi\circ h$ is not surjective. 
By definition of $c^*$, we have
\[\costb{\structA\mapsto\structB}{h}=\sum_{(f,\tuple{z})\in \tup{\structA}_{>0}} f^{\structA}(\tuple{z})f^{\structB}(h(\tuple{z}))
=\sum_{(f,\tuple{z})\in \tup{\structA}_{>0}} f^{\structA}(\tuple{z})c^*(f,\pi(h(\tuple{z}))) > M^*\]
and the claim follows.
\renewcommand{\qedsymbol}{$\blacksquare$}\end{proof}

The next step is to exhibit a solution to SA$_k(\structA,\structB)$ of cost exactly $M^*$. Recall that the signature $\sigma_k$ of the modified instance $(\structA_k,\structB_k)$ used to define SA$_k(\structA,\structB)$ contains an additional function symbol $\rho_k$ of arity $k$. 

\begin{claim}
\label{claim:nec-scopes-frac}
$\optfrac{k}{\structA,\structB} \leq M^*$.
\end{claim}

\begin{proof}
Let $(f,\tuple{z}) \in \tup{\structA_k}_{>0}$ and $s : \toset{\tuple{z}} \to B_k$. If $\pi(s(\tuple{z})) \neq \tuple{z}$ then we set $\lambda(f,\tuple{z},s) \defeq 0$. If instead $\pi(s(\tuple{z})) = \tuple{z}$, then we set $\lambda(f,\tuple{z},s)$ according to the following rules:
\begin{itemize}
\item If $(p,\tuple{x}) \neq (f,\tuple{z}) \neq (q,\tuple{y})$, then $\lambda(f,\tuple{z},s) \defeq 1/2^{\tupsetsize{\tuple{z}}}$;
\item If $(f,\tuple{z}) = (p,\tuple{x})$, then $\lambda(p,\tuple{x},s) \defeq 1/2^{\tupsetsize{\tuple{x}}-1}$ if $\bigoplus_{i \in I_x} \eta(s(\tuple{x})[i]) = 0$ and $\lambda(p,\tuple{x},s) \defeq 0$ otherwise;
\item If $(f,\tuple{z}) = (q,\tuple{y})$, then $\lambda(q,\tuple{y},s) \defeq 1/2^{\tupsetsize{\tuple{y}}-1}$ if $\bigoplus_{i \in I_y} \eta(s(\tuple{y})[i]) = 1$ and $\lambda(q,\tuple{y},s) \defeq 0$ otherwise.
\end{itemize}

Because $\lambda(f,\tuple{z},s)$ may only be nonzero if $\pi(s(\tuple{z})) = \tuple{z}$ and $f^{\structA}(\tuple{z}) = \infty$ implies $f^\structB(s(\tuple{z})) = c^*(f,\pi(s(\tuple{z}))) = c^*(f,\tuple{z}) = 0$ for all such $(s,\tuple{z})$, $\lambda$ satisfies the condition \textcolor{red}{\textbf{(SA3)}}. By construction, $\lambda$ satisfies the conditions \textcolor{red}{\textbf{(SA2)}} and \textcolor{red}{\textbf{(SA4)}} as well. 

For any $(f,\tuple{z}) \in \tup{\structA_k}$ we use $\text{id}(f^{\structB},\tuple{z})$ to denote the set of all assignments $r : \toset{\tuple{z}} \to B$ such that $\pi(r(\tuple{z})) = \tuple{z}$ and either $f = \rho_k$ or $r(\tuple{z}) \in R^{\rels{B}}_f$. Then, the cost of $\lambda$ is

\begin{align*}
\sum_{\substack{(f,\tuple{z}) \in \tup{\structA_k}_{>0}, \; s: \toset{\tuple{z}} \to B_k, \\ f^{\structA_k}(\tuple{z}) \times f^{\structB_k}(s(\tuple{z})) < \infty}} &\lambda(f,\tuple{z},s) f^{\structA_k}(\tuple{z}) f^{\structB_k}(s(\tuple{z})) \\
&= \sum_{(f,\tuple{z}) \in \tup{\structA_k}_{>0}, \; s \in \text{id}(f^{\structB},\tuple{z})} \lambda(f,\tuple{z},s) f^{\structA_k}(\tuple{z}) f^{\structB_k}(s(\tuple{z}))\\
&= \sum_{(f,\tuple{z}) \in \tup{\structA}_{>0}, \; s \in \text{id}(f^{\structB},\tuple{z})} \lambda(f,\tuple{z},s) f^{\structA}(\tuple{z}) f^{\structB}(s(\tuple{z}))\\
&= \sum_{(f,\tuple{z}) \in \tup{\structA}_{>0}, \; s \in \text{id}(f^{\structB},\tuple{z})} \lambda(f,\tuple{z},s) f^{\structA}(\tuple{z}) c^*(f,\pi(s(\tuple{z})))\\
&= \sum_{(f,\tuple{z}) \in \tup{\structA}_{>0}, \; s \in \text{id}(f^{\structB},\tuple{z})} \lambda(f,\tuple{z},s) f^{\structA}(\tuple{z}) c^*(f,\tuple{z})\\
&= \sum_{(f,\tuple{z}) \in \tup{\structA}_{>0}} f^{\structA}(\tuple{z}) c^*(f,\tuple{z})\\
&= M^*.
\end{align*}

At this point, to prove that $\lambda$ is indeed a solution to
SA$_k(\structA,\structB)$ of cost $M^*$ we need only show that it
satisfies the condition \textcolor{red}{\textbf{(SA1)}}. Let $(f,\tuple{z}), (g,\tuple{w})\in \tup{\structA_k}_{>0}$ 
such that $\toset{\tuple{z}}\subseteq \toset{\tuple{w}}$ and
  $\tupsetsize{\tuple{z}} \leq k$. Let $s : \toset{\tuple{z}} \to B_k$. If
  $\pi(s(\tuple{z})) \neq \tuple{z}$ then all the lambdas involved in the
  equation \textcolor{red}{\textbf{(SA1)}} for the triple $((f,\tuple{z}),
  (g,\tuple{w}), s)$ are zero, so in this case the condition holds. Now, assume
  that $\pi(s(\tuple{z})) = \tuple{z}$ and observe that $(p,\tuple{x}) \neq (f,\tuple{z}) \neq (q,\tuple{y})$ because $\tupsetsize{\tuple{x}}, \tupsetsize{\tuple{y}} > k$. We consider two cases.

If $(p,\tuple{x}) \neq (g,\tuple{w}) \neq (q,\tuple{y})$ then there are exactly $2^{\tupsetsize{\tuple{w}}-\tupsetsize{\tuple{z}}}$ mappings $r$ from $\toset{\tuple{w}}$ to $B_k$ such that $\pi(r(w)) = w$ for all $w \in \toset{\tuple{w}}$ and $\eta(r(z)) = \eta(s(z))$ for all $z \in \toset{\tuple{z}}$. Thus we have
\[
\lambda(f,\tuple{z},s) = 1/2^{\tupsetsize{\tuple{z}}} = 2^{\tupsetsize{\tuple{w}}-\tupsetsize{\tuple{z}}} 1/2^{\tupsetsize{\tuple{w}}} = \sum_{r : \toset{\tuple{w}} \to B_k, \, r|_{\toset{\tuple{z}}} = s} \lambda(g,\tuple{w},r)
\]
and \textcolor{red}{\textbf{(SA1)}} holds for the triple $((f,\tuple{z}), (g,\tuple{w}), s)$.

In the second case, we have either $(g,\tuple{w}) = (p,\tuple{x})$ or $(g,\tuple{w}) = (q,\tuple{y})$. Let $Z' \defeq \toset{\tuple{z}} \cap \toset{\tuple{x}} \cap \toset{\tuple{y}}$ and $Z'' \defeq \toset{\tuple{z}} \backslash Z'$. Since $Z'$ is a strict subset of $\toset{\tuple{x}} \cap \toset{\tuple{y}}$, the mapping $s$ has exactly $2^{|\toset{\tuple{x}} \cap \toset{\tuple{y}}| - |Z'| - 1}$ distinct extensions $r : \toset{\tuple{z}} \cup (\toset{\tuple{x}} \cap \toset{\tuple{y}}) \to B_k$ such that $\pi(r(w)) = w$ for all $w \in \toset{\tuple{z}} \cup (\toset{\tuple{x}} \cap \toset{\tuple{y}})$ and $\bigoplus_{v \in \toset{\tuple{x}} \cap \toset{\tuple{y}}} \eta(r(v)) = 0$ (if $(g,\tuple{w}) = (p,\tuple{x})$) or $\bigoplus_{v \in \toset{\tuple{x}} \cap \toset{\tuple{y}}} \eta(r(v)) = 1$ (if $(g,\tuple{w}) = (q,\tuple{y})$). In turn, each of these mappings $r$ has $2^{\tupsetsize{\tuple{w}} - |\toset{\tuple{x}} \cap \toset{\tuple{y}}| - |Z''|}$ distinct extensions defined on the whole of $\toset{\tuple{w}}$ that satisfy $\pi(r(w)) = w$ for all $w \in \toset{\tuple{w}}$. Putting everything together, $s$ has $2^{\tupsetsize{\tuple{w}} - (|Z'| + |Z''|) - 1} = 2^{\tupsetsize{\tuple{w}} - \tupsetsize{\tuple{z}} - 1}$ extensions to $\toset{\tuple{w}}$ with the correct bit sum. Thus we have
\[
\lambda(f,\tuple{z},s) = 1/2^{\tupsetsize{\tuple{z}}} = 
2^{\tupsetsize{\tuple{w}}-\tupsetsize{\tuple{z}}-1} 1/2^{\tupsetsize{\tuple{w}} - 1} = \sum_{r : \toset{\tuple{w}} \to B_k, \, r|_{\toset{\tuple{z}}} = s} \lambda(g,\tuple{w},r)
\]
and \textcolor{red}{\textbf{(SA1)}} holds in this last case as well, which concludes the proof of the claim.
\renewcommand{\qedsymbol}{$\blacksquare$}\end{proof}

The proof of Theorem~\ref{thm:nece-scopes} is now established by Claim~\ref{claim:nec-scopes-int} and Claim~\ref{claim:nec-scopes-frac}.
As in Theorem~\ref{thm:nece-core}, we can make $\structB$ to be finite-valued by replacing $\infty$ with a sufficiently large number $N$. 
\end{proof}

\begin{theorem}
\label{thm:nece-scopes-core}
Let $\structA$ be a valued $\sigma$-structure and let $k\geq 1$.  Let $\structA'$ be the core of $\structA$. If the overlap of $\structA'$ is at least $k+1$, then the Sherali-Adams relaxation of level $k$ is not always tight for $\structA$. 
\end{theorem}

\begin{proof}
We can apply Theorem~\ref{thm:nece-scopes} to $\structA'$, and obtain $\structB$ such that 
$\optfrac{k}{\structA',\structB}<\opt{\structA',\structB}$. 
Since $\structA'\equiv \structA$, we know that
$\opt{\structA,\structB}=\opt{\structA',\structB}$ (by the definition of
equivalence and cores)
and $\optfrac{k}{\structA,\structB}=\optfrac{k}{\structA',\structB}$ (by
Proposition~\ref{prop:corefrac}). 
Hence, $\optfrac{k}{\structA,\structB}<\opt{\structA,\structB}$, and the result follows. 
\end{proof}

\section{Search VCSP($\C$, $-$)} 
\label{sec:search}

If a class $\C$ of valued structures has bounded treewidth modulo equivalence then the Sherali-Adams LP hierarchy can be used to solve in polynomial time \problem{VCSP($\C$,$-$)}, that is, to compute the minimum cost of a mapping from $\structA \in \C$ to some arbitrary valued structure $\structB$. 
However, it may be the case that computing a mapping of that cost is NP-hard even though we
know that one exists. In this section we will focus on the search version of the \problem{VCSP}, which explicitly asks for a minimum-cost mapping and will be denoted by \problem{SVCSP}.

If $\C$ is a class of valued structures, we denote by \problem{Core Computation($\C$)} the
problem that takes as input some $\structA \in \C$, and asks to compute a
mapping $g: A \mapsto A$ such that $g(\structA)$ is a core of $\structA$ and
there exists an IFH $\omega$ from $\structA$ to
$\structA$ such that $g \in \supp{\omega}$.

Building on our results from Section~\ref{sec:sa} and adapting techniques
from~\cite{tz16:jacm}, we will prove our third main result.

\begin{theorem}[\textbf{Search classification}]
\label{thm:search}
Assume FPT $\neq$ W[1]. Let $\C$ be a recursively enumerable class of valued structures of bounded arity. Then, the following are equivalent:
\begin{enumerate}
\item \problem{SVCSP($\C$, $-$)} is in PTIME.
\item $\C$ is of bounded treewidth modulo equivalence and \problem{Core Computation($\C$)}
is in PTIME.
\end{enumerate}
\end{theorem}

\begin{remark}
\label{rem:crispsearch}
Given a class $\C$ of relational structures, let \problem{SCSP($\C$, $-$)} denote the search variant of \problem{CSP($\C$, $-$)}; i.e., given $\rels{A}$ and $\rels{B}$ with $\rels{A}\in\C$, the task is to return a homomorphism from $\rels{A}$ to $\rels{B}$ if one exists. When applied to (bounded-arity, recursively enumerable) classes of relational structures, Theorem~\ref{thm:search} states that \problem{SCSP($\C$, $-$)} is in PTIME if and only if $\C$ has bounded treewidth modulo homomorphic equivalence and computing a homomorphism from any given $\rels{A} \in \C$ to one of its cores is in PTIME. This result is folklore and can be easily derived from Grohe~\cite{Grohe07:jacm} and Dalmau, Kolaitis, and Vardi~\cite{Dalmau02:width}. We provide here a brief sketch of the argument since our proof of Theorem~\ref{thm:search} for classes of valued structures follows roughly the same strategy, although the technical details are significantly more involved. 

Let $\C$ be a class of relational structures and $\C'$ denote the class of all cores of structures in $\C$. If computing a homomorphism $g$ from any given $\rels{A} \in \C$ to one of its cores is in PTIME, then \problem{SCSP($\C$, $-$)} is polynomial-time reducible to \problem{SCSP($\C'$, $-$)} by simply replacing each instance $(\rels{A},\rels{B})$ of \problem{SCSP($\C$, $-$)} with $(g(\rels{A}),\rels{B})$. (If a homomorphism $h$ from $g(\rels{A})$ to $\rels{B}$ is returned, then $h \circ g$ is a homomorphism from $\rels{A}$ to $\rels{B}$. Otherwise, no homomorphism from $\rels{A}$ to $\rels{B}$ exists.) Furthermore, if $\C$ has bounded treewidth modulo homomorphic equivalence then \problem{SCSP($\C'$, $-$)} is in PTIME; one can, for example, compute an optimal tree decomposition of any $\rels{A'} \in \C'$ in linear time (using Bodlaender's algorithm~\cite{bodlaender96:tree}) and then use standard dynamic programming techniques. It follows that \problem{SCSP($\C$, $-$)}  is in PTIME for any class $\C$ satisfying both conditions. For the converse implication, if \problem{SCSP($\C$, $-$)} is in PTIME then, by Grohe's result~\cite{Grohe07:jacm} (and under our assumptions), $\C$ is of bounded treewidth modulo homomorphic equivalence. It only remains to prove that computing a homomorphism from $\rels{A} \in \C$ to one of its cores is in PTIME; for this task we will use a simple argument from Chen and Mengel~\cite{Chen15:icdt}. Observe that a relational structure $\rels{C}$ is not a core if and only if there exists a relational structure $\rels{C}_{a}$ obtained by removing one element $a$ from the universe of $\rels{C}$ (as well as all tuples containing $a$) that is homomorphically equivalent to $\rels{C}$. It follows that an algorithm that starts with the pair of structures $(\rels{A},\rels{A})$ and greedily removes elements from the universe of the right-hand side structure while maintaining homomorphic equivalence will eventually terminate with $(\rels{A},\rels{A'})$, where $\rels{A'}$ is the core of $\rels{A}$. Recall that $\rels{A} \in \C$ and $\C$ has bounded treewidth modulo homomorphic equivalence. Thus the homomorphism tests required by the algorithm outlined above can be done in polynomial time~\cite{Dalmau02:width}. It then suffices to run the assumed algorithm for \problem{SCSP($\C$, $-$)} on the pair $(\rels{A},\rels{A'})$ to compute a homomorphism from $\rels{A}$ to one of its cores. \end{remark}

If $\C$ is a class of valued structures, the problem \problem{Reduction
Step($\C$)} takes as input some $\structA \in \C$ and a mapping $g: A \mapsto A$
that belongs to the support of some IFH from $\structA$ to $\structA$. The goal
is to compute a mapping $g^+: A \mapsto A$ such that $g^+(A) \subsetneq g(A)$
and $g^+$ belongs to the support of some IFH from $\structA$ to $\structA$, or assert that no such mapping exist. The relevance of this problem to \problem{Core Computation} is highlighted by the following proposition.

\begin{proposition}
\label{prop:reduc}
Let $\structA$ be a valued structure and $g: A \mapsto A$ be a mapping that
  belongs to the support of some IFH from $\structA$ to $\structA$. Then,
  $g(\structA)$ is not the core of $\structA$ if and only if there exists a
  mapping $g^+: A \mapsto A$ such that $g^+(A) \subsetneq g(A)$ and $g^+$
  belongs to the support of some IFH from $\structA$ to $\structA$.
\end{proposition}

\begin{proof}
If there exists such a mapping $g^+$ then by Proposition~\ref{prop:easy} (3) from Appendix~\ref{app:equiv} we have $\structA \equiv g(\structA) \equiv g^+(\structA)$. Then, the cores of $g(\structA)$ and $g^+(\structA)$ are equivalent and by Proposition~\ref{prop:iso} they are isomorphic. By Proposition~\ref{prop:exist-core} the universe of the core of $g^+(\structA)$ has size at most $|g^+(A)| < |g(A)|$, so $g(\structA)$ is not a core.

For the converse implication, suppose that $g(\structA)$ is not a core. Let
  $\omega$ be an IFH from $\structA$ to $\structA$ such that $g \in
  \supp{\omega}$. By the definition of a core, there exists a non-surjective IFH
  from $g(\structA)$ to $g(\structA)$. Let $g^*$ be a mapping in its support
  such that $g^*(g(A)) \subsetneq g(A)$. By Proposition~\ref{prop:easy} (3), we
  have $\structA \equiv g(\structA) \equiv g^*(g(\structA))$. Then, by
  Proposition~\ref{prop:charfrac} there exist an IFH $\omega_1$ from $\structA$
  to $g^*(g(\structA))$ and an IFH $\omega_2$ from $g^*(g(\structA))$ to $\structA$. Let $g_1, g_2$ be arbitrary mappings in the support of $\omega_1, \omega_2$, respectively. If we define $\omega^+ \defeq \omega \circ \omega_2 \circ \omega_1$, where 
\[\omega^+(h) = (\omega \circ \omega_2 \circ \omega_1)(h)=\sum_{\substack{h_1:A\mapsto g^*(g(A)),\\ h_2:g^*(g(A))\mapsto A,\\ h_3: A \mapsto A:\\ h_3 \circ h_2 \circ h_1=h}} \omega(h_3)\omega_2(h_2)\omega_1(h_1)\]
then $\omega^+$ is an IFH from $\structA$ to $\structA$. Let $g^+ \defeq g \circ g_2 \circ g_1$. Because $|g^+(A)| \leq |g_1(A)| \leq |g^*(g(A))| < |g(A)|$, we have $g^+(A) \subsetneq g(A)$. Moreover,
\[
\omega^+(g^+) \geq \omega(g)\omega_2(g_2)\omega_1(g_1) > 0
\]
so $g^+$ belongs to the support of at least one IFH from $\structA$ to $\structA$, which concludes the proof.
\end{proof}

\begin{lemma}
\label{lem:ret-step}
Let $\C$ be a class of valued structures. If \problem{Reduction Step($\C$)} is in
PTIME, then \problem{Core Computation($\C$)} is in PTIME.
\end{lemma}

\begin{proof}
Suppose that \problem{Reduction Step($\C$)} is in PTIME, and let $\structA \in
  \C$. We initialize a variable $g: A \mapsto A$ to the identity mapping on $A$
  and invoke the polynomial-time algorithm R for \problem{Reduction Step($\C$)}
  on input $(\structA,g)$. If R asserts that no mapping $g^+: A \mapsto A$ such
  that $g^+(A) \subsetneq g(A)$ and $g^+$ belongs to the support of some IFH from $\structA$ to $\structA$ exists, then from Proposition~\ref{prop:reduc} we deduce that $g(\structA)$ is a core and we are done. Otherwise, we set $g \defeq g^+$ and repeat the procedure until R finds that $g(\structA)$ is a core (via Proposition~\ref{prop:reduc}). In this case, by Proposition~\ref{prop:easy} (3) from Appendix~\ref{app:equiv}, it holds that $g(\structA) \equiv \structA$, so $g(\structA)$ is the core of $\structA$ and we return $g$. The procedure terminates after at most $|A|$ calls to R.
\end{proof}

We start by proving the implication (1) $\Rightarrow$ (2) of Theorem~\ref{thm:search}. If \problem{SVCSP($\C$, $-$)} is in PTIME, then, by Theorem~\ref{theo:main} (and under the assumption that FPT $\neq$ W[1]), $\C$ is of bounded treewidth modulo equivalence; the nontrivial part is to show that \problem{Core Computation($\C$)} is in PTIME. To achieve this, we adapt an algorithm from~\cite[Proposition~4.7]{tz16:jacm} originally used to determine in polynomial time the complexity of core finite-valued constraint languages (here, ``core'' refers to the notion for right-hand side valued structures, which differs from our own; see~\cite[Definition~2.6]{tz16:jacm} for a precise definition). The central idea is to show that \problem{Reduction Step($\C$)} can be solved by the \emph{ellipsoid algorithm} using a separation oracle that makes polynomially many calls to the assumed polynomial-time algorithm for \problem{SVCSP($\C$, $-$)}.

Before we proceed with the main proof we need the following definitions and result from combinatorial optimisation, as well as two minor technical lemmas.

\begin{definition}[{\cite{grotschel88:geom}}]
Let $A \in \mathbb{Q}^{m \times n}$, $\overline{b} \in \mathbb{Q}^m$ and $P = \{\overline{x} \in \mathbb{Q}^n: \, A\overline{x} \leq \overline{b}\}$. A \emph{strong separation oracle} for $P$ is an algorithm that, given on input a vector $\overline{y} \in \mathbb{Q}^n$, either concludes that $\overline{y} \in P$ or computes a vector $\overline{a} \in \mathbb{Q}^n$ such that $\overline{a}^T \overline{y} > \overline{a}^T \overline{x}$ for all $\overline{x} \in P$.
\end{definition}

\begin{definition}[{\cite{grotschel88:geom}}]
Let $A \in \mathbb{Q}^{m \times n}$, $\overline{b} \in \mathbb{Q}^m$, $P = \{\overline{x} \in \mathbb{Q}^n: \, A\overline{x} \leq \overline{b}\}$, $\overline{c} \in \mathbb{Q}^n$ and SEP be a strong separation oracle for $P$. A \emph{basic optimum dual solution with oracle inequalities} is a set of inequalities $\overline{a}_1^T \overline{x} \leq \alpha_1, \ldots, \overline{a}_k^T \overline{x} \leq \alpha_k$ valid for $P$, where $\overline{a}_1, \ldots, \overline{a}_k$ are linearly independent outputs of SEP, and dual variables $\lambda_1, \ldots, \lambda_k \in \qplus$ such that $\lambda_1\overline{a}_1 + \ldots + \lambda_k\overline{a}_k = \overline{c}$ and $\lambda_1 \alpha_1 + \ldots + \lambda_k \alpha_k = \max_{\overline{x} \in P}\overline{c}^T \overline{x}$.
\end{definition}

\begin{lemma}[{\cite[Lemma~6.5.15]{grotschel88:geom}}]
\label{lem:ellipsoid}
Let $A \in \mathbb{Q}^{m \times n}$, $\overline{b} \in \mathbb{Q}^m$, $P = \{\overline{x} \in \mathbb{Q}^n: \, A\overline{x} \leq \overline{b}\}$ and $\overline{c} \in \mathbb{Q}^n$. Suppose that the bit sizes of the coefficients of $A$ and $\overline{b}$ are bounded by $\phi$. Given a strong separation oracle SEP for $P$ where every output has encoding size at most $\phi$, one can, in time polynomial in $n$, $\phi$ and the encoding size of $\overline{c}$, and with polynomially many oracle queries to SEP, either
\begin{itemize}
\item find a basic optimum dual solution with oracle inequalities, or
\item assert that the dual problem is unbounded or has no solution.
\end{itemize}
\end{lemma}

\begin{lemma}
\label{lem:epsilondelta}
There exists a polynomially computable function $\varepsilon_{\Delta}$ which maps any two valued $\sigma$-structures $\structA$, $\structB$ to a positive rational number $\varepsilon_{\Delta}(\structA,\structB)$ such that for any two mappings $h_1,h_2 : A \mapsto B$ satisfying $\cost{h_1} < \cost{h_2} < \infty$, we have $\cost{h_2} - \cost{h_1} > \varepsilon_{\Delta}(\structA,\structB)$.
\end{lemma}

\begin{proof}
We assume without loss of generality that every nonnegative rational number $p/q$ is encoded as a sequence of two nonnegative integers $p$ and $q$. Let $q_{\structA}$ (resp. $q_{\structB}$) denote the product of all denominators of the values $f^{\structA}(\tuple{x})$ with $(f,\tuple{x}) \in \tup{\structA}_{<\infty}$ (resp. values $f^{\structB}(\tuple{x})$ with $(f,\tuple{x}) \in \tup{\structB}_{<\infty}$). The cost of any mapping $h : A \to B$ is an integer multiple of $1/(q_{\structA}q_{\structB})$, so for any two mappings $h_1,h_2 : A \mapsto B$ satisfying $\cost{h_1} < \cost{h_2} < \infty$ we have that $\cost{h_2} - \cost{h_1} \geq 1/(q_{\structA}q_{\structB})$. We can then choose $\varepsilon_{\Delta}(\structA,\structB) \defeq 1/2 \cdot 1/(q_{\structA}q_{\structB})$, which is positive and polynomially computable.
\end{proof}

\begin{lemma}
\label{lem:epsilonomega}
There exists a polynomially computable function $\varepsilon_{\Omega}$ which maps any two valued $\sigma$-structures $\structA$, $\structB$ to a positive rational number $\varepsilon_{\Omega}(\structA,\structB)$ such that for any subset $\H_S$ of $B^A$, the following statements are equivalent:
\begin{itemize}
\item[(i)] There exists an IFH $\omega$ from $\structA$ to $\structB$ such that $\sum_{h \in \H_S}\omega(h) > 0$.
\item[(ii)] There exists an IFH $\omega$ from $\structA$ to $\structB$ such that $\sum_{h \in \H_S}\omega(h) \geq \varepsilon_{\Omega}(\structA,\structB)$.
\end{itemize}
\end{lemma}

\begin{proof}
Let $\phi$ be a nondecreasing polynomial such that every feasible linear program with encoding size $n$ has an optimum value with encoding size at most $\phi(n)$. (Since linear programming is solvable in polynomial time, such a polynomial exists and depends on the encoding scheme chosen; for more details we refer the reader to~\cite{Schrijver86:ILP}.) Let $f_{\min}$ be the polynomially computable function that maps each natural number $n \geq 2$ to the smallest positive rational number with encoding size at most $n$. We define the set
\begin{align*}
 \H_{<\infty}& \defeq\{h\in B^A: f^{\structB}(\tuple{x})<\infty \Rightarrow f^\structA(h^{-1}(\tuple{x}))<\infty, \text{ for all } (f,\tuple{x})\in \tup{\structB}\}
\end{align*}
and observe that the statement (i) is true if and only if the linear program
\begin{equation}
\label{eq:lbprimal}
\begin{aligned}
&\max \left(\sum_{h \in \H_S \cap \H_{<\infty}}\omega(h)\right) &\\
&\sum_{h \in \H_{<\infty}}\omega(h)f^{\structA}(h^{-1}(\tuple{x})) \leq f^{\structB}(\tuple{x}) & \forall (f,\tuple{x}) \in \tup{\structB}_{<\infty}\\
&\sum_{h \in \H_{<\infty}}\omega(h) = 1 &\\
&\omega(h) \geq 0 & \forall h \in \H_{<\infty}
\end{aligned}
\end{equation}
is feasible and its optimum value is positive. This program has polynomially many inequalities, hence if it is feasible then there exists an optimum solution $\omega^*$ such that $\supp{\omega^*}$ has polynomial size. The restriction of the linear program~(\ref{eq:lbprimal}) to the variables in $\supp{\omega^*}$ has encoding size at most $p(|\structA| + |\structB|)$ for some polynomial $p$, and has the same optimum value. Now, we define
\[
\varepsilon_{\Omega}(\structA,\structB) \defeq f_{\min}(\phi(p(|\structA| + |\structB|)))
\]
and we observe that if the linear program~(\ref{eq:lbprimal}) is feasible and its optimum value is positive, then it is at least $\varepsilon_{\Omega}(\structA,\structB)$. This establishes the implication (i) $\Rightarrow$ (ii) for the function $\varepsilon_{\Omega}$. The implication (ii) $\Rightarrow$ (i) is trivial and given $\structA, \structB$ the function $\varepsilon_{\Omega}(\structA,\structB)$ is polynomially computable, so the claim follows.
\end{proof}

The following lemma is the main technical part in establishing (1) $\Rightarrow$
(2) in Theorem~\ref{thm:search}.

\begin{lemma}
\label{lem:compute-retraction}
Let $\C$ be a class of valued structures. If \problem{SVCSP($\C$, $-$)} is in
PTIME then \problem{Reduction Step($\C$)} is in PTIME as well.
\end{lemma}

\begin{proof}
Let $(\structA,g)$ be an input to \problem{Reduction Step($\C$)}. We assume
that $|g(A)| > 1$; otherwise the problem is trivial. If $\structA$ is $\{0,\infty\}$-valued then we can solve the instance by following the argument described in Remark~\ref{rem:crispsearch}, so we also assume that $\max_{(f,\tuple{x}) \in \tup{\structA}_{<\infty}}(f^{\structA}(\tuple{x})) > 0$. We define the set
\begin{align*}
 \H_{<\infty}& \defeq\{h\in A^A: f^{\structA}(\tuple{x})<\infty \Rightarrow f^\structA(h^{-1}(\tuple{x}))<\infty, \text{ for all } (f,\tuple{x})\in \tup{\structA}\}
\end{align*}
and we recall that for every IFH $\omega$ from $\structA$
to $\structA$, every mapping $g' \in \supp{\omega}$ belongs to $\H_{<\infty}$. We
denote by $\H^*$ the set of all mappings $h$ in $\H_{<\infty}$ such that $h(A)
\subsetneq g(A)$. Let us consider the following linear program:
\begin{equation}
\label{eq:newprimal}
\begin{aligned}
&\min 0&\\
&\sum_{h \in \H_{<\infty}}\omega(h)f^{\structA}(h^{-1}(\tuple{x})) \leq f^{\structA}(\tuple{x}) & \forall (f,\tuple{x}) \in \tup{\structA}_{<\infty}\\
&\sum_{h \in \H_{<\infty}}\omega(h) = 1 &\\
&\sum_{h \in \H^*}\omega(h) \geq \varepsilon_{\Omega}(\structA,\structA)\\
&\omega(h) \geq 0 & \forall h \in \H_{<\infty}
\end{aligned}
\end{equation}
By Lemma~\ref{lem:epsilonomega}, this program is not feasible if and only if $g(\structA)$ is a core; otherwise a solution gives a mapping $g \in \H^*$ with the desired property. This program has exponentially many variables, and hence we will solve its dual instead:
\begin{equation}
\label{eq:newdual}
\begin{aligned}
&\max \delta_1 + \varepsilon_{\Omega}(\structA,\structA) \delta_2 &\\
&\sum_{(f,\tuple{x}) \in \tup{\structA}_{<\infty}}z(f,\tuple{x})\left(f^{\structA}(\tuple{x}) - f^{\structA}(h^{-1}(\tuple{x}))\right) + \delta_1 \leq 0 & \forall h \in \H_{<\infty} \backslash \H^*\\
&\sum_{(f,\tuple{x}) \in \tup{\structA}_{<\infty}}z(f,\tuple{x})\left(f^{\structA}(\tuple{x}) - f^{\structA}(h^{-1}(\tuple{x}))\right) + \delta_1 + \delta_2 \leq 0 & \forall h \in \H^*\\
&z(f,\tuple{x}) \geq 0 & \forall (f,\tuple{x}) \in \tup{\structA}_{<\infty}\\
&\delta_2 \geq 0
\end{aligned}
\end{equation}

We now describe a strong separation oracle for the associated polyhedron $P$.
Given a vector $(z,\delta_1,\delta_2) \in \mathbb{Q}^{|\tup{\structA}_{<\infty}|} \times
\mathbb{Q}^2$, we first check if there exists $(f^*,\tuple{x}^*) \in
\tup{\structA}_{<\infty}$ such that $z(f^*,\tuple{x}^*) < 0$; if it is the case
then $a(f^*,\tuple{x}^*) = -1$ and $0$ otherwise defines a separating
hyperplane. Similarly, if $\delta_2 < 0$ then we can set $a(\delta_2) = -1$ and $0$ otherwise to obtain a separation.

Now, let $\structB_0$ be a valued $\sigma$-structure with
universe $B_0 = A$ such that for all $(f,\tuple{x}) \in \tup{\structB_0}$,
$f^{\structB_0}(\tuple{x}) = z(f,\tuple{x})$ if $(f,\tuple{x}) \in
\tup{\structA}_{<\infty}$ and $f^{\structB_0}(\tuple{x}) = 0$ otherwise. We will be interested in computing a mapping $h \in B_0^A \cap \H_{<\infty}$ with minimum cost. However, invoking the assumed algorithm for \problem{SVCSP($\C$, $-$)} on the instance $(\structA,\structB_0)$ is not sufficient for this task because the returned mapping might not belong to $\H_{<\infty}$. In order to solve this problem, we define
\[
\varepsilon \defeq 
\frac{\varepsilon_{\Delta}(\structA,\structB_0)}{|\tup{\structA}| \cdot \max_{(f,\tuple{x}) \in \tup{\structA}_{<\infty}}(f^{\structA}(\tuple{x}))}
\]
where the function $\varepsilon_\Delta$ is as in Lemma~\ref{lem:epsilondelta}, and we let $\structB$ be the valued $\sigma$-structure with
universe $B = A$ such that for all $(f,\tuple{x}) \in \tup{\structB}$,
$f^{\structB}(\tuple{x}) = z(f,\tuple{x}) + \varepsilon$ if $(f,\tuple{x}) \in
\tup{\structA}_{<\infty}$ and $f^{\structB}(\tuple{x}) = 0$ otherwise. Because $\varepsilon > 0$, the set of finite-cost mappings $A \mapsto B$ is precisely $\H_{<\infty}$ and cannot be empty because of the identity mapping. Since
$(\structA,\structB)$ is an instance of \problem{SVCSP($\C$, $-$)}, we can find
a mapping $h^*$ from $A$ to $B$ of minimum cost in polynomial time. Then, we have
\begin{align*}
h^* &\in \text{argmin}_{h \in B^A} \left( \sum_{(f,\tuple{x}) \in \tup{\structB}}f^{\structB}(\tuple{x})f^{\structA}(h^{-1}(\tuple{x})) \right)\\
&= \text{argmin}_{h \in \H_{<\infty}} \left( \sum_{(f,\tuple{x}) \in \tup{\structB}}f^{\structB}(\tuple{x})f^{\structA}(h^{-1}(\tuple{x})) \right)\\
&= \text{argmin}_{h \in \H_{<\infty}} \left( \sum_{(f,\tuple{x}) \in \tup{\structA}_{<\infty}}z(f,\tuple{x})f^{\structA}(h^{-1}(\tuple{x})) + \varepsilon \left( \sum_{(f,\tuple{x}) \in \tup{\structA}_{<\infty}}f^{\structA}(h^{-1}(\tuple{x}))\right) \right).
\end{align*}

Recall that we are looking for a mapping in $B_0^A \cap \H_{<\infty}$ with minimum cost with respect to the instance $(\structA, \structB_0)$. The mapping $h^*$ we have just computed realises the minimum of a slightly different objective function, but the next claim shows that this is not an issue.

\begin{claim}
\label{claim:argmin}
Let $\H_S$ be a non-empty subset of $\H_{<\infty}$ and $h' \in \H_S$. If 
\[
h' \in \text{argmin}_{h \in \H_S}\left( \sum_{(f,\tuple{x}) \in \tup{\structA}_{<\infty}}z(f,\tuple{x})f^{\structA}(h^{-1}(\tuple{x})) + \varepsilon \left( \sum_{(f,\tuple{x}) \in \tup{\structA}_{<\infty}}f^{\structA}(h^{-1}(\tuple{x}))\right) \right)
\]
then
\[
h' \in \text{argmin}_{h \in \H_S} \left( \sum_{(f,\tuple{x}) \in \tup{\structA}_{<\infty}}z(f,\tuple{x})f^{\structA}(h^{-1}(\tuple{x})) \right).
\]
\end{claim}

\begin{proof}
We prove the statement by contraposition. Let $h' \in \H_S$ and suppose that there exists $h_1 \in \H_S$ such that \[
\sum_{(f,\tuple{x}) \in \tup{\structA}_{<\infty}}z(f,\tuple{x})f^{\structA}(h_1^{-1}(\tuple{x})) < \sum_{(f,\tuple{x}) \in \tup{\structA}_{<\infty}}z(f,\tuple{x})f^{\structA}(h'^{-1}(\tuple{x})).
\] 
These two quantities are the respective costs of $h_1$ and $h'$ with respect to the instance $(\structA, \structB_0)$, and both are finite. By Lemma~\ref{lem:epsilondelta} we have
\[
\sum_{(f,\tuple{x}) \in \tup{\structA}_{<\infty}}z(f,\tuple{x})f^{\structA}(h_1^{-1}(\tuple{x})) + \varepsilon_\Delta(\structA,\structB_0) < \sum_{(f,\tuple{x}) \in \tup{\structA}_{<\infty}}z(f,\tuple{x})f^{\structA}(h'^{-1}(\tuple{x}))
\] 
and since 
\[
\frac{\sum_{(f,\tuple{x}) \in \tup{\structA}_{<\infty}}f^{\structA}(h_1^{-1}(\tuple{x}))} {|\tup{\structA}| \cdot \max_{(f,\tuple{x}) \in \tup{\structA}_{<\infty}}(f^{\structA}(\tuple{x}))} \leq 1
\]
it follows from the definition of $\varepsilon$ that 
\[
\sum_{(f,\tuple{x}) \in \tup{\structA}_{<\infty}}z(f,\tuple{x})f^{\structA}(h_1^{-1}(\tuple{x})) + \varepsilon \left( \sum_{(f,\tuple{x}) \in \tup{\structA}_{<\infty}}f^{\structA}(h_1^{-1}(\tuple{x}))\right)
\]
is strictly smaller than $\sum_{(f,\tuple{x}) \in \tup{\structA}_{<\infty}}z(f,\tuple{x})f^{\structA}(h'^{-1}(\tuple{x}))$, and a fortiori strictly smaller than $\sum_{(f,\tuple{x}) \in \tup{\structA}_{<\infty}}z(f,\tuple{x})f^{\structA}(h'^{-1}(\tuple{x})) + \varepsilon \left( \sum_{(f,\tuple{x}) \in \tup{\structA}_{<\infty}}f^{\structA}(h'^{-1}(\tuple{x}))\right)$, which concludes the proof.
\renewcommand{\qedsymbol}{$\blacksquare$}\end{proof}

As a consequence of Claim~\ref{claim:argmin} (with $\H_S = \H_{<\infty}$) we have that $h^*$ realises the maximum of the function  \[
\zeta(h) = \sum_{(f,\tuple{x}) \in \tup{\structA}_{<\infty}}z(f,\tuple{x})\left(f^{\structA}(\tuple{x}) - f^{\structA}(h^{-1}(\tuple{x}))\right) + \delta_1
\]
over $\H_{<\infty}$. If this maximum is positive, then the vector $a$ such that $a(f,\tuple{x}) = f^{\structA}(\tuple{x}) - f^{\structA}({h^*}^{-1}(\tuple{x}))$, $a(\delta_1) = 1$ and $a(\delta_2) = 0$ defines a separating hyperplane (this is true even if $h^* \in \H^*$ because $\delta_2$ is nonnegative). In this case we output $a$ together with the mapping $h^*$. (A separation oracle is supposed to only output the vector $a$, but later on we will need to know to which mapping it corresponds.) Otherwise, only two possibilities remain: either $(z,\delta_1,\delta_2) \in P$, or the maximum of $\sum_{(f,\tuple{x}) \in \tup{\structA}_{<\infty}}z(f,\tuple{x})\left(f^{\structA}(\tuple{x}) - f^{\structA}(h^{-1}(\tuple{x}))\right) + \delta_1 + \delta_2$ over $\H^*$ is positive. 

To verify the latter condition, for every $a \in g(A)$ we construct a valued $\sigma$-structure $\structB_a$ with
universe $B_a = g(A) \backslash \{a\}$ such that for all $(f,\tuple{x}) \in \tup{\structB_a}$,
$f^{\structB_a}(\tuple{x}) = z(f,\tuple{x}) + \varepsilon$ if $(f,\tuple{x}) \in
\tup{\structA}_{<\infty}$ and $f^{\structB_a}(\tuple{x}) = 0$ otherwise. Then, for each $a \in g(A)$ we use the algorithm for \problem{SVCSP($\C$, $-$)} to compute in polynomial time a minimum-cost mapping $h^*_a$ for the instance $(\structA,\structB_a)$. If the cost of $h^*_a$ is infinite for each $a \in g(A)$ then $\H^*$ is empty; it follows that $g(\structA)$ is a core and we can stop. Otherwise, for each $a \in g(A)$ such that the cost of $h^*_a$ is finite we have
\begin{align*}
h^*_a &\in \underset{h \in B_a^A}{\text{argmin}} \left( \sum_{(f,\tuple{x}) \in \tup{\structB_a}}f^{\structB_a}(\tuple{x})f^{\structA}(h^{-1}(\tuple{x})) \right)\\
&= \underset{h \in \H_{<\infty} : h(A) \subseteq g(A) \backslash \{a\}}{\text{argmin}} \left( \sum_{(f,\tuple{x}) \in \tup{\structA}_{<\infty}}z(f,\tuple{x})f^{\structA}(h^{-1}(\tuple{x})) + \varepsilon \left( \sum_{(f,\tuple{x}) \in \tup{\structA}_{<\infty}}f^{\structA}(h^{-1}(\tuple{x}))\right) \right)
\end{align*}
and by Claim~\ref{claim:argmin},
\[
h^*_a \in \underset{h \in \H_{<\infty} : h(A) \subseteq g(A) \backslash \{a\}}{\text{argmin}} \left( \sum_{(f,\tuple{x}) \in \tup{\structA}_{<\infty}}z(f,\tuple{x})f^{\structA}(h^{-1}(\tuple{x})) \right).
\]
Therefore, if we pick $h^{**} \in \text{argmin}_{h^*_a : a \in g(A)} \left(\sum_{(f,\tuple{x}) \in \tup{\structA}_{<\infty}}z(f,\tuple{x})f^{\structA}(h^{-1}(\tuple{x})) \right)$ then we have
\begin{align*}
h^{**} &\in \underset{h \in \H_{<\infty}: h(A) \subset g(A)}{\text{argmin}} \left( \sum_{(f,\tuple{x}) \in \tup{\structA}_{<\infty}}z(f,\tuple{x})f^{\structA}(h^{-1}(\tuple{x})) \right)
\end{align*}
and either $\sum_{(f,\tuple{x}) \in \tup{\structA}_{<\infty}}z(f,\tuple{x})\left(f^{\structA}(\tuple{x}) - f^{\structA}(h^{**-1}(\tuple{x}))\right) + \delta_1 + \delta_2 > 0$ and the vector $a$ such that $a(f,\tuple{x}) = f^{\structA}(\tuple{x}) - f^{\structA}({h^{**}}^{-1}(\tuple{x}))$, $a(\delta_1) = 1$ and $a(\delta_2) = 1$ defines a separating hyperplane (which we output together with the mapping $h^{**}$) or $(z,\delta_1,\delta_2) \in P$. This concludes the description of our strong separation oracle.

We now apply Lemma~\ref{lem:ellipsoid} to the linear program~\eqref{eq:newdual}. Its dual~\eqref{eq:newprimal} is bounded, and if the ellipsoid algorithm returns that it is not feasible then $g(\structA)$ is a core. Otherwise, the algorithm will return a set of polynomially many valid inequalities of the form
\begin{align*}
&\sum_{(f,\tuple{x}) \in \tup{\structA}_{<\infty}}z(f,\tuple{x})\left(f^{\structA}(\tuple{x}) - f^{\structA}(h^{-1}(\tuple{x}))\right) + \delta_1 \leq \alpha'_h & \forall h \in \H' \subseteq (\H_{<\infty} \backslash \H^*)\\
&\sum_{(f,\tuple{x}) \in \tup{\structA}_{<\infty}}z(f,\tuple{x})\left(f^{\structA}(\tuple{x}) - f^{\structA}(h^{-1}(\tuple{x}))\right) + \delta_1 + \delta_2 \leq \alpha''_h & \forall h \in \H'' \subseteq \H^*\\
&- z(f,\tuple{x}) \leq \alpha_{f,\tuple{x}} & \forall (f,\tuple{x}) \in T \subseteq \tup{\structA}_{<\infty}\\
&- \delta_2 \leq \alpha_2
\end{align*}
where each mapping $h$ appearing in the inequalities is explicitly known (because we modified the output of the strong separation oracle), and dual variables that satisfy
\begin{equation*}
\begin{aligned}
&\sum_{h \in \H'}\lambda'_h \left(f^{\structA}(\tuple{x}) - f^{\structA}(h^{-1}(\tuple{x}))\right) + \sum_{h \in \H''}\lambda''_h \left(f^{\structA}(\tuple{x}) - f^{\structA}(h^{-1}(\tuple{x}))\right) - \lambda_{f,\tuple{x}} = 0& \forall (f,\tuple{x}) \in \tup{\structA}_{<\infty}\\
&\sum_{h \in \H'}\lambda'_h + \sum_{h \in \H''}\lambda''_h = 1&\\
&\sum_{h \in \H''}\lambda''_h - \lambda_2 = \varepsilon_{\Omega}(\structA,\structA)&
\end{aligned}
\end{equation*}
where we set $\lambda_{f,\tuple{x}} \defeq 0$ if $(f,\tuple{x}) \notin T$. Now, we define
\[
\omega(h) \defeq
\begin{cases}
\lambda'_h &\quad \text{if $h \in \H'$}\\
\lambda''_h &\quad \text{if $h \in \H''$}\\
0 &\quad \text{otherwise}
\end{cases}
\]
and we deduce from the above system that
\begin{align}
&\sum_{h \in \H_{<\infty}}\omega(h) \left(f^{\structA}(\tuple{x}) - f^{\structA}(h^{-1}(\tuple{x}))\right) - \lambda_{f,\tuple{x}} = 0& \forall (f,\tuple{x}) \in \tup{\structA}_{<\infty}\label{eq:wvar1}\\
&\sum_{h \in \H_{<\infty}}\omega(h) = 1& \label{eq:wvar2}\\
&\sum_{h \in \H^*}\omega(h) \geq \varepsilon_{\Omega}(\structA,\structA) > 0& \label{eq:wvar3}
\end{align}
Then, from~\eqref{eq:wvar1},~\eqref{eq:wvar2} and the nonnegativity of the dual variables we can deduce that for all $(f,\tuple{x}) \in \tup{\structA}_{<\infty}$,
\[\sum_{h \in \H_{<\infty}}\omega(h)f^{\structA}(h^{-1}(\tuple{x})) \leq f^{\structA}(\tuple{x})\]
and hence $\omega$ (complemented with $\omega(h)=0$ for all $h \notin \H_{<\infty}$) is an
IFH with a support of polynomial size. Finally, we
search $\supp{\omega}$ for a mapping $g^+$ such that $g^+(A) \subsetneq g(A)$, which
is guaranteed to exist by the definition of $\omega$ and~\eqref{eq:wvar3}.
\end{proof}

The next lemma is the last missing ingredient in the proof of
Theorem~\ref{thm:search}.

\begin{lemma}
\label{lem:twms-search}
Let $k \geq 1$ and $\C$ be a class of valued $\sigma$-structures such that for every $\structA \in \C$, $\twms{\structA} \leq k-1$ and the overlap of $\structA$ is at most $k$. Then, \problem{SVCSP($\C$, $-$)} is in PTIME. 
\end{lemma}

\begin{proof}
Given a pair $(\structA,\structB)$ of valued structures over some signature $\sigma$ and two elements $a \in A$, $b \in B$, we say that $a$ is \emph{fixed} to $b$ if there exists a unary symbol $f_{ab} \in \sigma$ such that $f_{ab}^{\structA}(\tuple{x}) = \infty$ if $\tuple{x} = (a)$ and $0$ otherwise, and $f_{ab}^{\structB}(\tuple{x}) = 0$ if $\tuple{x} = (b)$ and $\infty$ otherwise. Clearly, if $a$ is fixed to $b$ then any finite-cost mapping from $\structA$ to $\structB$ must map $a$ to $b$.

Now, let $(\structA,\structB)$ be an instance of \problem{SVCSP($\C$, $-$)}. Suppose that $\opt{\structA,\structB}$ is finite and that the overlap of $\structA$ is positive (otherwise the problem is trivial). If every element $a \in A$ is fixed to some element $g(a) \in B$, then the mapping $g$ is the only finite-cost mapping from $\structA$ to $\structB$ and we can output $g$. Otherwise, we pick an element $a \in A$ that is not fixed to any element of $B$. For each $b \in B$, we create one new instance $(\structA_{ab},\structB_{ab})$ of SVCSP by adding to $(\structA,\structB)$ a new symbol $f_{ab}$ that fixes $a$ to $b$. The overlap and treewidth modulo scopes of $\structA_{ab}$ are no greater than those of $\structA$, so we can use Theorem~\ref{thm:suff} to compute $\opt{\structA_{ab},\structB_{ab}}$ in polynomial time. 

Observe that there necessarily exists a value $b^* \in B$ such that $\opt{\structA,\structB} = \opt{\structA_{ab^*},\structB_{ab^*}}$ (just take $b^* = h(a)$ for any minimum-cost mapping $h$ from $\structA$ to $\structB$). The instance $(\structA_{ab^*},\structB_{ab^*})$ has the same optimum cost as $(\structA,\structB)$, but one extra element of $A$ is fixed. We can then repeat the operation until all elements of the left-hand side structure are fixed, in which case the unique finite-cost mapping (which can be found in polynomial time by inspecting the signature) is a minimum-cost mapping from $\structA$ to $\structB$.
\end{proof}

\begin{proof}[Proof of Theorem~\ref{thm:search}]
If $(1)$ holds, then, by Theorem~\ref{theo:main}, $\C$ has bounded treewidth
modulo equivalence. Furthermore, by Lemma~\ref{lem:compute-retraction},
\problem{Reduction Step($\C$)} is in PTIME and, by Lemma~\ref{lem:ret-step},
\problem{Core Computation($\C$)} is in PTIME as well.
For the converse implication, assume that $(2)$ holds and let $(\structA,\structB)$ be an instance of \problem{SVCSP($\C$, $-$)}. Let $k$ denote the maximum treewidth of the core of a structure in $\C$. 
We compute in polynomial time the core $g(\structA)$ of $\structA$ and the associated mapping $g:A \mapsto A$. Because $\tw{g(\structA)} \leq k$ implies both $\twms{g(\structA)} \leq k$ and an upper bound of $k+1$ for the overlap of $g(\structA)$, we can use  Lemma~\ref{lem:twms-search} to compute in polynomial time a minimum-cost solution $h: g(A)\mapsto B$ to $(g(\structA),\structB)$. 
Since $\costb{g(\structA)\mapsto\structB}{h} = \costb{\structA\mapsto\structB}{h\circ g}$ and $\opt{\structA,\structB}= \opt{g(\structA),\structB}$, we conclude that $h \circ g$ is a minimum-cost mapping from $A$ to $B$ and the theorem follows.
\end{proof}

\section{Related problems}
\label{sec:meta}

In this section we provide a quick overview of the complexity of deciding the various natural questions on valued structures that arise from our characterisations. 
We also highlight some interesting implications of our results in the context of database theory.

We establish tight complexity bounds for the following problems. We note that
while hardness mostly follows directly from existing results on relational structures, the technical machinery of Section~\ref{sec:equiv} is required in order to derive precise upper bounds.

\begin{itemize}
\item \problem{Improvement}: given two valued structures $\structA,\structB$, is it true that $\structA \less \structB$?
\item \problem{Equivalence}: given two valued structures $\structA,\structB$, is it true that $\structA \equiv \structB$?
\item \problem{Core Recognition}: given a valued structure $\structA$, is $\structA$ a core?
\item \problem{Core Treewidth}: given a valued structure $\structA$ and $k\geq 1$, is the treewidth of the core of $\structA$ at most $k$?
\item \problem{Sherali-Adams Tightness}: given a valued structure $\structA$ and $k\geq 1$, is the Sherali-Adams relaxation of level $k$ tight for $\structA$?
\end{itemize}

\begin{proposition}
\label{prop:imp-equiv}
\problem{Improvement} and \problem{Equivalence} are NP-complete.
\end{proposition}

\begin{proof}
We first prove that \problem{Improvement} is in NP, which implies that
  \problem{Equivalence} is in NP as well. By Proposition~\ref{prop:charfrac}, an
  instance $(\structA,\structB)$ of \problem{Improvement} is a yes-instance if
  and only if there exists an IFH from $\structA$ to $\structB$, or equivalently if $\G_{<\infty} \defeq \{g\in B^A\mid \text{$f^\structA(g^{-1}(\tuple{x}))<\infty$ for all $(f,\tuple{x})\in \tup{\structB}_{<\infty}$}\}$ (where $\tup{\structB}_{<\infty}\defeq\{(f,\tuple{x})\in \tup{\structB}\mid f^{\structB}(\tuple{x})<\infty\}$) is not empty and the system
\begin{align*}
&\sum_{g \in \G_{<\infty}}\omega(g) f^{\structA}(g^{-1}(\tuple{x})) \leq f^{\structB}(\tuple{x}) &\forall (f,\tuple{x})\in \tup{\structB}_{<\infty}\\
&\sum_{g \in \G_{<\infty}}\omega(g) \leq 1 &\\
&-\sum_{g \in \G_{<\infty}}\omega(g) \leq -1 &\\
&-\omega(g) \leq 0 &\forall g \in \G_{<\infty}\\
\end{align*}
has a rational solution $\omega$. Since the number of inequalities is polynomial in $|\structA|$ and $|\structB|$, this system has a solution if and only if it has one with a polynomial number of non-zero variables. Such a subset of non-zero variables is an NP certificate: the corresponding restriction of the system has polynomial size and its satisfiability can be checked in polynomial time.

For hardness, we note that in the special case of $\{0,\infty\}$-valued structures (that is, relational structures) the  \problem{Improvement} and \problem{Equivalence} problems correspond respectively to \problem{Homomorphism} and \problem{Homomorphic Equivalence}, which are well-known to be NP-complete even in the bounded arity case~\cite{Chandra77:stoc}.
\end{proof}

\begin{proposition}
\problem{Core Recognition} is coNP-complete.
\end{proposition}

\begin{proof}
We start by establishing membership in coNP. Let $(\structA)$ be an instance of
  \problem{Core Recognition}. By the definition of a core, $(\structA)$ is a
  no-instance if and only if there exists a non-surjective IFH from $\structA$ to $\structA$. This is true if and only if the optimum of the linear program
\begin{equation*}
\begin{aligned}
&\min \left(-\sum_{g \in \G^*}\omega(g)\right) &\\
&\sum_{g \in \G_{<\infty}}\omega(g)f^{\structA}(g^{-1}(\tuple{x})) \leq f^{\structA}(\tuple{x}) \hspace{15mm}& \forall (f,\tuple{x}) \in \tup{\structA}_{<\infty}\\
&\sum_{g \in \G_{<\infty}}\omega(g) \leq 1 &\\
&-\sum_{g \in \G_{<\infty}}\omega(g) \leq -1 &\\
&-\omega(g) \leq 0 & \forall g \in \G_{<\infty}
\end{aligned}
\end{equation*}
is strictly negative, where $\tup{\structA}_{<\infty}\defeq\{(f,\tuple{x})\in \tup{\structA}\mid f^{\structA}(\tuple{x})<\infty\}$, $\G_{<\infty} \defeq \{g\in A^A\mid \text{$f^\structA(g^{-1}(\tuple{x}))<\infty$ for all $(f,\tuple{x})\in \tup{\structA}_{<\infty}$}\}$ and $\G^*$ is the restriction of $\G_{<\infty}$ to non-surjective mappings. Again, the number of inequalities in this system is polynomial in $|\structA|$ so there exists a solution of minimum cost with a polynomial number of non-zero variables. Such a subset of variables is a coNP certificate.

On $\{0,\infty\}$-valued structures with a single binary symmetric function symbol, \problem{Core Recognition} coincides with the problem of deciding if a graph is a core in the usual sense (that is, the problem of deciding if all of its endomorphisms are surjective). This problem is coNP-complete~\cite{hell92:core}, so \problem{Core Recognition} is coNP-complete as well.
\end{proof}

\begin{proposition}
\problem{Core Treewidth} is NP-complete even for fixed $k \geq 1$, and \problem{Sherali-Adams Tightness} is NP-complete even for fixed $k \geq 1$.
\end{proposition}

\begin{proof}
First, we prove that these problems belong to NP when $k$ is part of the input. For \problem{Core Treewidth}, the certificate for a yes-instance $(\structA,k)$ is a valued structure $\structB$, a polynomially-sized certificate that $\structB \equiv \structA$ (which exists because \problem{Equivalence} is in NP by Proposition~\ref{prop:imp-equiv}) and a tree decomposition of $G(\pos{\structB})$ of width at most $k$. Correctness follows from Proposition~\ref{prop:tw-equiv-core}. For \problem{Sherali-Adams Tightness}, the certificate for a yes-instance $(\structA,k)$ is a valued structure $\structB$ whose overlap is at most $k$, a polynomially-sized certificate that $\structB \equiv \structA$ and a tree decomposition of $G(\pos{\structB})$ of width modulo scopes at most $k-1$. For correctness, by Theorem~\ref{thm:sa} it is sufficient to prove that this certificate exists if and only if the core $\structA'$ of $\structA$ has treewidth modulo scopes at most $k-1$ and overlap at most $k$. One implication is immediate: if $\structA'$ has treewidth modulo scopes at most $k-1$ and overlap at most $k$ then we can take $\structB \defeq \structA'$. For the converse implication, if this certificate exists then by Theorem~\ref{thm:suff} the Sherali-Adams relaxation of level $k$ is always tight for $\structB$. Then, by Theorem~\ref{thm:sa} the core $\structB'$ of $\structB$ has treewidth modulo scopes at most $k-1$ and overlap at most $k$. Furthermore, $\structA \equiv \structB$ so by Proposition~\ref{prop:iso} $\structA'$ and $\structB'$ are isomorphic, and finally $\structA'$ has treewidth modulo scopes at most $k-1$ and overlap at most $k$.

As before, we derive hardness from the $\{0,\infty\}$-valued case. Determining whether a graph has a core of treewidth at most $k$ is NP-complete for all fixed $k \geq 1$~\cite{Dalmau02:width}, so \problem{Core Treewidth} is NP-complete even for fixed $k \geq 1$ and arity at most $2$. For fixed $k \geq 2$ and on valued structures of arity at most $2$, \problem{Sherali-Adams Tightness} is equivalent to \problem{Core Treewidth} with $k' = k-1$, 
and hence it is NP-complete. 
For the case $k=1$ and arity at most $2$ (i.e., for directed graphs), \problem{Sherali-Adams Tightness} is equivalent to deciding 
whether the core of a directed graph is a disjoint union of \emph{oriented trees}, i.e., simple directed graphs whose underlying undirected graphs are trees. It follows from the proof of \cite[Theorem~13]{Dalmau02:width} that this problem is NP-complete. 
\end{proof}

\subsection{Application to database theory}
\label{subsec:database}

It is well-known that the evaluation/containment problem for \emph{conjunctive
queries} (CQs) (i.e., first-order queries using only conjunction and existential quantification) 
is equivalent to the homomorphism problem, and hence equivalent to CSPs \cite{Chandra77:stoc,Kolaitis98:pods}. 
This observation has been fundamental in providing principled techniques for the static analysis and optimisation of CQs. 
Indeed, in their seminal work \cite{Chandra77:stoc}, 
Chandra and Merlin exploited this connection to show that the containment and equivalence problem for CQs are NP-complete. 
They also provided tools for minimising CQs with strong theoretical guarantees. 
In terms of homomorphisms, minimising a CQ corresponds essentially to computing the (relational) core of a relational structure.

The situation is less clear in the context of annotated databases~\cite{GT17:pods}. 
In this framework, the tuples of the database are annotated with values from a particular semiring ${\cal K}$,   
and the semantics of a CQ is a value from ${\cal K}$. 
For instance, the Boolean semiring $(\{0,1\},\lor,\land,0,1)$ gives us the usual semantics of CQs, and 
the natural semiring $(\mathbb{N},+,\times,0,1)$ corresponds to the so-called \emph{bag semantics} of CQs. 
Another semiring considered in the literature is the
 \emph{tropical} semiring $(\qinf,\min,+,\infty,0)$, which provides a \emph{minimum-cost} semantics~\cite{GT17:pods}.
Unfortunately, the homomorphism machinery cannot be applied directly to the study of containment and equivalence in the semiring setting. 
While there are some works in this direction (see, e.g.~\cite{Kostylev14:tds,green11}), several basic problems remain open.
In particular, the precise complexity of containment/equivalence of CQs over the
tropical semiring is open (it was shown in~\cite{Kostylev14:tds} to be NP-hard
and in $\Pi_{2}^p$, the second level of the polynomial-time hierarchy). 
Our first observation is that these two problems are actually NP-complete. 
Indeed, it is well known that VCSP is equivalent to CQ evaluation over the tropical semiring. 
Moreover, containment and equivalence of CQs over the tropical semiring correspond to improvement and (valued) equivalence of valued structures. 
By applying Proposition~\ref{prop:imp-equiv}, we directly obtain NP-completeness of these problems. 

Our second observation is that our notion of (valued) core provides a notion of minimisation of CQs over the tropical semiring with theoretical guarantees. 
 Indeed, as the following proposition shows, the core of a valued structure is always an equivalent valued structure with minimal number of elements, or in terms of CQs, 
 with minimal number of variables. 

\begin{proposition}
Let $\structA$ and $\structB$ be valued $\sigma$-structures. Then the following are equivalent:
\begin{enumerate}
\item $\structB$ is the core of $\structA$. 
\item $\structB$ is a minimal (with respect to the size of the universe) valued structure equivalent to $\structA$. 
\end{enumerate}
\end{proposition}

\begin{proof}
For (1) $\Rightarrow$ (2), suppose that $\structB$ is the core of $\structA$. By contradiction, 
assume that (2) is false, i.e., there is a valued $\sigma$-structure $\structB'$ such that $|B'|<|B|$ and $\structB'\equiv \structA$. 
In particular, $\structB'\equiv \structB$ and then by Proposition~\ref{prop:iso} the core $\structB''$ of $\structB'$ is isomorphic to $\structB$. 
By Proposition~\ref{prop:exist-core}, we have that $|B''|\leq |B'|<|B|$; a contradiction. For (2) $\Rightarrow$ (1), suppose by contradiction that 
$\structB$ is not the core of $\structA$. Since $\structB\equiv \structA$ by hypothesis, the only possibility is that $\structB$ is not a core. 
Hence, there is an IFH $\omega$ and a non-surjective mapping $g:B\mapsto B$ such that $g\in \supp{\omega}$. 
By Proposition~\ref{prop:easy} (3) from Appendix~\ref{app:equiv}, $g(\structB)\equiv \structB\equiv \structA$. Since $|g(B)|<|B|$, this is a contradiction. 
\end{proof}

Proposition~\ref{prop:exist-core} also gives an algorithm to compute the core of a CQ over the tropical semiring. (In fact, a PSPACE algorithm.)
Finally, it is worth mentioning that our classification result from Theorem~\ref{theo:main} can be interpreted as a characterisation of the classes of CQs over the tropical semiring that can be evaluated 
in PTIME. 

\section*{Acknowledgements}

We would like to thank the anonymous referees of both the conference~\cite{crz18:focs}
and this full version of the paper. In particular, one of the journal referees
suggested a more concise way of presenting some of the results, which
shortened the paper.

\bibliographystyle{alpha}
\bibliography{crz21struct}

\appendix

\section{Missing Proofs from Section~\ref{sec:equiv}}
\label{app:equiv}

We prove several propositions and show how they establish results from
Section~\ref{sec:equiv}.

\begin{proposition}\label{prop:easy}
  ~\begin{enumerate}[(1)]
  \item For every valued $\sigma$-structures $\structA$ and $\structB$ with
    $A=B$, $\structA\vecpred\structB$ implies $\structA\less\structB$.
  \item 
    Let $\structA$ be a valued $\sigma$-structure, $g\in B^A$, and $\omega$
    a probability distribution on $B^A$. Then, $\structA\less g(\structA)$ and
    $\structA\less\omega(\structA)$.
  \item If $\omega$ is an IFH from $\structA$ to $\structA$ then
    $\structA\equiv\omega(\structA)\equiv g(\structA)$ for every
    $g\in\supp{\omega}$.
\end{enumerate}
\end{proposition}
\begin{proof}

~\begin{enumerate}[(1)]
\item $\structA\vecpred\structB$ implies, for every $\structC$ and $g:A\mapsto
  C$, $\costb{\structA\mapsto\structC}{g}\leq\costb{\structB\mapsto\structC}{g}$. Thus, $\opt{\structA,\structC}\leq\opt{\structB,\structC}$ for every $\structC$, and so $\structA\less\structB$.
\item Since
  $\costb{g(\structA)\mapsto\structC}{h}=\costb{\structA\mapsto\structC}{h\circ
    g}$, for every $\structC$ and $h:B\mapsto C$, we obtain $\structA\less g(\structA)$.

Then, since $\opt{\sum_{g\in B^A}\omega(g)g(\structA),\structC}\geq\sum_{g\in B^A}\omega(g)\opt{g(\structA),\structC}$, we get $\structA\less\omega(\structA)$. 

\item From~(2), $\structA\less\omega(\structA)$. From the definition of IFH
  (Definition~\ref{def:ifh}), we have $\omega(\structA)\vecpred\structA$.
    By~(1), $\omega(\structA)\less\structA$. Together, 
    $\structA\equiv\omega(\structA)$.

    By~(2) again, $\structA\less g(\structA)$ for every $g\in\supp{\omega}$.
    In particular, $\opt{\structA,\structC}\leq\opt{g(\structA),\structC}$ for
    every $\structC$.
    Assume for contradiction that
  $\opt{g(\structA),\structC}>\opt{\structA,\structC}$ for at least one 
    $g\in\supp{\omega}$ and some $\structC$. Then,
\begin{multline*}
  \opt{\omega(\structA),\structC}=\opt{\sum_{g\in A^A}\omega(g)g(\structA),\structC}\geq
  \sum_{g\in A^A}\omega(g)\opt{g(\structA),\structC}>
  \\
  \sum_{g\in A^A}\omega(g)\opt{\structA,\structC}
  = \opt{\structA,\structC},
\end{multline*}
which contradicts the already established $\structA\equiv\omega(\structA)$. Hence $\structA\equiv g(\structA)$.
\end{enumerate}
\end{proof}

\begin{proposition}\label{prop:core}
Let $\structA$ be a core valued $\sigma$-structure and $\omega$ be an IFH from $\structA$ to itself. 
Then, for every $g\in\supp{\omega}$ and every $(f,\tuple{x})\in\tup{\structA}$, we have
  $f^{\structA}(\tuple{x})=f^{\structA}(g(\tuple{x}))$.
\end{proposition}

\begin{proof}
  By contradiction, take $(f, \tuple{x})\in \tup{\structA}$ and $g\in \supp{\omega}$ such that 
  $f^{\structA}(\tuple{x})\neq f^{\structA}(g(\tuple{x}))$
  with minimal $c:=f^{\structA}(\tuple{x})$.  
  By minimality, for any
  $h\in\supp{\omega}$, we have
  $f^{\structA}(\tuple{y})=f^{\structA}(h(\tuple{y}))$ if
  $f^{\structA}(\tuple{y})<c$. 
Since $h$ is a bijection ($\structA$ is a core), and by counting, we obtain 
$f^{\structA}(\tuple{y})\geq c$ if and only if
  $f^{\structA}(h(\tuple{y}))\geq c$.  If $c=\infty$, then  $f^{\structA}(g(\tuple{x}))$ needs to be $>\infty$; a contradiction. 
  
  If $c<\infty$, we have 
  $f^{\structA}(\tuple{x})<f^{\structA}(g(\tuple{x}))$. By counting, this
  implies that there is $\tuple{w}$ such that $f^{\structA}(g(\tuple{w}))=c$
  and $f^{\structA}(\tuple{w})>c$. 
  Then   $\sum_{h\in \supp{\omega};\, h(\tuple{y})=g(\tuple{w})} \omega(h) f^{\structA}(\tuple{y})>c=f^{\structA}(g(\tuple{w}))$
  since $f^{\structA}(\tuple{y})\geq c$ and at least for one $\tuple{y}$
  we have $f^{\structA}(\tuple{y})>c$. This contradicts $\omega$ being an IFH.
\end{proof}

\begin{corollary}\label{cor:core}
  Let $\structA$ and $\structB$ be two core valued $\sigma$-structures. Assume
  that there is an IFH $\omega$ from $\structA$ to $\structB$ and an IFH
  $\omega'$ from $\structB$ to $\structA$. Then, $\structA$ and $\structB$ are
  isomorphic.
\end{corollary}

\begin{proof}
  By composition, $\omega'\circ\omega$ is an IFH from $\structA$ to itself. Since
  $\structA$ is a core, the support of $\omega'\circ\omega$ contains only
  bijections. Similarly, since $\structB$ is a core
  the support of $\omega\circ\omega'$ contains only bijections. These two facts
  imply that $|A|=|B|$ and the supports of $\omega$ and $\omega'$ also contain
  only bijections.
  Assume for contradiction that there is $g\in\supp{\omega}$ such that for some
  $(f,\tuple{x})\in\tup{\structA}$ we have $f^{\structA}(\tuple{x})\neq f^{\structB}(g(\tuple{x}))$.

  Case 1: $f^{\structA}(\tuple{x})<f^{\structB}(g(\tuple{x}))$. Since $\omega'$ is an IFH, there must
  exist $g'\in\supp{\omega'}$ such that for some $\tuple{y}$ we have
  $g'(\tuple{y})=\tuple{x}$ and $f^{\structB}(\tuple{y})\leq
  f^{\structA}(\tuple{x})$. But then, $g\circ g'\in\supp{\omega\circ\omega'}$,
  which contradicts Proposition~\ref{prop:core}.

  Case 2: $f^{\structA}(\tuple{x})>f^{\structB}(g(\tuple{x}))$. In this case, we
  consider any $g'\in\supp{\omega'}$ and $(f,g(\tuple{x}))$. As $g'\circ g \in \supp{\omega'\circ\omega}$, Proposition~\ref{prop:core}
  guarantees that $f^{\structB}(g(\tuple{x}))< f^{\structA}(g'(g(\tuple{x})))$. We can then proceed as in Case~1, interchanging the roles of $\structA$ and $\structB$.
\end{proof}

We will need the following variant of Farkas' Lemma, known as Motzkin's transposition theorem. \cite[Lemma~2.8]{tz16:jacm} shows how it can be derived from Farkas' Lemma~\cite[Corollary~7.1k]{Schrijver86:ILP}.

\begin{lemma}
\label{lem:motzkin}
For any $A \in \mathbb{Q}^{m \times n}$ and $B \in \mathbb{Q}^{p \times n}$ exactly one of the following holds:
\begin{itemize}
\item $A\overline{y} > 0, B\overline{y} \geq 0$ for some $\overline{y} \in \qplus^n$, or
\item $A^T\overline{z_1} + B^T\overline{z_2} \leq 0$ for some $0 \neq \overline{z_1} \in \qplus^m$, $\overline{z_2} \in \qplus^p$.
\end{itemize}
\end{lemma}

\begin{proposition}\label{prop:new}
  Let $\structA$ and $\structB$ be two valued $\sigma$-structures and let
  $\G\subseteq B^A$. Then, either there is an IFH $\omega$ from $\structA$ to
  $\structB$ such that $\G\cap\supp{\omega}\neq\emptyset$, or there is a finite-valued
  $\sigma$-structure $\structC$ with $C=B$ such that
  $\costb{\structA\mapsto\structC}{g}>\costb{\structB\mapsto\structC}{\id}$ for
  every $g\in\G$. Moreover, deciding whether such a $\omega$ exists and if so
  then computing it, as well as computing $\structC$ if no such $\omega$ exists,
  can be done effectively  from $\structA$, $\structB$ and $\G$. 
\end{proposition}

\begin{proof}
Let $\M\subseteq B^A$ be the set of all mappings $g$ such that for every $(f,\tuple{x})\in \tup{\structA}$, 
$f^{\structB}(g(\tuple{x}))=\infty$ whenever $f^{\structA}(\tuple{x})=\infty$.    

  Case 1: Assume $\G\cap \M=\emptyset$, i.e., for every $g\in\G$ there is some
  $(f,\tuple{x})\in\tup{\structA}$ satisfying
  $f^{\structA}(\tuple{x})=\infty$ and $f^{\structB}(g(\tuple{x}))<\infty$.
  Define $\structC$ by setting $f^{\structC}(\tuple{x})=0$ if
  $f^{\structB}(\tuple{x})=\infty$ and $f^{\structC}(\tuple{x})=1$ otherwise.
  By the definition of $\structC$, we have
  $\costb{\structB\mapsto\structC}{\id}<\infty$. By the definition of $\structC$
  again and the assumption, $\costb{\structA\mapsto\structC}{g}=\infty$ for
  every $g\in\G$. Hence the desired inequality holds.

Case 2: Assume there is $g\in \G\cap \M$ such that for every $(f,\tuple{x})\in \tup{\structA}$, 
$f^{\structB}(g(\tuple{x}))=\infty$ whenever $0<f^{\structA}(\tuple{x})<\infty$. 
Define $\omega$ by $\omega(g)=1$; it is an
 IFH from $\structA$ to $\structB$, since for every
 $(f,\tuple{y})\in\tup{\structB}$, 
  $\sum_{\tuple{x}:g(\tuple{x})=\tuple{y}}f^{\structA}(\tuple{x})\leq
  f^{\structB}(\tuple{y})$. 

  Case 3: Assume that neither Case~1 nor Case~2 applies. We will use
  Lemma~\ref{lem:motzkin}. 
  We set $n:=|\M|+1$ and $m:=1$. Thus the matrix
  $A$ is just a row vector, where we have a coordinate for every map $g\in \M$
  and one more (say, the last), coordinate. We define $A(g)=1$ if $g\in \M\cap\G$ and $A(g)=0$ otherwise. 
  Note that since Case~1 does not apply, then $\M\cap\G\neq\emptyset$. 
  We set $A$ to be $0$ in the last coordinate.
  Condition $A\overline{y}>0$ means that $\overline{y}(g)>0$ for some
  $g\in\M\cap \G$. 
  Let $\tup{\structB}_{<\infty}:=\{(f,\tuple{y})\in \tup{\structB}: f^{\structB}(\tuple{y})<\infty\}$. 
  We set $p:=|\tup{\structB}_{<\infty}|+1$ and define the matrix $B$ as follows.
  We have a row in $B$ for every $(f,\tuple{y})\in\tup{\structB}_{<\infty}$.
  The entry of $B$ with row $(f,\tuple{y})$ and column $g\in \M$ is
  $-\sum_{\tuple{x}:g(\tuple{x})=\tuple{y}}f^{\structA}(\tuple{x})$. In the
  last column of $B$ we put $f^{\structB}(\tuple{y})$. Finally, we add one more
  row to $B$ of the form $(1,\ldots,1,-1)$. Note that $B$ is actually a rational matrix. 
  
  We now distinguish the two cases of Lemma~\ref{lem:motzkin}. 

  Case 3a: In the first case of Lemma~\ref{lem:motzkin}, we will construct an IFH
  $\omega$ from $\structA$ to $\structB$. Let $c$ be the value of the last
  coordinate of $\overline{y}$.

  If $c=1$, we define $\omega$ as follows: $\omega(g)$ is the $g$-th coordinate
  of $\overline{y}$. Using all, except for the last one, rows of
  $B\overline{y}$, we get
  $\omega(\structA)\vecpred\structB$. Using $A\overline{y}>0$,
  $\G\cap\supp{\omega}\neq\emptyset$. Finally, using the last row of
  $B\overline{y}$, we get that $d:=\sum_{g\in B^A}\omega(g)\geq c=1$. If $d\neq
  1$, to make $\omega$ a probability distribution, we scale it by $1/d$.
  
  If $c>0$, we can scale $\overline{y}$ by $1/c$ and are in the $c=1$ case.

  If $c=0$, let $g\in \M\cap \G$ be a coordinate in
  $\overline{y}$ with $\overline{y}(g)>0$. 
  Since Case 2 does not apply,  there is 
  $(f,\tuple{x})\in\tup{\structA}$ such that $0<f^{\structA}(\tuple{x})<\infty$ and $f^{\structB}(g(\tuple{x}))<\infty$. 
  We claim that
  in $B\overline{y}$ the entry corresponding to $(f, g(\tuple{x}))\in\tup{\structB}_{<\infty}$ is negative,
  which contradicts $B\overline{y}\geq 0$. Indeed, by the definition of $B$, 
  the entry is equal to  $\sum_{h\in \M}\overline{y}(h)(-\sum_{\tuple{y}:h(\tuple{y})=g(\tuple{x})}f^{\structA}(\tuple{y}))$. 
  This is negative as every term in the sum is $\leq 0$ and there is one term $< 0$ (take $h=g$ and $\tuple{y}=\tuple{x}$). 
  
  Case 3b: In the second case of Lemma~\ref{lem:motzkin}, we construct a valued
  $\sigma$-structure $\structC$ as required. Let $c$ be the last entry of
  $\overline{z_2}$ given by Lemma~\ref{lem:motzkin}. Define $\structC$ by
  $f^{\structC}(\tuple{x})$ to be $\overline{z_2}((f,\tuple{x}))+\varepsilon$ for every $(f,\tuple{x})\in\tup{\structB}_{<\infty}$, where $\varepsilon$ is a positive number to be defined later. 
  For  $(f,\tuple{x})\in\tup{\structB}\setminus \tup{\structB}_{<\infty}$ we set $f^{\structC}(\tuple{x})=0$. 
  The last row of the inequality
  $A^T\overline{z_1}+B^T\overline{z_2}\leq 0$ implies that
 $\costb{\structB\mapsto\structC}{\id} = \big(\sum_{(f,\tuple{x})\in \tup{\structB}_{<\infty}} f^{\structB}(\tuple{x}) \overline{z_2}((f,\tuple{x}))\big) + \varepsilon N \leq c + \varepsilon N$, 
  where $N:= \sum_{(f,\tuple{x})\in \tup{\structB}_{<\infty}} f^{\structB}(\tuple{x})$. 
  Since  $\costb{\structB\mapsto\structC}{\id} < \infty$, the required inequality holds for every  $g\in \G\setminus \M$, as 
  we have $\costb{\structA\mapsto\structC}{g}=\infty$ (here we use the fact that $f^{\structC}(\tuple{x})>0$ whenever $f^{\structB}(\tuple{x})<\infty$). 
   On the other hand, for every
  $g\in \G\cap \M$ we have $\costb{\structA\mapsto\structC}{g}\geq \sum_{(f,\tuple{x})\in \tup{\structB}_{<\infty}} \overline{z_2}((f,\tuple{x})) \sum_{\tuple{y}:g(\tuple{y})=\tuple{x}}f^{\structA}(\tuple{y})$. 
  Using the row corresponding to $g$ in the inequality $A^T\overline{z_1}+B^T\overline{z_2}\leq 0$, we obtain $\costb{\structA\mapsto\structC}{g}\geq \overline{z_1}+c$. 
  We can choose $\varepsilon$ so $\overline{z_1}>\varepsilon N$ and obtain the required inequality. 
  
  Computability of $\omega$ and $\structC$ follows directly from the previous arguments and from solving the corresponding system of linear inequalities (from Lemma~\ref{lem:motzkin}) whenever necessary. 
  \end{proof}

We now show how the previous propositions imply the results from
Section~\ref{sec:equiv}.

Propositions~\ref{prop:easy} (1) and (2) imply the ``if direction'' of Proposition~\ref{prop:charfrac}. 
For the ``only if direction'', we can use Proposition~\ref{prop:new} with $\G=B^A$. 
We obtain that if there is no IFH from $\structA$ to $\structB$, then there is  
a structure $\structC$ with
$\opt{\structA,\structC}>\opt{\structB,\structC}$, i.e., $\structA\less\structB$ does not hold.

Proposition~\ref{prop:iso} follows from Corollary~\ref{cor:core} and Proposition~\ref{prop:charfrac}. 

The existence of a core for every valued structure $\structA$ in Proposition~\ref{prop:exist-core} follows from the fact that 
either $\structA$ is a core or, by Proposition~\ref{prop:easy} (3), $\structA\equiv g(\structA)$, where $g:A\to A$ is a non-surjective mapping.  
We can iterate this argument until we find a core of $\structA$, whose universe has size $\leq |A|$. 
Uniqueness of cores follows from Corollary~\ref{cor:core}. 
The computability of the core can be obtained by invoking Proposition~\ref{prop:new} with $\structA=\structB$ and $\G$ being the set of
non-surjective functions on $A$. In this case, Proposition~\ref{prop:new} gives us that if $\structA$ is not a core, then we can compute 
a non-surjective IFH $\omega$ from $\structA$ to itself, and in particular, we can compute a non-surjective $g:A\to A$ such that $\structA\equiv g(\structA)$.
Hence, we can turn the previous existence argument into an algorithm for computing the core.

For the ``only if direction'' of Proposition~\ref{prop:core-char}, we again use Proposition~\ref{prop:new} with $\structA=\structB$ and $\G$ being the set of
non-surjective functions on $A$. We obtain that if $\structA$ is a core, then there exists the required mapping $c$ (which is the finite-valued structure $\structC$) 
and it can be effectively computed. 
For the ``if direction'', and towards a contradiction, suppose $\structA$ is not a core and assume there exists a mapping $c$ as in the statement of the proposition. 
Take a non-surjective IFH $\omega$ from $\structA$ to itself and a non-surjective $g\in \supp{\omega}$. 
Let $\structC$ be the finite-valued structure with $C=A$ and $f^{\structC}(\tuple{x})=c(f,\tuple{x})$. 
By Proposition~\ref{prop:easy} (3), $\structA\equiv g(\structA)$. Let $h$ be a minimum-cost mapping from $g(\structA)$ to $\structC$. 
Since $\costb{g(\structA)\mapsto\structC}{h}=\costb{\structA\mapsto\structC}{h\circ g}$ and $\opt{g(\structA), \structC}=\opt{\structA, \structC}$, 
we obtain that $h\circ g$ is a minimum-cost mapping from $\structA$ to $\structC$. 
However, by the properties of $c$ (and hence $\structC$) and the fact that $h\circ g:A\to A$ is non-surjective, 
we obtain that $\costb{\structA\mapsto\structC}{h\circ g}>\costb{\structA\mapsto\structC}{\id}$. 
This is a contradiction. 

\section{Proof of Example~\ref{ex:btw}}
\label{app:example}

We first recall the construction from Example~\ref{ex:btw}. Consider the signature $\sigma=\{f,\mu\}$, where $f$ and $\mu$ 
are binary and unary function symbols, respectively. 
For $n\geq 1$, let $\structA_n$ be the valued $\sigma$-structure with universe $A_n=\{1,\dots,n\}\times \{1,\dots,n\}$ such that 
(i) $f^{\structA_n}((i,j),(i',j'))=\infty$ if $i\leq i'$, $j\leq j'$, and $(i'-i)+(j'-j)=1$; otherwise $f^{\structA_n}((i,j),(i',j'))=0$, and (ii) $\mu^{\structA_n}((i,j))=1$, for all $(i,j)\in A_n$. 
Also, for $n\geq 1$, let $\structA'_n$ be the valued $\sigma$-structure with universe $A'_n=\{1,\dots,2n-1\}$ such that 
(i) $f^{\structA_n'}(i,j)=\infty$ if $j=i+1$; otherwise $f^{\structA_n'}(i,j)=0$, 
and (ii) $\mu^{\structA_n'}(i)=i$, for $1\leq i\leq n$, and $\mu^{\structA_n'}(i)=2n-i$, for $n+1\leq i\leq 2n-1$. 

We prove the following.

\begin{proposition*}
For every $n \geq 1$, $\structA'_n$ is the core of $\structA_n$.
\end{proposition*}

\begin{proof}
Let $n \geq 1$. First, $\pos{\structA'_n}$ is a relational core, so $\structA'_n$ is a core. We define $g:A_n \mapsto A'_n$ as $g((i,j)) \defeq i+j-1$, then for all $k \in A'_n$ we have
\[
|g^{-1}(k)| = |\{(i,j) \in A_n \mid i+j-1 = k\}| = 
\begin{cases}
k \quad &\text{if } 1 \leq k \leq n\\
2n-k \quad &\text{if } n+1 \leq k \leq 2n-1\\
\end{cases}
\]
and hence for every $k \in A'_n$ we have $\mu^{\structA_n}(g^{-1}(k)) =
  |g^{-1}(k)| \leq \mu^{\structA'_n}(k)$. Furthermore,
  $f^{\structA_n}((i,j),(i',j')) = \infty$ implies $g(i',j') - g(i,j) = (i'-i) +
  (j'-j) = 1$ and in turn $f^{\structA'_n}(g(i,j),g(i',j')) = \infty$. In
  particular, for all $k_1,k_2 \in A'_n$ it holds that
  $f^{\structA_n}(g^{-1}((k_1,k_2))) \leq f^{\structA'_n}((k_1,k_2))$, which
  completes the proof that the distribution $\omega(g) \defeq 1$ is an IFH from $\structA_n$ to $\structA'_n$. 

We now turn to the more delicate task of constructing an 
IFH $\omega'$ from $\structA'_n$ to $\structA_n$. Note that for any (hypothetical) mapping $g' \in \supp{\omega'}$, if $f^{\structA'_n}(\tuple{x}) = \infty$ then $f^{\structA_n}(g'(\tuple{x})) = \infty$. This implies that for any such mapping the sequence of vertices $g'(1),\ldots,g'(2n-1)$ is a directed path from $(1,1)$ to $(n,n)$ in $\pos{\structA_n}$ and if $g'(m) = (i,j)$ then $i+j-1 = m$ (i.e. $(i,j)$ lies on the $m$th diagonal). Additionally, the definitions of $\mu^{\structA_n}$ and $\mu^{\structA_n'}$ imply that $\omega'$ is an IFH if and only if the marginal for each vertex on the $m$th diagonal is $1/m$ if $m \leq n$ and $1/(2n-m)$ if $n \leq m \leq 2n-1$. Formally, for every $i,j$,
\begin{align}
\label{eq:exB}
\sum_{g': g'(i+j-1) = (i,j)} \omega'(g') =
\begin{cases}
1/(i+j-1) &\text{ if } i+j-1 \leq n\\
1/(2n-i-j+1) &\text{ otherwise }
\end{cases}
\end{align}
Now, let $B_n \defeq \{(i,j) \in A_n : i+j-1 \leq n\}$ and $\structB_n$ be the restriction of $\structA_n$ to $B_n$. Similarly, let $B_n' \defeq \{1,\ldots,n\}$ and $\structB_n'$ be the restriction of $\structA_n'$ to $B_n'$. Note that it suffices to find an IFH $\omega'_B$ from $\structB_n'$ to $\structB_n$ satisfying Equation~\ref{eq:exB} (the ``otherwise'' part of the equation is void for $\structB_n, \structB'_n$): given any such $\omega'_B$ we can define $\omega'(g') \defeq \omega'_B(g'|_{B_n'})$ if for all $m \leq n$, $g'(m) = (i,j)$ implies that $g'(2n-m) = (n-j,n-i)$, and $0$ otherwise (i.e. every mapping $g'$ in the support of $\omega'$ is a mirror image of itself via the diagonal).

Recall that, in order for $\omega'$ to be an IFH, any $g'_B \in \supp{\omega'_B}$ must define a path in $\pos{\structB_n}$ of length $n$ from $(1,1)$ to a vertex on the last diagonal (i.e. a vertex $(i,j)$ such that $i+j-1 = n$). Each vertex $(i,j)$ in $B_n$, except for those on the last diagonal, has two outgoing edges: to $(i+1,j)$ and to $(i,j+1)$. We choose the first edge with probability $i/(i+j)$ and the second with probability $j/(i+j)$. The probability $\omega'_B(g'_B)$  is the product of the probabilities of the edges in the path defined by $g'_B$.

It remains to prove that Equation~\ref{eq:exB} holds for $\omega'_B$. Note that the LHS of Equation~\ref{eq:exB} is the probability that $(i,j)$ is on a path chosen according to distribution $\omega'_B$, i.e., the probability that the initial part of the path $g'_B(1),\ldots,g'_B(i+j-1)$ ends in the vertex $(i,j)$. We now prove that Equation~\ref{eq:exB} holds by induction on the diagonal of $\structB_n$ (i.e. for increasing values of $i+j-1$). Choose $(i,j)$. If $i=j=1$, the LHS of Equation~\ref{eq:exB} is equal to $1$ and the claim holds. Otherwise, we have two cases. If $i \neq 1 \neq j$, then the probability that $(i,j)$ is on the path is equal to the probability that $(i-1,j)$ is and the edge $((i-1,j),(i,j))$ was chosen, plus the probability that $(i,j-1)$ is and the edge $((i,j-1),(i,j))$ was chosen. By induction, this probability is equal to $1/(i+j-2) \times (j-1)/(i+j-1) + 1/(i+j-2) \times (i-1)/(i+j-1)$, which is equal to $1/(i+j-1)$ as required. In the case where $i=1$ or $j=1$, one of the summands is $0$, but the required equality still holds. This concludes the proof that $\omega'$ is an IFH from $\structA'_n$ to $\structA_n$, and hence $\structA'_n$ is the core of $\structA_n$, as claimed.

As a final remark, we note that if we alter the definitions of $\structA_n$ and $\structA'_n$ so that  the edges of $\structA'_n$ have weight $1$ (instead of $\infty$) and the weight of each edge $e$ of $\structA'_n$ is the marginal of $\omega'$ on $e$, then the distributions $\omega$ and $\omega'$ are still IFHs between $\structA_n$ and $\structA'_n$. Since $\structA'_n$ is still a core, this new construction shows that bounded treewidth modulo (valued) equivalence is a strictly more general property than bounded treewidth \emph{even for finite-valued structures}.
\end{proof}

\section{Proof of Proposition~\ref{prop:lessfrac}}
\label{app:lessfrac}

We prove the following.

\begin{proposition*}
Let $\structA, \structB$ be valued $\sigma$-structures and $k\geq 1$. If there
  exists an IFH from $\structA$ to $\structB$, then $\structA \lessfrac \structB$.
\end{proposition*}

\begin{proof}
Let $\structC$ be an arbitrary valued $\sigma$-structure, $\omega$ be an 
IFH from $\structA$ to $\structB$ and $\lambda$ be a
solution to SA$_k(\structB,\structC)$ of minimum cost. 
We can write the cost of $\lambda$ as a sum over all $(f,\tuple{x}) \in \tup{\structB_k}_{>0}$ and $s: \toset{\tuple{x}} \mapsto C_k$, 
as constraint \textcolor{red}{\textbf{(SA3)}} ensures that $\lambda(f,\tuple{x},s)=0$, whenever $f^{\structB_k}(\tuple{x})\times f^{\structC_k}(s(\tuple{x}))=\infty$. 
Then, we have
\begin{align*}
\optfrac{k}{\structB,\structC} &= \sum_{(f,\tuple{x}) \in \tup{\structB_k}_{>0},
s:\toset{\tuple{x}} \mapsto C_k} \lambda(f,\tuple{x},s) f^{\structB_k}(\tuple{x}) f^{\structC_k}(s(\tuple{x}))\\
&=\sum_{(f,\tuple{x}) \in \tup{\structB},
s:\toset{\tuple{x}} \mapsto C} \lambda(f,\tuple{x},s) f^{\structB}(\tuple{x}) f^{\structC}(s(\tuple{x}))\\ &\geq  \sum_{(f,\tuple{x}) \in \tup{\structB}, s:\toset{\tuple{x}} \mapsto C} \left(
\sum_{g \in \supp{\omega}} \omega(g)f^{\structA}(g^{-1}(\tuple{x})) \right) \lambda(f,\tuple{x},s) f^{\structC}(s(\tuple{x}))\\
&= \sum_{g \in \supp{\omega}} \omega(g) \left( \sum_{(f,\tuple{x}) \in \tup{\structB},
s:\toset{\tuple{x}} \mapsto C} \lambda(f,\tuple{x},s) f^{\structA}(g^{-1}(\tuple{x})) f^{\structC}(s(\tuple{x})) \right)\\
&= \sum_{g \in \supp{\omega}} \omega(g) \left( \sum_{(f,\tuple{y}) \in \tup{\structA},
s:\toset{g(\tuple{y})} \mapsto C} \lambda(f,g(\tuple{y}),s) f^{\structA}(\tuple{y}) f^{\structC}(s(g(\tuple{y}))) \right)
\end{align*}
and hence there exists $g\in \supp{\omega}$ such that 
\begin{align}
\optfrac{k}{\structB,\structC} &\geq \sum_{(f,\tuple{y}) \in \tup{\structA},
s:\toset{g(\tuple{y})} \mapsto C}  \lambda(f,g(\tuple{y}),s) f^{\structA}(\tuple{y}) f^{\structC}(s(g(\tuple{y})))\nonumber\\
& = \sum_{(f,\tuple{y}) \in \tup{\structA_k}_{>0},
s:\toset{g(\tuple{y})} \mapsto C_k}  \lambda(f,g(\tuple{y}),s) f^{\structA_k}(\tuple{y}) f^{\structC_k}(s(g(\tuple{y}))) \label{eq:supp}
\end{align}

Since $g\in \supp{\omega}$, we have that $g$ is a homomorphism from $\pos{\structA}$ to $\pos{\structB}$ 
(see remark at the end of Section~\ref{sec:inv-frac}). 
It follows that $(f,g(\tuple{y}))\in \tup{\structB_k}_{>0}$, for every $(f,\tuple{y})\in \tup{\structA_k}_{>0}$. Hence, 
for any $(f,\tuple{y}) \in \tup{\structA_k}_{>0}$ and $r: \toset{\tuple{y}} \mapsto C_k$, we can define
\[
\lambda'(f,\tuple{y},r)=
\begin{cases}
\lambda(f,g(\tuple{y}),s) &\quad \text{if there exists $s:\toset{g(\tuple{y})}\mapsto C_k$ such that $s \circ g = r$}\\
0 &\quad \text{otherwise} 
\end{cases}
\]
Equation~\eqref{eq:supp} then becomes 
\[\optfrac{k}{\structB,\structC} \geq \sum_{(f,\tuple{y}) \in \tup{\structA_k}_{>0},
r:\toset{\tuple{y}} \mapsto C_k} \lambda'(f,\tuple{y},r) f^{\structA_k}(\tuple{y}) f^{\structC_k}(r(\tuple{y}))\]
All that remains to do is to show that $\lambda'$ is a feasible solution to
SA$_k(\structA,\structC)$. The fact that the condition
\textcolor{red}{\textbf{(SA4)}} is satisfied is immediate. Also, it follows
from $\omega(g) > 0$ that $f^{\structA_k}(\tuple{y}) = \infty$ implies
$f^{\structB_k}(g(\tuple{y})) = \infty$ and $f^{\structA_k}(\tuple{y})>0$ implies
$f^{\structB_k}(g(\tuple{y})) >0$; thus $f^{\structA_k}(\tuple{y}) \times
f^{\structC_k}(r(\tuple{y})) = \infty$ implies that
$f^{\structB_k}(g(\tuple{y})) \times f^{\structC_k}(r(\tuple{y})) = \infty$ and
the condition \textcolor{red}{\textbf{(SA3)}} is satisfied for all $r$
that can be written as $s \circ g$ for some $s$, as $\lambda$ satisfies \textcolor{red}{\textbf{(SA3)}}. 
The definition of $\lambda'$
ensures that it also holds for all other mappings $r$. For the condition
\textcolor{red}{\textbf{(SA2)}}, observe that for all $(f,\tuple{y}) \in \tup{\structA_k}_{>0}$ we have
\[\sum_{r: \toset{\tuple{y}} \mapsto C_k}\lambda'(f,\tuple{y},r) = \sum_{s:
\toset{g(\tuple{\tuple{y}})} \mapsto C_k}\lambda'(f,\tuple{y},s \circ g) =
\sum_{s: \toset{g(\tuple{y})} \mapsto C_k}\lambda(f,g(\tuple{y}),s) = 1\]
That only leaves the condition
\textcolor{red}{\textbf{(SA1)}}. Let $(f,\tuple{x}), (p,\tuple{y})\in \tup{\structA_k}_{>0}$ such that 
$\toset{\tuple{x}} \subseteq \toset{\tuple{y}}$ and $|\toset{\tuple{x}}|\leq k$, and let $t:
\toset{\tuple{x}} \mapsto C_k$ be any mapping. 
If there does not exist
a mapping $z: \toset{g(\tuple{x})} \mapsto C_k$ such that $t=z\circ g$, then 
\[\lambda'(f,\tuple{x},t) = 0 = \sum_{r:\toset{\tuple{y}} \mapsto C_k, r|_{\toset{\tuple{x}}} = t}\lambda'(p,\tuple{y},r) \]
and the relevant constraints are satisfied, so let us assume that such a mapping $z$ exists. In this case, we have
\begin{align*}
\sum_{r:\toset{\tuple{y}} \mapsto C_k, r|_{\toset{\tuple{x}}} = t}\lambda'(p,\tuple{y},r) 
&= \sum_{s:\toset{g(\tuple{y})} \mapsto C_k, s\circ g|_{\toset{\tuple{x}}}=t}\lambda'(p,\tuple{y},s \circ g) \\ 
&= \sum_{s:\toset{g(\tuple{y})} \mapsto C_k, s\circ g|_{\toset{\tuple{x}}}=t}\lambda(p,g(\tuple{y}),s) \qquad \text{(by definition of $\lambda'$)}\\
&= \sum_{s:\toset{g(\tuple{y})} \mapsto C_k, s|_{\toset{g(\tuple{x}})}=z}\lambda(p,g(\tuple{y}),s) \\
&= \lambda(f,g(\tuple{x}),z) \qquad \text{(applying \textcolor{red}{\textbf{(SA1)}} to $\lambda$, $(f,g(\tuple{x}))$ and $(p,g(\tuple{y}))$)}\\ 
&= \lambda'(f,\tuple{x},t)
\end{align*}
and again the condition \textcolor{red}{\textbf{(SA1)}} is satisfied. Therefore $\lambda'$ is a feasible solution to SA$_k(\structA,\structC)$, and finally $\structA \lessfrac \structB$.
\end{proof}

\end{document}